\documentclass{article}
\usepackage[utf8]{inputenc}
\usepackage[english]{babel}
\usepackage[margin=1in]{geometry}
\usepackage{url}
\usepackage{amsfonts, amsmath, amsthm, algorithm, algorithmicx, xcolor, algpseudocode, amssymb, hyperref}
\newtheorem{theorem}{Theorem}[section]
\newtheorem{lemma}[theorem]{Lemma}
\newtheorem{claim}[theorem]{Claim}
\newtheorem{definition}[theorem]{Definition}
\newtheorem{remark}[theorem]{Remark}
\newtheorem{problem}{Problem}
\newtheorem{conjecture}{Conjecture}
\newtheorem{fact}{Fact}

\newcommand{\R}{\mathbb{R}}
\newcommand{\Z}{\mathbb{Z}}
\newcommand{\calE}{\mathcal{E}}

\newcommand{\calC}{\mathcal{C}}
\newcommand{\calEHat}{\widehat{\mathcal{E}}}
\newcommand{\calETilde}{\widetilde{\mathcal{E}}}
\newcommand{\calN}{\mathcal{N}}

\newcommand{\calI}{\mathcal{I}}
\newcommand{\eps}{\varepsilon}
\newcommand{\Prob}[1]{\text{Pr}[#1]}
\newcommand{\ProbBig}[1]{\text{Pr}\Big[#1\Big]}
\newcommand{\N}{\mathbb{N}}
\newcommand{\poly}{\text{poly}}
\newcommand{\polylog}{\text{polylog}}
\newcommand{\argmin}{\text{argmin}}
\newcommand{\med}{\text{med}}
\newcommand{\nnz}{\text{nnz}}

\title{Optimal $\ell_1$ Column Subset Selection and a Fast PTAS for Low Rank Approximation}
\author{Arvind V. Mahankali\\CMU\\amahanka@andrew.cmu.edu \and David P. Woodruff\\CMU\\dwoodruf@cs.cmu.edu}
\date{}

\begin{document}
\maketitle
\begin{abstract}
We study the problem of entrywise $\ell_1$ low rank approximation. We 
give the first polynomial time column subset selection-based $\ell_1$ low rank approximation algorithm sampling $\widetilde{O}(k)$ columns and achieving an $\widetilde{O}(k^{1/2})$-approximation for any $k$, improving upon the previous best $\widetilde{O}(k)$-approximation and matching a prior lower bound for column subset selection-based $\ell_1$-low rank approximation which holds for any $\poly(k)$ number of columns. We extend our results to obtain tight upper and lower bounds for column subset selection-based $\ell_p$ low rank approximation for any $1 < p < 2$, closing a long line of work on this problem.

We next give a $(1 + \eps)$-approximation algorithm for entrywise $\ell_p$ low rank approximation of an $n \times d$ matrix, for $1 \leq p < 2$, that is not a column subset selection algorithm. First, we obtain an algorithm which, given a matrix $A \in \R^{n \times d}$, returns a rank-$k$ matrix $\widehat{A}$ in $2^{\poly(k/\eps)} + \poly(nd)$ running time that achieves the following guarantee:
$$\|A - \widehat{A}\|_p \leq (1 + \eps) \cdot OPT + \frac{\eps}{\poly(k)}\|A\|_p$$
where $OPT = \min_{A_k \text{ rank }k} \|A - A_k\|_p$. Using this algorithm, in the same running time we give an algorithm which obtains error at most $(1 + \eps) \cdot OPT$ and outputs a matrix of rank at most $3k$ --- these algorithms significantly improve upon all previous $(1 + \eps)$- and $O(1)$-approximation algorithms for the $\ell_p$ low rank approximation problem, which required at least $n^{\poly(k/\eps)}$ or $n^{\poly(k)}$ running time, and either required strong bit complexity assumptions (our algorithms do not) or had bicriteria rank $3k$. Finally, we show hardness results which nearly match our $2^{\poly(k)} + \poly(nd)$ running time and the above additive error guarantee.
\end{abstract}

\newpage
\tableofcontents
\newpage

\thispagestyle{empty}
\clearpage
\setcounter{page}{1}

\section{Introduction}

Low rank approximation is one of the most fundamental problems in data science and randomized numerical linear algebra. In this problem, one is
given an $n \times d$ matrix $A$ and a rank parameter $k$, and one would like to approximately decompose $A$ as $U \cdot V$, where
$U \in \R^{n \times k}$ and $V \in \R^{k \times d}$ are the low rank factors. This gives a form of compression, since rather than storing the $nd$ parameters
needed to represent $A$, we can store only $(n+d)k$ parameters to store the low rank factors. One can also multiply $U \cdot V$ by a vector in
$O((n+d)k)$ time, rather than the $O(nd)$ time needed for multiplication by $A$. The singular value decomposition (SVD) can be used to find the
best low rank approximation of $A$ with respect to the sum of squares of differences, i.e., the Frobenius norm error measure, but this measure
is often not robust enough in applications, since the need to fit single large outliers in $A$ is exacerbated by the squared error measure, causing
$U$ and $V$ to overfit such outliers and not capture enough of the remaining entries of $A$.

To overcome this, a large body of work has studied other, more robust error measures, with a notable one being the entrywise $\ell_1$-low rank approximation problem: given $A \in \mathbb{R}^{n \times d}$, a rank parameter $k$, and an approximation parameter $\alpha \geq 1$, 
find $U \in \mathbb{R}^{n \times k}$ and $V \in \mathbb{R}^{k \times d}$ so that:
$$\|U \cdot V - A\|_1 \leq \alpha \cdot \min_{U' \in \mathbb{R}^{n \times k}, V' \in \mathbb{R}^{k \times d}} \|U'V'-A\|_1,$$
where for a matrix $B \in \mathbb{R}^{n \times d}$, $\|B\|_1 = \sum_{i=1}^n \sum_{j=1}^d |B_{i,j}|$ is its entrywise $1$-norm. Since we
take the sum of absolute values of differences of entries of $A$ and corresponding entries of $U \cdot V$, this is often considered more robust than
the Frobenius norm error measure, which takes the squared differences. There is a large body of work on applications of $\ell_1$ matrix factorization, as well as practical improvements obtained by optimizing the $\ell_1$ error measure rather than the more well-understood Frobenius norm error measure, in areas such as computer vision and machine learning. For instance, in \cite{kk05} and later in \cite{zlsyO12} it was shown that optimizing the $\ell_1$-based objective above, or a regularized version of it, yields much better performance than an IRLS-based approach or other $\ell_2$-based methods on the structure-from-motion problem. Various other works mentioned below, also motivated by problems in machine learning and signal processing, developed heuristics for $\ell_1$ low rank approximation and similar problems.

The $\ell_1$-low rank approximation problem, for $\alpha = 1$, was shown to be NP-hard in \cite{gv15}, and assuming the Exponential Time Hypothesis,
it requires $2^{\Omega(1/\epsilon)}$ time for $\alpha = 1+\epsilon$. While these results rule out extremely accurate solutions to the $\ell_1$-low
rank approximation problem in polynomial time, they leave open the possibility of larger approximation factors. Such approximation factors are of
considerable interest, and we note that a low rank approximation $U \cdot V$ corresponding to an approximation factor $\alpha \ll \sqrt{nd}$ for 
the $\ell_1$-low rank approximation problem may be much better in applications than a low rank approximation corresponding to an exact solution to the 
Frobenius norm error measure, since the error measures are incomparable\footnote{We note that the exact solution to the Frobenius norm error measure gives
a $\sqrt{nd}$-approximation to $\ell_1$-low rank approximation by relating the entrywise $2$-norm to the entrywise $1$-norm.}. A number of heuristics were proposed for
the $\ell_1$-low rank approximation problem in  \cite{kk03,kk05,kimefficient,kwak08,zlsyO12,bj12,bd13,bdb13,meng2013cyclic,mkp13,mkp14,mkcp16,park2016iteratively}. The first rigorous approximation factors
were proven in \cite{l1_lower_bound_and_sqrtk_subset}, where it was shown that $\alpha = \textrm{poly}(k) \log n$-approximation is achievable in poly$(ndk)$ time\footnote{Note that the problem is symmetric in $n$ and $d$, so if $d < n$, we can instead write this as $\textrm{poly}(k) \log d$.}. Later, in \cite{algorithm3_original} an
alternative poly$(ndk)$ time algorithm with poly$(k \log(nd))$-approximation factor was given which held for every entrywise $\ell_p$ norm, for any constant $p \geq 1$.

In fact, each of the algorithms above is, or can be used to obtain a poly$(dk)$ time {\it column subset selection} algorithm for the $\ell_1$-low rank approximation problem with a
$\textrm{poly}(k \log (nd))$ approximation factor. Column subset selection is a special
type of low rank approximation in which the left factor $U$ corresponds to a subset of columns of $A$ itself. Column subset selection is a widely studied and extremely important special case of low rank approximation (see, e.g., \cite{bmd09,dr10,bdm11} and the references therein); it has a number of advantages - for example, if the columns of $A$ are sparse, then the columns in the left factor $U$ are also sparse. One might argue that the column subset selection problem, also known as feature selection (since the columns of $U$ can be thought of as the important features), is sometimes more important than the low rank approximation problem itself. Due to these advantages, it has been of interest to determine how well column subset selection algorithms can do compared to the optimal low-rank approximation error, in the Frobenius norm \cite{dv06, dr10, bmd09}, spectral norm \cite{dr10, bmd09} and in the $\ell_1$ norm \cite{algorithm3_original, SongWZ19a} and even for more general loss functions \cite{DBLP:conf/nips/SongWZ19}.

Often for column subset selection,
we allow for a {\it bicriteria } approximation, namely, for outputting a low rank matrix of rank $r$ with $r$ a bit larger than $k$. 
Bicriteria approximations are common in the low rank approximation literature
\cite{dv07,FFSS07,cw15focs}, and correspond to the case when $k$ is not known or is not a hard constraint; note that bicriteria approximations still capture the original compression and
fast multiplication motivations of low rank approximation discussed above. Moreover, bicriteria approximations are necessary in the context of column subset
selection for some norms; indeed, for Frobenius norm it is known \cite{dv06} that with exactly $k$ columns the best approximation factor possible is $\Theta(k)$, 
while with $O(k)$ columns an $O(1)$ approximation is possible.

A natural question is what the limits of column subset selection algorithms for $\ell_1$-low rank approximation are. It was shown in \cite{DBLP:conf/nips/SongWZ19} that the $O(k \log (nd))$-approximation can be improved to $O(k \log k)$ with bicriteria rank $r = O(k \log n)$, via a polynomial time algorithm. Moreover, if one is willing to spend $n^{O(k \log k)}$
time, it was shown in \cite{l1_lower_bound_and_sqrtk_subset} how to obtain an approximation factor\footnote{See Theorem C.1, part III, of \cite{l1_lower_bound_and_sqrtk_subset}. Note that part I of that theorem also obtains an $O(\sqrt{k \log k})$-approximation, but it is not a column subset selection algorithm.} of $O(\sqrt{k \log k})$. From a purely combinatorial perspective,
this is close to best possible as \cite{l1_lower_bound_and_sqrtk_subset} shows that there exist matrices for which any subset of at most poly$(k)$ columns provides at best a $k^{1/2-\gamma}$-approximation, for an arbitrarily small constant $\gamma > 0$. We note
that for every $p \geq 2,$ there are tight bounds on the size of the best column subset selection algorithms for entrywise $\ell_p$-low rank approximation known\footnote{For entrywise $\ell_p$-low rank approximation, we seek to find rank-$k$ factors $U$ and $V$ so as to minimize $\sum_{i,j} |(U \cdot V-A)_{i,j}|^p$.}, and there are polynomial time algorithms achieving these (see Theorem 4.1 of \cite{dwzzr19}). For $1 \leq p \leq 2$, however, the best known upper bounds have an approximation factor of roughly $k^{1/p}$, while the known lower bounds are only $k^{1-1/p}$ for $1 < p < 2$ \cite{dwzzr19}. In addition, the lower bounds of \cite{dwzzr19} only hold when $k$ columns are selected, and do not rule out smaller approximation factors through slightly larger bicriteria ranks. While these results impose a limit on the approximation error achievable through column subset selection, the current gap of an $O(k \log k)$ approximation
factor versus a $k^{1/2-\gamma}$ lower bound for column subset selection algorithms for $p = 1$ has remained elusive:
\begin{center}
    {\it {\bf Question 1:} What is the best approximation factor for column subset selection for $\ell_1$-low rank approximation achievable by a polynomial time algorithm? What about for
    $1 < p < 2$?}
\end{center}

While the above approximation factors are highly non-trivial, the lower bound of \cite{l1_lower_bound_and_sqrtk_subset} rules out better than a $k^{1/2-\gamma}$-approximation with column subset selection
methods, for arbitrarily small constant $\gamma > 0$. A natural question is if one can improve these approximation factors to $O(1)$ or even $(1 + \eps)$ in polynomial time. This was the main question underlying 
\cite{ptas_for_lplra}, which, motivated by the fact that $k$ is often small, took a parameterized complexity approach and showed how to obtain a $(1+\epsilon)$-approximation in $n^{\textrm{poly}(k/\epsilon)}$ time under the assumption that the entries of the input matrix $A$
are integers in the range $\{-\textrm{poly}(n), \ldots, \textrm{poly}(n)\}$, or alternatively, outputting a matrix of rank $3k$ instead of $k$ with the same approximation factor guarantee, while removing this assumption on the entries of $A$. This algorithm has several drawbacks though: (1) the $n^{\Omega(k)}$ running time even for constant $\epsilon$ is prohibitive and makes the algorithm super-polynomial time for any $k = \omega(1)$, i.e., if $k$ is larger than any constant, and (2) the $O(\log n)$ 
bit complexity assumption on the entries of $A$ may not be realistic; ideally one should allow poly$(nd)$ bit complexity. This motivates the following central question:
\begin{center}
    {\it {\bf Question 2:} Is it possible to obtain a polynomial time algorithm for $\ell_1$-low rank approximation for $k = \omega(1)$ with a small bicriteria rank $r$ and with a $(1 + \eps)$ approximation factor? }
\end{center}

\subsection{Our Results}
In this work we obtain results for both column subset selection and general low rank approximation.

\subsubsection{$\ell_p$ Column Subset Selection for $1 \leq p < 2$} 
We resolve Question 1 above up to small factors by giving a $\poly(ndk)$ time algorithm for finding a subset of $O(k \log k \log d)$ columns providing an $O(\sqrt{k} \log^{3/2} k)$-approximation for $\ell_1$-low rank approximation. This nearly matches the lower bound of \cite{l1_lower_bound_and_sqrtk_subset} which shows that there exist matrices for which the span of {\it any } subset of poly$(k)$ columns has approximation
factor at least $k^{1/2-\gamma}$ for an arbitrarily small constant $\gamma > 0$. 
See Table \ref{table:results} for a table containing our results and a comparison to prior work on column subset selection-based algorithms for $\ell_1$-low rank approximation.

\begin{table}[h!]
\centering
    \begin{tabular}{| c c c c |}
    \hline
    Paper & Number of Columns & Approximation Factor & Time Complexity \\
    \hline
    \cite{l1_lower_bound_and_sqrtk_subset} & $O(k \log k)$ & $O(k \log k \log d)$ & poly$(ndk)$ \\
    \cite{l1_lower_bound_and_sqrtk_subset} & $O(k \log k)$ & $O(\sqrt{k \log k})$ & $d^{O(k \log k)}$ \\
    \cite{algorithm3_original} and \cite{dwzzr19} & $O(k \log d)$ & $O(k \log d)$ & poly$(ndk)$\\
    \cite{algorithm3_original} and \cite{dwzzr19} & $k$ & $O(k)$ & $d^{O(k)}$ \\ 
    \cite{DBLP:conf/nips/SongWZ19} & $O(k \log d)$ & $O(k \log k)$ & poly$(ndk)$\\
    \hline
    This paper & $O(k \log k \log d)$ & $O(\sqrt{k} (\log k)^{\frac{3}{2}})$ & poly$(ndk)$\\
    \hline
    \end{tabular}\\
    \caption{Summary of results for $\ell_1$ column subset selection. The result in this paper is nearly optimal as it was shown in \cite{l1_lower_bound_and_sqrtk_subset} that there exist matrices for which any subset of columns of size at least $\poly(k)$ spans at best a $k^{1/2-\gamma}$-approximation, for arbitrarily small constant $\gamma > 0$. We note that \cite{dwzzr19} obtains the same results as \cite{algorithm3_original} for $p = 1$, but obtains strict improvements 
    for every $p > 1$. Another related work is \cite{SongWZ19a} which obtains a column subset giving a $(1+\epsilon)$-approximation, but assumes the entries of the input matrix are i.i.d. from a distribution with a certain moment assumption.}
    \label{table:results}
\end{table} 

We next extend our result
to any $1 < p < 2$, giving a $\poly(ndk)$ time column subset selection-based algorithm for $\ell_p$ low rank approximation, that chooses $O(k \log k (\log \log k)^2 \log d)$ columns and 
provides an $O(k^{1/p-1/2} \log^{2/p - 1/2} k (\log \log k)^{2/p - 1})$-approximation. We are not aware of any matching lower bounds for column subset selection for $1 < p < 2$ (the tight lower bounds
in \cite{dwzzr19} are specifically for $p > 2$ and the lower bound in \cite{l1_lower_bound_and_sqrtk_subset} is only for $p = 1$), so we also prove the first lower bounds for $1 < p < 2$, showing that there exist matrices for which
the span of any column subset of size $k \cdot \polylog(k)$ has approximation factor at least $k^{1/p-1/2-\gamma}$ for an arbitrarily
small constant $\gamma > 0$. This shows that our algorithm is nearly optimal for $1 < p < 2$ as well.

Our results complement the results of \cite{dwzzr19} for the case of $p > 2$; together with our results we obtain
optimal bounds on column subset selection, up to small factors, for every $\ell_p$-norm, $p \geq 1$, closing a line
of work on this problem \cite{l1_lower_bound_and_sqrtk_subset,DBLP:conf/nips/SongWZ19,SongWZ19a,algorithm3_original,dwzzr19}. We note that all the column subset selection algorithms in this line of work are bicriteria algorithms, as motivated and defined above, meaning that the actual number of columns returned is $k \cdot \poly(\log(nk))$, which is optimal up to polylogarithmic factors.

We use the above algorithm as a subroutine to achieve two additional results, a polynomial time $\widetilde{O}(k)$-approximation algorithm with bi-criteria rank $\widetilde{O}(k^2)$ (and no dependence on $n$ and $d$, even logarithmic, in either the approximation factor or rank), and a $\poly(k)$-approximation algorithm with rank $k$ (and no dependence on $n$ and $d$, even logarithmic, in the approximation factor) with running time $\poly(nd)$ --- the first of these algorithms is a column subset selection algorithm, while the second is not. These algorithms may be of independent interest --- note that all previously known algorithms for getting a $\poly(k)$-approximation factor independent of $n$ or $d$ require either $d^k$ time or a $\log d$ term in the bi-criteria rank.

We note that there are also works studying $\ell_1$ column subset selection, which obtain guarantees in terms of the error from the optimal \textit{column subset} --- two such works are \cite{bl_18_css_nonnegative, css_distributed_2020}. Our column subset selection results are not comparable to these. \cite{bl_18_css_nonnegative} studies the problem of fitting the columns of a matrix $B$ using a subset of the columns of a (potentially different) matrix $A$, and considers the case where both $A$ and $B$ have nonnegative entries. The following guarantee is obtained by \cite{bl_18_css_nonnegative}: if $A$ has a column subset of size $k$ obtaining error at most $\eps \|B\|_1$, then the algorithm of \cite{bl_18_css_nonnegative} (given a $\delta > 0$) finds a subset of columns of $A$ of size $O(k\log(1/\delta)/\delta^2)$ that fits $B$ with error at most $\sqrt{\delta + \eps}\|B\|_1$. This error could potentially be larger than that of our algorithm which obtains relative error guarantees, if $\|B\|_1$ is much larger than $OPT$ --- in that case, even if there is a column subset obtaining error at most $OPT$, the algorithm of \cite{bl_18_css_nonnegative} could obtain error up to $\sqrt{OPT \cdot \|B\|_1}$, according to this guarantee. Hence, our algorithm could potentially be significantly better in the case where $\|B\|_1 \gg \poly(k)OPT$. \cite{css_distributed_2020} gives a protocol for distributed column subset selection in the $\ell_p$-norm, which obtains an $\widetilde{O}(\sqrt{k} \log(d))$-approximation \textit{relative to the error of the best column subset} when $p = 1$. This could be larger than the error from our algorithm in cases where the error due to the best column subset is significantly larger than $OPT$ --- note that by the lower bound of \cite{l1_lower_bound_and_sqrtk_subset}, the error of the best column subset could be $\Omega(k^{\frac{1}{2} - \alpha} \cdot OPT)$ for an arbitrarily small constant $\alpha$.

\subsubsection{General $\ell_p$ Low Rank Approximation for $1 \leq p < 2$}

We next consider Question 2 
and give a new bicriteria algorithm for entrywise $\ell_1$
low rank approximation which goes beyond column subset selection. We improve upon the previous best bicriteria algorithms achieving $O(1)$ and $(1 + \eps)$-approximation, which require at least $n^{\textrm{poly}(k/\eps)}$ time and are not polynomial time when $k = \omega(1)$. In contrast, our algorithm is polynomial time even for slow-growing functions $k$ of $n$, such as $k = \Theta(\log^c n)$, for an absolute constant $c > 0$. 

As mentioned above, bicriteria algorithms are natural and well-studied in this context. We note though that the main algorithm
of \cite{ptas_for_lplra} outputs a rank-$k$ matrix while ours is bicriteria; however, no $(1 + \eps)$ or $O(1)$ approximations in less than $n^{\textrm{poly}(k)}$ time were known even for bicriteria algorithms. In particular, even the bicriteria algorithm of \cite{ptas_for_lplra} requires $n^{\poly(k/\eps)}$ time.

We give a $(1 + \eps)$-approximation algorithm with $2^{\poly(k/\eps)} + \poly(nd)$ running time and output rank at most $3k$. See Table \ref{table:resultsGeneral} for a table containing our result and a comparison to prior
work for $\ell_1$-low rank approximation, where we list all previous $O(1)$ and $(1 + \eps)$ approximation algorithms. Our algorithm also works for $\ell_p$ low rank approximation for $1 \leq p < 2$.

An interesting aspect of our result is that it shows entrywise $\ell_p$-low rank approximation for $p \in [1, 2)$, with
the above approximation factor and bicriteria rank, is {\it not $W[1]$-hard}, in the language of parameterized
complexity, which a priori may have been the case as all previous algorithms required $n^{\textrm{poly}(k)}$ time.

\begin{table}[h!] 
\centering
    \begin{tabular}{| c c c c c |}
    \hline
    Paper & Bicriteria Rank & Approximation Factor & Time Complexity & Notes \\
    \hline
    \cite{l1_lower_bound_and_sqrtk_subset}, Lemma C.10 & $k$ & $O(1)$ & $n^{\textrm{poly}(k)}$ & BCA \\
    \cite{l1_lower_bound_and_sqrtk_subset}, Theorem C.9 & $3k$ \footnotemark & $O(1)$ & $n^{\textrm{poly}(k)}$ & \\
    \cite{ptas_for_lplra} & $k$ & $1 + \eps$ & $n^{\textrm{poly}(k/\eps)}$ & BCA \\
    \cite{ptas_for_lplra} & $3k$ & $1 + \eps$ & $n^{\textrm{poly}(k/\eps)}$ & \\
    \hline
    This paper & $3k$ & $1 + \eps$ & $2^{\poly(k/\eps)} + \poly(nd)$ & \\
    \hline
    \end{tabular}\\
    \caption{Summary of our results for general $\ell_1$-Low Rank Approximation. We write BCA in the Notes next to each algorithm if it makes bit-complexity assumptions on the entries of the input matrix. We note that the non-bicriteria algorithm of \cite{ptas_for_lplra} actually has an
    $(Mn)^{\textrm{poly}(k)}$ running time, where $M$ is the maximum value of an entry of the input matrix, and thus is not even polynomial time for constant
    $k$ if the entries of the input matrix are expressed with more than $\log n$ bits. In contrast, our algorithm is polynomial in the input description
    length, and thus for example, can handle entries as large as $M = 2^{\textrm{poly}(n)}$. The prior algorithms from \cite{l1_lower_bound_and_sqrtk_subset} and
    \cite{ptas_for_lplra} require time $n^{\textrm{poly}(k)}$, and thus are not polynomial time for any $k = \omega(1)$. In contrast, our algorithms are
    all polynomial time even if $k$ is a slow-growing function of $n$, such as $\Theta(\log^c n)$ for an absolute
    constant $c > 0$. Another work \cite{k18_gd_lp_lra} obtains a $(1 + \eps)$-approximation algorithm via gradient descent. This algorithm is not polynomial time or fixed-parameter tractable in the worst case, since the number of iterations of gradient descent required is $O(\frac{\sigma_k(\widehat{X}_k^*)}{\eps OPT})$, where $\widehat{X}_k^*$ is the optimal rank-$k$ approximation to $A$ in the Frobenius norm, and $\sigma_k(\widehat{X}_k^*)$ is its $k^{th}$ singular value. This convergence rate can be arbitrarily large: for instance, when the input matrix $A \in \R^{n \times n}$ is a diagonal matrix with $2^{\poly(n)}$ for its first $k$ entries, and $1$ for the remaining $n - k$ entries, it is $\frac{2^{\poly(n)}}{n - k}$.}
    \label{table:resultsGeneral}
\end{table}

\footnotetext{There is a slight typo in \cite{l1_lower_bound_and_sqrtk_subset} in Theorem C.9 and its proof, where the bicriteria rank is said to be $2k$ --- the bicriteria rank of that algorithm is actually $3k$, since once the $\poly(n)$-approximation $B$ is subtracted off from the target matrix $A$, a good rank-$2k$ approximation $M$ for $B - A$ is needed to recover the original optimum for $A$, from $B - A$.}

\subsubsection{Hardness for $\ell_p$ Low Rank Approximation with Additive Error - Appendix \ref{appendix:additive_error_hardness}}

As an intermediate step for $\ell_p$ low rank approximation, we obtain an algorithm (Algorithm \ref{algorithm:guessing_eps_approximation} for $p = 1$ and Algorithm \ref{algorithm:guessing_eps_approximation_general_lp} for general $p$) which achieves the following guarantee: given $A \in \R^{n \times d}$ and $k \in \N$, it obtains a matrix $\widehat{A}$ of rank at most $k$ such that
\begin{equation} \label{eq:additive_error_guarantee_lp_lra}
\begin{split}
\|\widehat{A} - A\|_p \leq (1 + \eps) \min_{A_k \text{ rank }k} \|A_k - A\|_p + \frac{\eps}{f}\|A\|_p
\end{split}
\end{equation}
in $f^{\poly(k/\eps)} + \poly(nd)$ time, where $f$ can be any desired number greater than $1$. To our knowledge, there do not exist hardness results for such a guarantee (or for bi-criteria approximations). It is known, due to \cite{ptas_for_lplra}, that achieving an $O(1)$-approximation for $\ell_p$ low rank approximation ($p \in (1, 2)$) requires at least $2^{k^{\Omega(1)}}$ running time assuming the Small Set Expansion Hypothesis (SSEH) \cite{rs_2010_sseh_original} and Exponential-Time Hypothesis (ETH) \cite{ip_2001_eth_original}. We extend the techniques of \cite{ptas_for_lplra} to show that even achieving the guarantee in Equation \ref{eq:additive_error_guarantee_lp_lra} requires $2^{k^{\Omega(1)}}$ time assuming SSEH and ETH, for $p \in (1, 2)$, if $f \geq 2^{\poly(k)}$, even when $\eps = \Theta(1)$. Hence, our algorithm is close to optimal in a sense, since it can achieve that guarantee in $2^{\poly(k)} + \poly(nd)$ time.

We also show that it is optimal in the following related sense: it can achieve a similar guarantee in $2^{\poly(k)} + \poly(nd)$ time for the related problem of constrained $\ell_1$ low rank approximation.

\begin{problem}[Constrained $\ell_1$ Low Rank Approximation]
Given a matrix $A \in \R^{n \times d}$ and a subspace $V \subset \R^n$, find a matrix $\widehat{A}$ of rank at most $k$ minimizing $\|\widehat{A} - A\|_1$, such that the columns of $\widehat{A}$ are in $V$.
\end{problem}
Our algorithm for $p = 1$ (see Algorithm \ref{algorithm:guessing_eps_approximation}) can be modified very slightly so that it achieves the same guarantee for this problem as well --- the modified algorithm can compute, in $2^{\poly(k)} + \poly(nd)$ time, a matrix $\widehat{A}$ of rank at most $k$ such that
$$\|\widehat{A} - A\|_1 \leq O(1) \min_{A_k} \|A_k - A\|_1 + \frac{1}{2^{\poly(k)}}\|A\|_1$$
and the columns of $\widehat{A}$ are contained in $V$ --- here, the minimum on the right-hand side is also taken over $A_k$ with rank at most $k$, such that the columns of $A_k$ are in $V$. Assuming the SSEH and randomized ETH (used for instance in \cite{dhmtw14_randomized_eth}), we show that achieving this guarantee also requires at least $2^{k^{\Omega(1)}}$ time.

\subsection{Our Techniques}

We give an overview of our arguments for $\ell_1$-column subset selection, then for general $\ell_1$-low rank approximation. The arguments for $\ell_p$-column subset selection and $\ell_p$-low rank approximation, for $p \in (1, 2)$, are similar and are included in Section \ref{section:lp_css} and Section \ref{section:fpt_approx} respectively.

\subsubsection{$\ell_1$ Column Subset Selection}

\paragraph{Algorithm and Initial $\widetilde{O}(\sqrt{k}) \log(d)$ Approximation Factor.}

In our algorithm, we uniformly sample a column subset $S^{(0)}$ of size $t = O(k \log k)$ of our input matrix $A$ 
and argue that we can approximately span a constant fraction of remaining columns of $A$ using $S^{(0)}$, where to approximately span
the $i$-th column means to obtain a column vector $v^i$ for which $\|v^i - A_i\|_1 = O(\sqrt{k \log k}) OPT/d$, where
$OPT$ is the cost of the optimal rank-$k$ approximation to $A$. Thus, our total cost to cover a constant fraction of columns
will be $O(\sqrt{k \log k}) OPT$. We then recurse on the remaining fraction of columns. If there are $d_1$
columns in the next recursive call, we argue we approximately span a constant fraction of the remaining columns, where now to
approximately span the $i$-th column means to obtain a column vector $v^i$ for which $\|v^i-A_i\|_1 = O(\sqrt{k \log k}) OPT/d_1$. 
Again, the total cost to cover a constant fraction of remaining columns is $O(\sqrt{k \log k}) OPT$. We then recurse on 
the columns still remaining. Since we approximately span a constant fraction of columns in each recursive step, after $O(\log d)$ recursive
steps we will have spanned all $d$ columns. The total number of columns we will have chosen is $O(k \log k \log d)$
and the overall approximation factor will be $O(\sqrt{k \log k} \log d)$.

The algorithm described above is simple, and reminiscent of column sampling algorithms \cite{algorithm3_original,dwzzr19,DBLP:conf/nips/SongWZ19}
in prior work. 
However, our analysis is completely new and does not involve going through maximum determinant subsets, as in each of these previous column sampling
algorithms. These works argued that if one uniformly samples a set $S$ of $2k$ columns of $V^* \in \mathbb{R}^{k \times n}$, 
where $A_k = U^*V^*$ is the best rank-$k$ approximation to $A$, and considers a random additional 
column $c$, then with probability at least $1/2$, the maximum determinant subset (which is of size $k$) of columns of $V^*_{S \cup \{c\}} \in \mathbb{R}^{k \times (2k+1)}$
does not contain $c$, and consequently, by Cramer's rule, $c$ can be expressed as a linear combination of columns in our sample set $S$ with coefficients of absolute
value at most $1$, and thus by the triangle inequality one pays a cost at most what the subset $S$ pays to approximate the $c$-th column of $A$,
which since $S$ was chosen uniformly at random, is at most $O(k) OPT/d$ with constant probability. 

We do not know how to reduce the approximation factor in the analyses of all of these previous algorithms; intuitively, the difficulty
stems from the fact that the maximum determinant subset may not be the best subset to look at for $\ell_1$; indeed, it could be that for a random $c$,
one needs coefficients of absolute value $1$ to span it using the columns in the set $S$, and it is unclear how the error propagates other than
through the triangle inequality. We thus give the first analysis of the above sampling framework that {\it does not go through maximum determinant subsets}.

Instead, we argue that if one had $V^*$, then one could sample columns using its so-called {\it Lewis weights} \cite{lewis_weights}, creating a sampling and rescaling 
matrix $R$ with $t/2 = O(k \log k)$ columns, 
so that the solution $U = \textrm{argmin}_U \|A_{S \cup c}R -U (V^*)_{S \cup c} R\|_1$ would be an $O(1)$-approximate rank-$k$ left factor
for the submatrix of $A$ indexed by $S \cup c$. By relating $\ell_1$ and $\ell_2$-norms of the rows of $A_{S \cup c}R - U(V^*)_{S \cup c}R$ and using the normal equations for
least squares regression, we get that $U' = A_{S \cup c} R ((V^*)_{S \cup c} R)^{+}$ provides an $O(\sqrt{k \log k})$-approximate rank-$k$ left
factor for the submatrix of $A$ indexed by $S \cup c$. The advantage of $U'$ is that it is in the column span of $A_{S \cup c} R$. Note that this subroutine
should be reminiscent of the algorithm of \cite{l1_lower_bound_and_sqrtk_subset}, which argued one could obtain an $O(\sqrt{k \log k})$-approximate column
subset selection algorithm by enumerating over all subsets of $t/2$ columns of $A$; however \cite{l1_lower_bound_and_sqrtk_subset} actually does this and suffers
$d^{\Omega(k \log k)}$ time.

Instead, we argue as follows. As noted by \cite{l1_lower_bound_and_sqrtk_subset}, since $A_{S \cup c}R$ is an $O(\sqrt{k \log k})$-approximate left factor for $A_{S \cup c}$ with high probability (over the randomness of $R$), in particular, there exists a fixed matrix $R_0$, with one nonzero entry per column, such that $A_{S \cup c}R_0$ is an $O(\sqrt{k \log k})$-approximate left factor for $A_{S \cup c}$ (and in particular, $A_{S \cup c}R_0$ is a column subset of $A_{S \cup c}$). Recall that $A_c$ is a uniformly random column of $A_{S \cup c}$, $A_{S \cup c}$ has $t + 1$ columns, and $A_{S \cup c}R_0$ has $t/2$ columns (where $t/2 = O(k \log k)$ is the number of columns required by \cite{lewis_weights} to perform $\ell_1$ Lewis weight sampling). Hence, we can now simply argue as in \cite{algorithm3_original, DBLP:conf/nips/SongWZ19} that with probability at least $\frac{1}{2}$ over $S$ and $c$, $A_c$ is not in $A_{S \cup c}R_0$.

In summary, we use the facts that (1) column $A_c$ is not chosen by the matrix $R_0$ with probability at least $\frac{1}{2}$ over $S$ and $c$, (2) the overall cost of using the columns of $A_{S \cup c}R_0$ to approximately span all the columns of $A_{S \cup c}$ is $O(\sqrt{k \log k})OPT_{S \cup c}$, where $OPT_{S \cup c}$ is the total cost on $A_{S \cup c}$ under the optimal rank-$k$ approximation, (3) since $S \cup \{c\}$ is itself a uniformly random subset of $[d]$, the value $OPT_{S \cup c}$ is $O(t/d)OPT$ in expectation, and (4) $A_c$ is among $t/2$ random columns of $A_{S \cup c}$ which are not in the chosen subset of columns $A_{S \cup c}R_0$, and thus has at most a $2/t$ fraction of the total cost on $A_{S \cup c}$, in expectation. Combining these statements gives us that with large constant probability, the cost of approximately spanning $A_c$ using our sampled set is $O(\sqrt{k \log k})OPT/d$, completing the argument.

\begin{remark}
In fact, by considering the properties of the Lewis weight sampling matrix $R$, we can even show $A_{S \cup c}R$ is unlikely to contain the column $A_c$. The subtlety with such an approach is that the distribution of the sampling matrix $R$ may itself depend on $A_c$. In short, we can argue that since $c$ is a uniformly random column in $S \cup c$, this means that in expectation (and with constant probability), the Lewis weight of the column $A_c$ is small. Precisely, the probability that any individual column of $R$ selects $A_c$ is $\Theta(\frac{1}{r}) = \Theta(\frac{1}{k \log k})$. Thus, with constant probability, $A_c$ is not sampled by $R$, since $R$ samples $r = O(k \log k)$ columns of $A$ independently with replacement (where the sampling distribution is determined by the Lewis weights of $V^*$). In summary, with constant probability, it is simultaneously true that (i) $A_{S \cup c}R$ covers $A_{S \cup c}$, meaning it incurs an error of $O(\frac{\sqrt{k \log k}}{d})OPT$ on $A_c$ in expectation, since $A_c$ is a uniformly random column of $A_{S \cup c}$, and (ii) it does not contain $A_c$. The argument is the same from here onwards. This is the approach we took in the previous version of this work on arXiv.
\end{remark}

We generalize this approach to obtain optimal column subset selection algorithms for entrywise $\ell_p$ low rank approximation, for every $1 < p < 2$, replacing the $\ell_1$ Lewis
weights with the $\ell_p$ Lewis weights in the above analysis. Note we cannot use earlier sampling distributions, such as $\ell_p$ leverage scores or total sensitivities, as it
is also important in the argument above that one only needs to sample $O(k \log k)$ columns for a rank-$k$ space and these latter sampling distributions would require
a larger $k^{1+c}$ samples for a constant $c > 0$ (see, e.g., \cite{w14} for a survey); this is important not only for the overall number of
sampled columns but also for the approximation factor, since we also relate the $p$-norm to the $2$-norm through this number.

\paragraph{Removing $\log d$ from the Approximation Factor.}

An improved version of the above argument, inspired by \cite{DBLP:conf/nips/SongWZ19}, gives an $O(\sqrt{k}(\log k)^{\frac{3}{2}})$-approximation rather than $O(\sqrt{k \log k} \log d)$ --- to achieve this, we note that in each round, we can condition on the event that the columns being chosen are not among the $\frac{1}{t}$-fraction of columns which have the highest cost, under the optimal $\ell_1$ rank-$k$ approximation. This event occurs with constant probability. Moreover, for each of the other columns which are not in this top $\frac{1}{t}$-fraction of columns (which are indexed by a subset $F \subset [d]$), we can bound the cost by $O(\sqrt{k \log k})\frac{OPT_{-F}}{d}$, where $OPT_{-F}$ denotes the cost under the optimal rank-$k$ approximation, excluding the errors from the top $\frac{1}{t}$-fraction of columns $F$. Finally, we show that over the course of the $O(\log d)$ recursive rounds, a particular column of $A$ can contribute to $OPT_{-F}$ in at most $O(\log k)$ rounds --- that is, it cannot be outside of $F$ for more than $O(\log k)$ rounds without being approximately covered and discarded. A similar technique was used in \cite{DBLP:conf/nips/SongWZ19} to obtain an $O(k \log k)$-approximation factor independent of $\log d$.

\paragraph{Lower Bound for $\ell_p$ Column Subset Selection, $1 \leq p < 2$.}
Our nearly matching lower bound for entrywise $\ell_p$ low rank approximation is a technical generalization of that for $\ell_1$-low rank approximation given in \cite{l1_lower_bound_and_sqrtk_subset} and we defer the details to Appendix \ref{appendix:lp_css_lower_bound_full_proofs}. For $p \in (1, 2)$, we show that a proof strategy similar to that of \cite{l1_lower_bound_and_sqrtk_subset} can be used to show that any column subset selection algorithm that selects at most $O(k(\log k)^{c_1})$ columns (where $c_1$ can be any constant) achieves no better than an $\Omega(k^{\frac{1}{p} - \frac{1}{2} - \alpha})$ approximation factor in the worst case, where $\alpha$ can be an arbitrary number in $(0, \frac{1}{p} - \frac{1}{2})$.

\paragraph{Additional Column Subset Selection Results (Appendix \ref{appendix:poly_k_bicriteria_rank}): Decreasing the Bicriteria Rank.}

Our polynomial time $O(k(\log k)^2)$-approximation algorithm with $O(k^2 (\log k)^2)$ bicriteria rank makes use of the improved approximation factor that is independent of $\log d$, and relies on the following simple observation. Let $U \in \R^{n \times O(k \log k \log d)}$ be the left factor ultimately returned by our main column subset selection algorithm. In each recursive round, if $S$ is the set of $t = O(k \log k)$ columns which are sampled, then for each column $A_i$ which is discarded during that round, $A_i$ can be approximately covered using only the $t = O(k \log k)$ columns belonging to $S$. The implication of this is that there exists $M$ having rank $O(k \log k \log d)$, which provides an $O(\sqrt{k}(\log k)^{\frac{3}{2}})$-approximation for $A$, such that each column of $M$ can be written exactly as a linear combination of $O(k \log k)$ columns of $U$ (more specifically, $O(k \log k)$ columns of $U$ that were obtained in a single round of sampling from $A$). 

Now, as mentioned above, it was shown in \cite{l1_lower_bound_and_sqrtk_subset} that any matrix has a column subset of size $O(k \log k)$ spanning an $O(\sqrt{k\log k})$-approximation --- hence, $M$ has a column subset of size $O(k \log k)$ which spans an $O(\sqrt{k \log k})$-approximation \textit{to $M$}. By the triangle inequality, one can show that since $M$ is an $O(\sqrt{k}(\log k)^{\frac{3}{2}})$-approximation for $A$, this $O(k \log k)$-sized column subset of $M$ spans an $O(k(\log k)^2)$-approximation for $A$. To form our left factor, for each of these columns $A_i$, we could collect the $O(k \log k)$ columns that were sampled in the round when $A_i$ was covered. Since we do not actually know this $\poly(k)$-approximate subset of columns of $M$, we could na\"ively try all of them --- however, rather than checking all column subsets of $M$ of size $O(k \log k)$, which takes time $d^{O(k \log k)}$, it suffices to check all $O(k \log k)$-sized subsets of the $O(\log d)$ rounds of sampling done by our main algorithm, and take the best subset. There are $\binom{O(\log d)}{O(k \log k)} \leq d^{O(1)}$ such subsets, so this algorithm is polynomial time. Finally, to obtain a $\poly(k)$-approximate matrix with rank at most $k$, we combine this algorithm with an algorithm from \cite{l1_lower_bound_and_sqrtk_subset} (analyzed in Theorem C.19 of that work), which takes a bi-criteria solution of rank $r$ as input and reduces the rank to $k$, at the cost of an increase in the error by a factor of $\poly(r)$. Since the bi-criteria rank of our solution is only $\poly(k)$, the approximation error only increases by a factor of $\poly(k)$. \footnotemark\footnotetext{In an earlier version of this work, we instead used an algorithm based on Algorithm 4 of \cite{algorithm3_original} to reduce the rank to at most $k$. This older version of our algorithm had a running time of $2^{O(k \log k)} + \poly(nd)$. Our new algorithm that returns a matrix of rank at most $k$, shown in Algorithm \ref{algorithm:poly_k_rank_exactly_k} and using an algorithm of \cite{l1_lower_bound_and_sqrtk_subset}, is polynomial time.}

\subsubsection{General $\ell_1$ Low Rank Approximation} \label{subsubsection:general_l1_lra_techniques}

We next turn to general low rank approximation, where it is possible to obtain much smaller approximation factors than with column subset selection. A crucial novelty in our algorithm is the use of a randomized rounding technique for solving an integer linear program (ILP) --- to our knowledge, such an approach was not previously considered in the context of $\ell_1$ or $\ell_p$ low rank approximation. Randomized rounding of relaxations has been previously used in other subspace optimization problems, such as by \cite{dtv_11_soda_subspace_approximation} in the related problem of subspace approximation (in the $\ell_{p, 2}$ norm). However, \cite{dtv_11_soda_subspace_approximation} uses it to select random linear combinations of singular vectors of a matrix obtained by solving a convex relaxation, while we use randomized rounding of an LP to choose columns satisfying multiple linear constraints. One appealing aspect of our algorithm is that it does not use polynomial system solvers, which are somewhat impractical --- these have been used for several other NP-hard matrix factorization problems, such as in \cite{weighted_lra, tensor_lra}.

\paragraph{Background: The Algorithm of \cite{ptas_for_lplra} and its Bicriteria Variant.}

As a starting point, we recall the bicriteria variant of the main algorithm of \cite{ptas_for_lplra} (shown in Algorithm \ref{algorithm:previous_eps_approximation_algorithm}), which uses a median-based sketch to obtain a $(1 + \eps)$-approximation for $\ell_1$ low rank approximation. This sketch was previously considered in \cite{biprw_2018_median_estimator_first}, where it was shown that for a $k$-dimensional subspace $V$ of $\R^n$, if $S \in \R^{\poly(k/\eps) \times n}$ is a random matrix with i.i.d. standard Cauchy entries, then with constant probability, for all $v \in V$, $\med(Sv)$ is within a $(1 + \eps)$ factor of $\|v\|_1$, where $\med(v)$ is the median of the absolute values of the entries of the vector $v$. This sketch was then considered in the context of $\ell_1$ low rank approximation by \cite{ptas_for_lplra}, where the following ``one-sided embedding" property of this sketch is shown: for a given matrix $U \in \R^{n \times k}$, and a fixed matrix $A \in \R^{n \times d}$, if $S \in \R^{\poly(k/\eps) \times n}$, then with constant probability $\med(SUV - SA) \geq (1 - \eps) \|UV - A\|_1$ for \textit{all} matrices $V \in \R^{k \times d}$, where if $M$ is a matrix with $d$ columns, then $\med(M) := \sum_{i = 1}^d \med(M_i)$. These properties make this sketch useful in $\ell_1$ low rank approximation, as we now see.

\begin{algorithm}
\caption{$(1 + \eps)$-approximation algorithm from \cite{ptas_for_lplra} that gives bi-criteria rank $3k$. This is adapted from Algorithm 1 of \cite{ptas_for_lplra} and Theorems 10 and 23 of \cite{ptas_for_lplra}.}
\label{algorithm:previous_eps_approximation_algorithm}
\begin{algorithmic}
\Require $A \in \R^{n \times d}$, $k \in \N$, $\eps > 0$ with $n \geq d$
\Ensure $\hat{A} \in \R^{n \times d}$
\Procedure{PreviousOnePlusEpsApproximation}{$A, k, \eps$}
\State {$B \gets $ The rank-$k$ SVD of $A$}
\State {$C \gets A - B$}
\State {$S \gets $ A $\poly(k/\eps) \times n$ matrix of i.i.d. standard Cauchy random variables}
\State {$U, V \gets $ The $n \times 2k$ and $2k \times d$ zero matrices}
\State {$\widetilde{S} \gets S$ with each entry rounded to the nearest integer multiple of $\frac{\eps}{\poly(n)}$}
\State {Guess all possible values of $\widetilde{S}U^*$ with each entry rounded to the nearest integer multiple of $\frac{\eps^3}{\poly(n)}\|C\|_1$.}
\For {each guessed value $M$ of $\widetilde{S}U^*$}
    \State {$V_{guess} \gets \argmin_{V'} \med(MV' - \widetilde{S}C)$ such that $\|V'\|_\infty \leq \poly(k)$}
    \State {$U_{guess} \gets \argmin_{U'} \|U'V_{guess} - C\|_1$}
    \If {$\|U_{guess} V_{guess} - C\|_1 \leq \|UV - C\|_1$}
        \State {$U \gets U_{guess}$, $V \gets V_{guess}$}
    \EndIf
\EndFor \\
\Return {$B + UV$}
\EndProcedure
\end{algorithmic}
\end{algorithm}

Algorithm \ref{algorithm:previous_eps_approximation_algorithm} is the bicriteria variant of the main algorithm of \cite{ptas_for_lplra} (we cannot directly modify the main algorithm since it requires bit complexity assumptions). The main algorithm is given in Algorithm 1 of \cite{ptas_for_lplra}, and Theorem 10 of \cite{ptas_for_lplra} gives the analysis of that algorithm, while Theorem 23 of \cite{ptas_for_lplra} describes how the algorithm should be modified to remove the need for bit complexity assumptions at the cost of a bicriteria rank of $3k$. Hence, our summary of the analysis largely follows Theorem 10 of \cite{ptas_for_lplra}, with minor modifications as given by Theorem 23 of \cite{ptas_for_lplra}.

Briefly, its analysis proceeds as follows. Define $U^* \in \R^{n \times 2k}$ and $V^* \in \R^{2k \times d}$ so that $U^*V^*$ is the optimal rank-$2k$ approximation for $C$. First, $V^*$ is assumed without loss of generality to be a $\poly(k)$-well-conditioned basis, meaning for all row vectors $x \in \R^{2k}$, $\frac{1}{\poly(k)}\|x\|_1 \leq \|x^TV^*\|_1 \leq \poly(k)\|x\|_1$. Then, $U^*$ is assumed to have all of its entries rounded to the nearest integer multiple of $\frac{\eps \|C\|_1}{\poly(n)}$. This is not an issue, because if $\widetilde{U}$ is the rounded version of $U^*$, then
$$\|(\widetilde{U} - U^*)V^*\|_1 \leq \poly(k) \cdot \|\widetilde{U} - U^*\|_1 \leq \poly(k) \cdot O(nk) \cdot \frac{\eps \|C\|_1}{n \cdot \poly(n)} = \frac{\eps}{\poly(n)} \|C\|_1 \leq \eps \cdot OPT$$
where the second inequality is because $\widetilde{U}, U^*$ have $2nk$ entries, and the last inequality is because the rank-$k$ SVD of $A$ gives an $O(n)$-approximation, assuming $n \geq d$. 

In addition, $S$ is also assumed to be discretized, and the discretized version is written as $\widetilde{S}$. This is done as follows. First, because $V^*$ is a $\poly(k)$-well-conditioned basis, each of its entries is at most $\poly(k)$ (note that in Theorem 10 of \cite{ptas_for_lplra}, it is mentioned that we can assume each entry of $V^*$ is at most $\poly(ndk/\eps)$, but this can be decreased further to $\poly(k)$). Hence, we can restrict ourselves to right factors $V'$ for which each of its entries is at most $\poly(k)$. Now, if $V \in \R^{2k \times d}$ has each entry at most $\poly(k)$, then
$$\|U^*V - C\|_1 \leq \|U^*V\|_1 + \|C\|_1 = \sum_{i = 1}^d \|U^*V_i\|_1 + \|C\|_1 \leq \poly(k)d \|U^*\|_1 + \|C\|_1 = \poly(k) \cdot d \cdot \|C\|_1$$
where the first inequality is by the triangle inequality, the second inequality is because $\|V_i\|_\infty \leq \poly(k)$, and the last equality is because $\|U^*\|_1 \leq \poly(k) \|U^*V^* \|_1 = \poly(k)\|C\|_1$ since $V^*$ is a well-conditioned basis.

Now, the reason for rounding the entries of $S$ to the nearest integer multiple of $\frac{\eps}{\poly(n)}$ (and why this does not significantly increase the error) is that, for any $V$ with $\|V\|_\infty \leq \poly(k)$,
\begin{equation} \label{eq:prev_algorithm_additive_error}
\begin{split}
|\med(\widetilde{S}U^*V - \widetilde{S}C) - \med(SU^*V - SC)|
& \leq \sum_{i = 1}^d \|(\widetilde{S}U^*V_i - \widetilde{S}C_i) - (SU^*V_i - SC_i)\|_\infty \\
& = \sum_{i = 1}^d \|(\widetilde{S} - S)(U^*V_i - C_i)\|_\infty \\
& \leq \sum_{i = 1}^d \|\widetilde{S} - S\|_\infty \|U^*V_i - C_i\|_1 \\
& = \|S - \widetilde{S}\|_\infty \|U^*V - C\|_1
\end{split}
\end{equation}
where the first inequality is due to the fact that, if $v_1, v_2 \in \R^n$, then $|\med(v_1 + v_2) - \med(v_1)| \leq \|v_2\|_\infty$, and the second is because, for a matrix $B$ and a vector $v$, $\|Bv\|_\infty \leq \|B\|_\infty \|v\|_1$. 

Since the entries of $S$ are rounded to the nearest multiple of $\frac{\eps}{\poly(n)}$, $\|S - \widetilde{S}\|_\infty \leq \frac{\eps}{\poly(n)}$ and $\|S - \widetilde{S}\|_\infty \|U^*V - C\|_1$ is at most $\frac{\eps}{\poly(n)} \|C\|_1 = \eps \cdot OPT$. Therefore, not much additional error is incurred when minimizing $\med(\widetilde{S}U^*V - \widetilde{S}C)$, subject to the constraint that $\|V\|_\infty \leq \poly(k)$, as opposed to minimizing $\med(SU^*V - SC)$.

In summary, the algorithm works by guessing all possible values of $\widetilde{S}U^*$. By well-known properties of Cauchy matrices, the entries of $S$ are bounded above by $\poly(n)$, and those of $U^*$ can also be bounded above by $\poly(n) \|C\|_1$, meaning there are $\poly(n/\eps)$ choices for each entry of $\widetilde{S}U^*$, and $\widetilde{S}U^*$ has $\poly(k/\eps)$ entries, meaning the overall running time is $n^{\poly(k/\eps)}$.

\paragraph{Our Approach: Achieving FPT Time by Reducing the Number of Guesses Per Entry.}

Our approach, like those of \cite{ptas_for_lplra} and Theorem C.9 of \cite{l1_lower_bound_and_sqrtk_subset}, follows the general strategy of first taking a good initialization $B$, subtracting it from $A$, and finding a good rank $2k$ approximation for the residual $C := B - A$. Now, the running time of Algorithm \ref{algorithm:previous_eps_approximation_algorithm} is dominated by the time it takes to guess $\widetilde{S}U^*$, and there are $\poly(n/\eps)$ guesses per entry of $\widetilde{S}U^*$. We now describe how with our approach, we reduce the number of possibilities per entry to $\poly(k/\eps)$, while still obtaining a $(1 + \eps)$-approximation.

Perhaps the most obvious optimization to make is to use a better initialization --- rather than letting $B$ be the rank-$k$ SVD of $A$, we could run a $\poly(k)$-approximation algorithm on $A$ to obtain $B$, such as our Algorithm \ref{algorithm:poly_k_rank_exactly_k} which gives a $\poly(k)$-approximation with rank at most $k$ in $\poly(nd)$ time. A $\poly(k)\log(d)$-approximation algorithm, such as that of \cite{l1_lower_bound_and_sqrtk_subset}, would also suffice.

Using an initialization algorithm with a better approximation factor reduces the number of guesses per entry of $\widetilde{S}U^*$, but the number of guesses remains $\poly(n/\eps)$, rather than $\poly(k/\eps)$, mainly for the following reasons:
\begin{itemize}
    \item When $U^*$ is being discretized (recall that this is not explicitly done in the algorithm, but the analysis assumes $U^*$ is discretized in order to have a finite number of entries to guess) the entries still need to be rounded to the nearest integer multiple of $\frac{\eps \|C\|_1}{n \cdot \poly(k)}$. This is because $U^*$ has $n$ rows, and therefore, if $\widetilde{U}$ denotes the rounded version of $U^*$, then the additional error from using $\widetilde{U}$ instead of $U^*$ can still only be upper bounded by
    $$\|(\widetilde{U} - U^*)V^*\|_1 \leq \poly(k) \|\widetilde{U} - U^*\|_1 \leq \poly(k) \cdot O(nk) \cdot \|\widetilde{U} - U^*\|_\infty$$
    in the worst case --- meaning a rounding granularity of at least $\frac{1}{n}$ is needed.
    \item When $S$ is being discretized to obtain $\widetilde{S}$, then as mentioned above in Equation \ref{eq:prev_algorithm_additive_error}, the additive error is at most $\|S - \widetilde{S}\|_\infty \|U^*V - C\|_1$, where $V$ is the right factor that Algorithm \ref{algorithm:previous_eps_approximation_algorithm} obtains. Recall that the upper bound for $\|U^*V - C\|_1$ is $\poly(k) \cdot d \cdot \|C\|_1$ for \textit{all} $V$ with no entry larger than $\poly(k)$ (in our exposition of the algorithm of \cite{ptas_for_lplra}, we showed that we just need to consider $V$ with no entry larger than $\poly(k)$, while the original algorithm in \cite{ptas_for_lplra} in fact allowed $V$ to have entries at most $\poly(n/\eps)$ --- this small change can only improve the upper bound on the additive error). Hence, when rounding $S$, a granularity of at least $\frac{1}{d}$ seems to be needed.
\end{itemize}

It is not clear how to circumvent these issues if we round $S$ and $U^*$ separately. Instead, to avoid rounding with granularities of $\frac{\eps}{n}$ or $\frac{\eps}{d}$, we round $SU^*$ itself in our analysis. Specifically, we show the following. If $M$ is $SU^*$, but with each entry rounded to the nearest power of $1 + \frac{1}{\poly(k/\eps)}$, or set to $0$ if it is below $\poly(\eps/k) \cdot OPT$, then we obtain a small additive error by solving the following problem instead:
$$\min_{V'} \med(MV' - SC) \text{ subject to } \|V'\|_1 \leq \poly(k)$$
and then again finding a good left factor $U'$ for $V'$ through linear programming. The number of choices for each entry of $M$ is then $\poly(k/\eps)$, and since $M$ is a $\poly(k/\eps) \times k$ matrix (in fact, a $k \cdot \poly(1/\eps) \times k$ matrix) the number of guesses for $M$ is $2^{O(k^2 \cdot \poly(1/\eps) \cdot \polylog(k/\eps))}$. Note that the constraint on $V'$ is now different --- instead of having the constraint that $\|V'\|_\infty \leq \frac{\poly(n)}{\eps}$, as in the main algorithm of \cite{ptas_for_lplra}, or $\|V'\|_\infty \leq \poly(k)$, as in our presentation of that paper's algorithm, we instead enforce a constraint on the $\ell_1$-norm of $V'$. This has the benefit that the additive error obtained by minimizing $\med(MV' - SC)$ instead of $\med(SU^*V' - SC)$ is small --- this is necessary because we are now rounding $SU^*$ to a coarser granularity, $\frac{1}{\poly(k/\eps)}$ instead of $\frac{\eps^2}{\poly(n)}$. However, enforcing the constraint that $\|V'\|_1 \leq \poly(k)$ is nontrivial, as we now see.

\begin{remark}
We round each entry of $SU^*$ to the nearest power of $1 + \frac{1}{\poly(k/\eps)}$ --- a similar alternative approach that seems to work is rounding the entries to the nearest multiple of $\frac{1}{\poly(k/\eps)} \cdot \|C\|_1$, which is $\frac{1}{\poly(k/\eps)} \cdot OPT$ if $C$ is the residual from a $\poly(k)$-approximation algorithm.
\end{remark}

\paragraph{Ensuring that the Candidate Right Factor Has Norm At Most $\poly(k)$ Through Randomized Rounding of an ILP.}

How do we enforce the $\ell_1$-norm constraint on $V'$? First, let us discuss how the original median-based problem is solved in \cite{ptas_for_lplra}. For each column index $i \in [d]$, Algorithm \ref{algorithm:previous_eps_approximation_algorithm} finds $V_i \in \R^{2k}$ such that $V_i$ minimizes $\med(\widetilde{S}U^*V_i - \widetilde{S}C_i)$ subject to the constraint that $\|V_i\|_\infty$ is small. Note that there are $r!$ orderings of the coordinates of $\widetilde{S}U^*V_i - \widetilde{S}C_i$, where $r$ is the number of rows in $S$, meaning that all of those orderings can be tried --- for a fixed ordering of the coordinates, the ordering turns into a linear constraint, and the $\ell_\infty$ norm constraint can also be written as a linear constraint, meaning this can be solved with linear programming, and the overall running time is $r! \poly(nd) = 2^{O(r \log r)} \poly(nd)$, and this fits within the $n^{\poly(k/\eps)}$ running time of Algorithm \ref{algorithm:previous_eps_approximation_algorithm}.

Enforcing the constraint that $\|V'\|_1 \leq \poly(k)$ is more subtle. If we solve a similar problem on each $i \in [d]$ --- for instance, minimizing $\med(MV_i' - SC_i)$ such that $\|V_i'\|_1 \leq \poly(k)$ --- then the overall norm of $\|V'\|_1$ could still depend on $d$ in the worst case, if each minimizer $V_i'$ has norm roughly equal to $\poly(k)$. It is also not easy to directly include this constraint inside a median-based optimization problem that includes information from all the columns. For instance, one na\"ive way of minimizing $\med(MV' - SC)$ such that $\|V'\|_1 \leq \poly(k)$ is to do the following: simultaneously try all possible orderings, for each $i \in [d]$, of the coordinates of $MV'_i - SC_i$. The number of such orderings is $(r!)^d$, and this does not lead to an FPT running time.

Instead, we still solve separate median-based optimization problems for each $i \in [d]$, and combine the results for different $i$ using a relaxation of a suitable ILP. Instead of finding $V_i$ minimizing $\med(MV_i - SC_i)$, we instead seek to minimize the \textit{$\ell_1$ norm} of $V_i$. At the same time, we would like the cost $\med(MV_i - SC_i)$ to be small enough. Hence, for each column index $i \in [d]$, we find a column $V_{i, c}$ minimizing $\|V_{i, c}\|_1$, subject to the constraint that $\med(MV_{i, c} - SC_i) \leq c$ for a well-chosen $c$. For any $c$, the running time of this step is $r! \poly(nd)$ (by trying all orderings of the coordinates of $MV_{i, c} - SC_i$ and including the orderings as linear constraints in the LP), and this fits in our desired FPT running time.

Here, $c$ should be chosen so that it is not much higher than the cost $\med(MV_i^* - SC_i)$ of $V_i^*$ --- precisely, it should be within a $(1 + O(\eps))$ factor of $\med(MV_i^* - SC_i)$. Although we do not know $\med(MV_i^* - SC_i)$, we can guess all powers of $(1 + \eps)$ less than $O(1)\|C\|_1$ and greater than $O(\frac{\eps^2}{\poly(k)d}) \|C\|_1$ in the place of $c$. (The lower bound is chosen so that, even if the cost on some of the columns is $O(\frac{\eps^2}{\poly(k)d}) \|C\|_1$, the overall additive error of these columns is at most $O(\frac{\eps^2}{\poly(k)}) \|C\|_1$ which is acceptable.) The number of such cost bounds $c$ is thus polynomial in $d$, $k$ and $1/\eps$.

For each $i \in [d]$, we now have several minimizers $V_{i, c}$ for each possible cost bound $c$. Now, the question is, which cost bound $c$ should we pick for each $i \in [d]$? We can decide this through the following $0-1$ integer linear program. For each $i \in [d]$ and each possible cost bound $c$, we create a variable $x_{i, c}$ which can be $0$ or $1$ ($1$ representing the minimizer $V_{i, c}$ being chosen as the $i^{th}$ column of $V'$, and $0$ representing $V_{i, c}$ not being chosen). It is then natural to add the constraint that $\sum_c x_{i, c} = 1$ for each $i \in [d]$, since only one $V_{i, c}$ can be chosen as the $i^{th}$ column of $V'$. 

In addition, we wish to have $\|V'\|_1 \leq \poly(k)$ and $\med(MV' - SC) \leq (1 + O(\eps)) OPT_{C, 2k} + O(\eps^2/f) \|C\|_1$ (where $OPT_{C, 2k}$ is the optimal rank-$2k$ approximation error for $C$). Note that these can be made to hold if $V'$ is taken to be $V^*$, since for each $i \in [d]$, there is at least one cost bound $c$ for which $V^*_i$ is feasible. These can be represented as constraints that are linear in the $x_{i, c}$, since
$$\|V'\|_1 = \sum_{i = 1}^d \sum_c x_{i, c} \|V_{i, c}\|_1$$
and
$$\med(MV' - SC) = \sum_{i = 1}^d \sum_c x_{i, c} \med(MV_{i, c} - SC_i)$$

Now, solving this ILP will again take at least $2^{\Omega(kd)}$ time --- instead, we can relax the $0-1$ constraint on the $x_{i, c}$, so that we now have the constraints $x_{i, c} \in [0, 1]$ for all $i, c$. Since, for each $i \in [d]$, we also have the constraint $\sum_{c} x_{i, c} = 1$, this means that for each column $V_i'$ of $V'$, the $x_{i, c}$ give us a probability distribution on the cost bound $c$ to be chosen for $V_i'$. By the constraints of the new LP, if we sample for each $i \in [d]$ a single $V_{i, c}$ to be the $i^{th}$ column of $V'$, according to the distribution given by the $x_{i, c}$ (i.e., for each $i \in [d]$, $V_{i, c}$ is chosen with probability $x_{i, c}$) then the expectation of $\|V'\|_1$ is $\poly(k)$, while the expectation of $\med(MV' - SC)$ is at most $(1 + O(\eps))OPT + \frac{\eps}{\poly(k)}\|C\|_1$.

This gives the desired result, but a few subtleties arise when sampling according to the $x_{i, c}$ and analyzing this using Markov's inequality. To obtain a $(1 + O(\eps))$-approximation, we need $\med(MV' - SC)$ to be at most $(1 + O(\eps))$ times its expectation. By Markov's inequality, $\med(MV' - SC)$ is at most $(1 + 2\eps)$-times its expectation with probability at least $\eps$ (meaning this fails to occur with probability at most $1 - \eps$). To apply a union bound to control $\|V'\|_1$ as well, we note that $\|V'\|_1$ is at most $\frac{2}{\eps}$ times its expectation with failure probability at most $\frac{\eps}{2}$ --- having $\|V'\|_1 \leq \frac{\poly(k)}{\eps}$ instead of $\|V'\|_1 \leq \poly(k)$ is enough for our purposes, since we round $SU^*$ with a granularity of $\frac{1}{\poly(k/\eps)}$. By a union bound, $\|V'\|_1$ and $\med(MV' - SC)$ are both small enough with probability $\frac{\eps}{2}$ (i.e., failure probability $1 - \frac{\eps}{2}$), and we can simply sample $V'$ a total of $O(1/\eps)$ times, choosing the best solution found, to reduce this failure probability to a small constant independent of $\eps$.

Finally, since $V'$ has norm at most $\frac{\poly(k)}{\eps}$, the difference between $\med(MV' - SC)$ and $\med(SU^*V' - SC)$ is at most $\frac{1}{\poly(k/\eps)}\|C\|_1$, meaning $\med(SU^*V' - SC)$ is also small enough, and so is $\|U^*V' - C\|_1$. At this point, we can find an appropriate left factor $U'$ for $V'$ through linear programming.

As a summary of this discussion, we show our algorithm in the $\ell_1$-case in Algorithms \ref{algorithm:guessing_eps_approximation} and \ref{algorithm:get_eps_approximation_after_initialization}. Algorithm \ref{algorithm:get_eps_approximation_after_initialization} simply shows the process of obtaining an initial crude approximation $B$ and subtracting it from $A$ to obtain $C$, and Algorithm \ref{algorithm:guessing_eps_approximation} shows how we obtain a matrix $UV$ such that
$$\|UV - C\|_1 \leq (1 + O(\eps)) \min_{C^* \text{ rank }2k} \|C^* - C\|_1 + O\Big(\frac{\eps}{\poly(k)}\Big) \|C\|_1$$
meaning that $UV + B$ is a $(1 + \eps)$-approximation to the optimal rank-$k$ approximation error for $A$.

\begin{remark} \label{remark:guess_opt_remark_intro}
Note that we need to know $OPT$ in order to enforce the linear constraint that $\med(MV' - SC)$ is at most $(1 + O(\eps))OPT + \frac{\eps}{\poly(k)}\|C\|_1$. Observe that it suffices to have an estimate $\widehat{OPT}$ of $OPT$ that is accurate within a $(1 + \eps)$-factor. We can obtain such an $\widehat{OPT}$ as follows --- if $E$ is the error achieved by the rank-$k$ SVD of $C$, then $E$ is within a $\sqrt{nd}$ factor of $OPT$, meaning it suffices to guess all powers of $(1 + \eps)$ that are between $OPT$ and $\frac{1}{\sqrt{nd}}OPT$, and one of these will give a $(1 + O(\eps))$-approximate factorization with additive $\frac{\eps}{\poly(k)}\|C\|_1$ error.
\end{remark}

\begin{remark}
For each column index $i$ and cost bound $c$, we minimize the norm of $V_{i, c}$ such that $\med(MV_{i, c} - SC_i) \leq c$. The argument would also proceed similarly if we minimized $\med(MV_{i, c} - SC_i)$ while having a constraint on the norm of $V_{i, c}$. In particular, we can try all powers of $(1 + \eps)$ between $\frac{\poly(k)}{n}$ and $\poly(k)$, and the linear program will still be feasible because $\|V^*\|_1 \leq \poly(k)$.
\end{remark}

\begin{algorithm}
\caption{Obtaining a matrix $\widehat{A}$ such that $\|\widehat{A} - A\|_1 \leq (1 + \eps)OPT_{A, k} + \frac{\eps}{f}\|A\|_1$, where $OPT_{A, k} = \min_{A_k \text{ rank }k}\|A - A_k\|_1$. We first guess a sketched left factor $SU$, then find an appropriate right factor $V$ with norm at most $\poly(k)$. The argument $\eps$ is assumed to be at most $c$ for some absolute constant $c$.}
\label{algorithm:guessing_eps_approximation}
\begin{algorithmic}
\Require $A \in \R^{n \times d}$, $k \in \N$, $\eps \in (0, c)$, $f > 1$, $\widehat{OPT} \geq 0$
\Ensure $U \in \R^{n \times k}, V \in \R^{k \times d}$
\Procedure{GuessingAdditiveEpsApproximation}{$A, k, \eps, f, \widehat{OPT}$}
\State {If $A$ has rank $k$, return $A$.}
\State {$r \gets O(\max(k/\eps^6 \log(k/\eps), 1/\eps^9)$}
\State {$q \gets \poly(k)$}
\State {$S \gets $ An $r \times n$ matrix of i.i.d. standard Cauchy random variables}
\State {$\mathcal{I} \gets \{0\} \cup \Big\{\sigma \cdot (1 + \frac{1}{f\poly(k/\eps)})^t \mid t \in \Z, \frac{1}{f \poly(k/\eps)}\|A\|_1 \leq (1 + \frac{1}{f\poly(k/\eps)})^t \leq \poly(k/\eps)\|A\|_1 ,\, \sigma = \pm 1\Big\}$}
\State {$\calC \gets \Big\{ M \in \R^{r \times k} \mid M_{i, j} \in \mathcal{I} \,\,\, \forall i \in [r], j \in [k]\Big\}$ --- This is the set of (sketched) left factors we will guess.}

\item[]
\State{// Guess possible rounded (sketched) left factors and find a good right factor $V$, with $\|V\|_1 \leq \poly(k)$.}
\State {$U_{best} \gets 0 ,\, V_{best} \gets 0$}
\For {$M \in \calC$}
    \State {$\textsc{CostBounds} \gets \Big\{\frac{\eps^2\|A\|_1}{fd} \leq c \leq O(\|A\|_1)$ and $c$ is an integer power of $(1 + \eps)\Big\}$}
    \For {$i \in [d]$, $c \in \textsc{CostBounds}$}
        \State {$V_{i, c} \gets \argmin_{V_i} \|V_i\|_1$ subject to the constraint that $\med(MV_i - SA_i) \leq c$}
        \State {$C_{i, c} \gets \med(MV_{i, c} - SA_i)$}
    \EndFor
    
    \item[]
    \State{// Create LP to find a good distribution over $c \in \textsc{CostBounds}$ for each $i \in [d]$.}
    \State{$\textsc{Variables} \gets \{x_{i, c} \,\, \forall i \in [d], c \in \textsc{CostBounds}\}$}
    \State{$\textsc{Constraints} \gets \Big\{0 \leq x_{i, c} \,\, \forall i \in [d], c \in \textsc{CostBounds} \text{ and } \sum_{c \in \textsc{CostBounds}} x_{i, c} = 1 \,\, \forall i \in [d]\Big\}$}
    \State{$\textsc{Constraints} \gets \textsc{Constraints} \cup \Big\{ \sum_{i \in [d], c \in \textsc{CostBounds}} x_{i, c}\|V_{i, c}\|_1 \leq kq = \poly(k)\Big\}$}
    \State{$\Delta \gets (1 + O(\eps))\widehat{OPT} + O(\frac{\eps^2}{f})\|A\|_1$}
    \State{$\textsc{Constraints} \gets \textsc{Constraints} \cup \Big\{ \sum_{i \in [d], c \in \textsc{CostBounds}} x_{i, c}C_{i, c} \leq \Delta\Big\}$}
    \State{$x_{i, c} \gets $ Solution to the LP given by \textsc{Variables} and \textsc{Constraints}, for all $i \in [d]$, $c \in \textsc{CostBounds}$}
    \State{If LP is infeasible, then \textbf{continue} to next $M \in \calC$.}
    
    \item[]
    \State{// For each column, sample an appropriate cost bound. Do this $O(1/\eps)$ times, then $V'$ meets both}
    \State{// the cost and norm constraints with constant probability.}
    \For {$t = 1 \to 10/\eps$}
        \State{$c_i \gets $ An element $c \in \textsc{CostBounds}$ sampled according to the distribution on $\textsc{CostBounds}$} 
        \State{given by $\{x_{i, c} \mid c \in \textsc{CostBounds}\}$} 
        \State{$V_i' \gets V_{i, c_i}$ for all $i \in [d]$}
        \State{\textbf{Break} if $\|V'\|_1 \leq \frac{2kq}{\eps}$ and $\med(MV' - SA) \leq (1 + 2\eps)\Delta$}
    \EndFor
    \State{$U' \gets \argmin_U \|UV' - A\|_1$}
    \State{If $\|U'V' - A\|_1 \leq \|U_{best} V_{best} - A\|_1$ then $U_{best} \gets U'$ and $V_{best} \gets V'$.}
\EndFor \\
\Return {$U', V'$}
\EndProcedure
\end{algorithmic}
\end{algorithm}

\begin{algorithm}
\caption{Obtaining a $(1 + \eps)$-approximation with bicriteria rank at most $3k$. First apply \textsc{PolyKErrorNotBiCriteriaApproximation} from Algorithm \ref{algorithm:poly_k_rank_exactly_k} to $A$ to obtain a $\poly(k)$-approximation $B$ --- then, apply Algorithm \ref{algorithm:guessing_eps_approximation} to the residual to obtain an approximation $UV$ with additive error $\eps/\poly(k) \|A - B\|_1$. Finally, $B + UV$ gives a $(1 + \eps)$-approximation with rank $3k$.}
\label{algorithm:get_eps_approximation_after_initialization}
\begin{algorithmic}
\Require $A \in \R^{n \times d}$, $k \in \N$, $\eps \in (0, c)$
\Ensure $\widehat{A} \in \R^{n \times d}$ having rank $3k$
\Procedure{RoundingGuessingEpsApproximation}{$A, k, \eps$}
\State {$W, Z \gets \textsc{PolyKErrorNotBiCriteriaApproximation}(A, k)$}
\State {$B \gets WZ$}
\State {$C \gets A - B$}
\State {$f \gets \poly(k)$, the approximation factor of Algorithm \ref{algorithm:poly_k_rank_exactly_k}}

\item[]
\State {// Guess all $O((\log nd)/\eps)$ possible values for $\widehat{OPT}$ and try them for Algorithm \ref{algorithm:guessing_eps_approximation}} 
\State {// as described in Remark \ref{remark:guess_opt_remark_intro}.}
\State {$C_{SVD, 2k} \gets $ The optimal rank-$2k$ approximation for $C$ under the $\ell_2$ norm.}
\State {$\textsc{SvdError} \gets \|C - C_{SVD, 2k}\|_1$}
\State {$\widehat{A} \gets 0 \in \R^{n \times d}$}
\For {$t = 0 \to O(\frac{\log nd}{\eps})$}
    \State {$\widehat{OPT} \gets \textsc{SvdError}/(1 + \eps)^t$}
    \State {$U, V \gets \textsc{GuessingAdditiveEpsApproximation}(C, 2k, \eps, f, \widehat{OPT})$}
    \State {If $\|(B + UV) - A\|_1 \leq \|\widehat{A} - A\|_1$ then $\widehat{A} \gets B + UV$}
\EndFor \\
\Return {$\widehat{A}$}
\EndProcedure
\end{algorithmic}
\end{algorithm}

\paragraph{Hardness for Additive Error - Appendix \ref{appendix:additive_error_hardness}.} 

Our techniques for our hardness results are based on the proof by \cite{ptas_for_lplra} that, assuming the Small-Set Expansion Hypothesis and the Exponential Time Hypothesis, finding a constant-factor approximation for the optimal rank-$k$ approximation error takes at least $2^{k^c}$ time for some constant $c > 0$. To obtain our first result, that computing a matrix $\widehat{A}$ with rank at most $k$ such that
$$\|\widehat{A} - A\|_p \leq O(1) \min_{A_k \text{ rank }k} \|A_k - A\|_p + \frac{1}{2^{\poly(k)}}\|A\|_p$$
requires $2^{k^c}$ time for $p \in (1, 2)$ with the same hardness assumptions, we show that the reduction of \cite{ptas_for_lplra} from the Small Set Expansion problem to $\ell_p$ Low Rank Approximation can be performed in such a way that each entry of the input matrix $A$ ultimately has $\poly(k)$ bits in both its numerator and denominator. If this holds, then we can assume without loss of generality that the entries of $A$ are in fact integers with at most $\poly(k)$ bits, meaning $\|A\|_p \leq 2^{\poly(k)} \min_{A_k \text{ rank }k} \|A_k - A\|_p$, and the above guarantee is in fact equivalent to obtaining an $O(1)$-approximation.

The hardness results of \cite{ptas_for_lplra}, and our above hardness result, apply to $\ell_p$ low rank approximation for $p \in (1, 2)$, but not to $\ell_1$ low rank approximation. Intuitively, this is because the reduction of \cite{ptas_for_lplra} ultimately shows that finding the best rank---$(d - 1)$ approximation to an $n \times d$ matrix is NP-hard, assuming the Small Set Expansion Hypothesis. However, this is not true in the $p = 1$ case: finding the best rank---$(d - 1)$ subspace can actually be done in polynomial time when $p = 1$ \cite{bd13, sw_2011_hyperplane}. Instead, we show that a similar hardness result holds for the somewhat more general constrained $\ell_1$ low rank approximation problem. We do this by combining the techniques of \cite{ptas_for_lplra} with the following theorem:

\begin{theorem}[Embedding $\ell_p^n$ into $\ell_1^{n^{O(\log n)}}$ Deterministically]
\label{theorem:lp_to_l1_deterministic_intro}
Let $n \in \N$, and $p \in (1, 2)$. Then, there exists a matrix $R \in \R^{k \times n}$, for some $k = n^{O(\log n)}$, such that for all $x \in \R^n$,
$$\Omega(1) \|x\|_p \leq \|Rx\|_1 \leq O(1) \|x\|_p$$
$R$ can be computed deterministically in $n^{O(\log n)}$ time.
\end{theorem}

Below, we describe how we construct this embedding, and compare it with existing results related to embeddings of $\ell_p^n$ into $\ell_1^m$, for different values of $m$. Using this embedding, we first adapt the techniques of \cite{ptas_for_lplra} to show a similar hardness result for $\ell_{1, p}$ norm low-rank approximation:

\begin{theorem}[Hardness for Sum of Column $\ell_p$ Norms]
Let $A \in \R^{n \times d}$, $k \in \N$ and $p \in (1, 2)$. Then, assuming the Small-Set Expansion Hypothesis and the Exponential-Time Hypothesis, at least $2^{k^{\Omega(1)}}$ time is required to find a matrix $\widehat{A} \in \R^{n \times d}$ of rank at most $k$ such that
$$\|\widehat{A} - A\|_{1, p} \leq O(1) \min_{A_k \text{ rank }k} \|A_k - A\|_{1, p} + \frac{1}{2^{\poly(k)}}\|A\|_{1, p}$$
where for a matrix $M \in \R^{n \times d}$, its $\ell_{1, p}$ norm is $\|M\|_{1, p} = \sum_{j = 1}^d \|M_j\|_p$.
\end{theorem}

In fact, finding the best rank-$k$ approximation to $A$ in the $\ell_{1, p}$ norm requires $2^{k^{\Omega(1)}}$ time even in the special case when $A \in \R^{(k + 1) \times k^{O(\log k)}}$ --- these are the dimensions of the hard instances for the above theorem. Thus, to reduce from $\ell_{1, p}$ low rank approximation (on this family of hard instances $A$) to constrained $\ell_1$ low rank approximation, we can simply multiply on the left by an embedding matrix $R \in \R^{k^{O(\log k)} \times (k + 1)}$. In other words, to find $\widehat{A}$ such that
$$\|\widehat{A} - A\|_{1, p} \leq O(1) \min_{A_k \text{ rank }k} \|A_k - A\|_{1, p} + \frac{1}{2^{\poly(k)}} \|A\|_{1, p}$$
it suffices to find $\widehat{A} \in \R^{(k + 1) \times k^{O(\log k)}}$ with rank at most $k$ such that
$$\|R\widehat{A} - RA\|_1 \leq O(1) \min_{A_k \text{ rank }k} \|RA_k - RA\|_1 + \frac{1}{2^{\poly(k)}} \|RA\|_1$$
This is exactly constrained $\ell_1$ low rank approximation: the input matrix is $RA$, and the goal is to find a matrix $M$ of rank at most $k$, such that its columns are in the column span of $R$ (i.e. can be written as $R\widehat{A}$ for $\widehat{A}$ with rank at most $k$). The overall running time of the reduction is $k^{O(\log k)} = 2^{O(\log^2 k)} \ll 2^{k^{\Omega(1)}}$, meaning the same running time lower bound applies for constrained $\ell_1$ low rank approximation:

\begin{theorem}[Hardness for Constrained $\ell_1$ Low Rank Approximation]
Let $A \in \R^{m \times d}$, $R \in \R^{m \times n}$ and $k \in \N$. Assuming the Small-Set Expansion Hypothesis and Exponential-Time Hypothesis, at least $2^{k^{\Omega(1)}}$ time is required to find a matrix $\widehat{A} \in \R^{m \times d}$ with rank at most $k$, such that the columns of $\widehat{A}$ are in the column span of $R$ and, with constant probability,
$$\|\widehat{A} - A\|_1 \leq O(1) \min_{A_k} \|A_k - A\|_1 + \frac{1}{2^{\poly(k)}}\|A\|_1$$
The minimum on the right-hand side is taken over all matrices $A_k$ with rank at most $k$, whose columns are contained in the span of $R$.
\end{theorem}

Note that Algorithm \ref{algorithm:guessing_eps_approximation} can be modified to work for constrained $\ell_1$ low rank approximation, with the same running time guarantee. The only modification to the algorithm itself would be that, once the right factor $V'$ is obtained, then the left factor $U$ is set to be $R \cdot (\argmin_{U_0} \|RU_0V' - A\|_1)$ instead of $\argmin_U \|UV' - A\|_1$, where $R$ is the subspace to which the low-rank solution is constrained --- note that this new minimizer can also be found using linear programming. The only change that will be made to the analysis in the introduction and in Theorem \ref{theorem:correctness_of_guessing_eps_approx} is that $U^*$ will be the optimal rank-$k$ left factor whose columns are contained in the span of $R$, rather than the optimal (un-constrained) rank-$k$ left factor.

\begin{remark}
This reduction corrects an error in an earlier version of this work. Previously we claimed that $\ell_p$ low rank approximation could be directly reduced to constrained $\ell_1$ low rank approximation through the use of an embedding matrix. However, this line of reasoning leads to a reduction from $\ell_p$ low rank approximation to constrained $\ell_{p, 1}$ low rank approximation, and it is unclear how to obtain a reduction from constrained $\ell_{p, 1}$ low rank approximation to constrained $\ell_1$ low rank approximation. In this version, we fix this issue by adding the intermediate step of reducing to $\ell_{1, p}$ low rank approximation.
\end{remark}

\paragraph{Deterministic Embedding of $\ell_p$ into $\ell_1$}

We deterministically construct the embedding matrix mentioned in Theorem \ref{theorem:lp_to_l1_deterministic_intro}, using an observation from \cite{matousek_metric_embeddings}:
\begin{fact}[Observation 2.7.2 of \cite{matousek_metric_embeddings}]
Let $R_1, R_2, \ldots, R_n$ be real random variables on a probability space that has $k$ elements $\omega_1, \omega_2, \ldots, \omega_k$, and let $A \in \R^{k \times n}$ such that $A_{i, j} = \Prob{\omega_i} R_j(\omega_i)$. For $x \in \R^n$ let us set $X := \sum_{j = 1}^n R_j x_j$. Then, $E[ |X| ] = \|Ax\|_1$.
\end{fact}

By this observation, if $E[|X|] = \Theta(1) \|x\|_p$, then the corresponding matrix $A$ gives an embedding of the $\ell_p$ norm into the $\ell_1$ norm. For this, we can select the random variables $R_1, R_2, \ldots, R_n$ in the observation above to be standard $p$-stable random variables, for $p \in (1, 2)$ --- to show that $E[|X|] = \Theta(1) \|x\|_p$, we can make use of results on $p$-stable random variables from \cite{l1_lower_bound_and_sqrtk_subset}. For the sample space to have a finite number of elements $k$, we can round the $R_i$ to the nearest multiple of $\poly(n)$ and truncate them so that they are not more than $\poly(n)$.

However, if we let the $R_i$ be fully independent, then the sample space may have a very large size $k$ (i.e. it may be exponential in $n$). Instead, we let the $R_i$ be $O(\log n)$-wise independent --- by a result of \cite{km93_small_sample_space}, the $R_i$ can be constructed using a sample space of size $k = n^{O(\log n)}$, and we can use lemmas from \cite{knw_10_approximating_lp_norm_p_stable} and \cite{knpw11_pth_moment} on $p$-stable random variables with limited independence to show that $E[|X|] = \Theta(1)\|x\|_p$ even when the $R_i$ are $O(\log n)$-wise independent.

\begin{remark}
In an earlier version of this work, we instead used the randomized embedding of \cite{js1982_embedding}. To obtain a running time lower bound when this randomized embedding is used, we assumed a randomized version of ETH that was previously used in \cite{dhmtw14_randomized_eth}, for instance. In the current version, we use the deterministic embedding mentioned in Theorem \ref{theorem:lp_to_l1_deterministic_intro} in its place. The randomized embedding had the advantage that it provided an embedding of $\ell_p^n$ into $\ell_1^{O(n)}$, rather than $\ell_p^n$ into $\ell_1^{n^{O(\log n)}}$. One issue with our previous proof is that we assumed that the entries of the embedding of \cite{js1982_embedding} can be computed in $\poly(n)$ time --- while we believe this is true, it is not immediately clear how to prove this. Another work \cite{karnin2011_embedding_lp_into_l1} deterministically constructs $(1 - \eps)n$-dimensional subspaces of $\R^n$ in which the $\ell_1$ norm and the $\ell_p$ norm are the same up to normalization and constant factors. However, \cite{karnin2011_embedding_lp_into_l1} does not provide a matrix $R$ which can be used to embed $\R^{(1 - \eps)n}$ with the $\ell_p$ norm into $\R^n$ with the $\ell_1$ norm with low distortion, and hence is not applicable for our reduction.
\end{remark}

\subsection{Paper Outline}

\subsubsection{Main Results - Algorithms}

In Section \ref{section:lp_css}, we describe our $\widetilde{O}(k^{1/p - 1/2})$-approximate algorithm for $\ell_p$ column subset selection for $p \in [1, 2)$, and its analysis. In Section \ref{section:fpt_approx}, we analyze our $(1 + \eps)$-approximation algorithm for $\ell_p$ low rank approximation, for $p \in [1, 2)$, which returns a matrix of rank at most $3k$ in $2^{\poly(k/\eps)} + \poly(nd)$ time.

\subsubsection{Appendices - Additional Results}

In Appendix \ref{appendix:additive_error_hardness}, we show how to extend the hardness results of \cite{ptas_for_lplra} to show that even obtaining an $O(1)$-approximation for $\ell_p$ low rank approximation with $\frac{1}{2^{\poly(k)}}\|A\|_p$ additive error is hard, when $p \in (1, 2)$. Using this, we show that obtaining an $O(1)$-approximation for constrained $\ell_1$ low rank approximation with $\frac{1}{2^{\poly(k)}}\|A\|_1$ additive error is hard. Next, in Appendix \ref{appendix:lp_css_lower_bound_full_proofs} we show how the lower bound for $\ell_1$ column subset selection due to \cite{l1_lower_bound_and_sqrtk_subset} can be extended to obtain a lower bound for $\ell_p$ column subset selection, for $p \in (1, 2)$. Finally, in Appendix \ref{appendix:poly_k_bicriteria_rank}, we show how to obtain a $\poly(k)$-approximation algorithm with running time $\poly(nd)$, based on Algorithm \ref{algorithm:column_subset_sampling} and techniques from \cite{algorithm3_original, l1_lower_bound_and_sqrtk_subset}.

\newpage
\section{Optimal $\ell_1$ Column Subset Selection via Random Sampling}
\label{section:lp_css}

\subsection{Preliminaries: Notation and Lewis Weight Sampling} \label{subsection:preliminaries}

Suppose $A \in \R^{n \times d}$. We use the following notation for the rows, columns, and submatrices of $A$. For $i \in [n]$, we let $A^i$ be the $i^{th}$ row of $A$, and for $j \in [d]$, we let $A_j$ be the $j^{th}$ column of $A$. Moreover, for $S \subset [d]$, we let $A_S$ be the submatrix of $A$, such that its columns are those of $A$ whose indices are in $S$, and for $R \subset [n]$, we let $A^R$ be the submatrix of $A$, such that its rows are those of $A$ whose indices are in $R$.

We now recall basic facts about row sampling with Lewis weights. We use Lewis weight sampling as a black box, for more details we refer the reader to \cite{lewis_weights, l1_lower_bound_and_sqrtk_subset}.

\begin{lemma} [Sampling and Rescaling Matrix Based on Lewis weights - Adapted from Theorem 7.1 of \cite{lewis_weights}]
\label{lemma:sampling_by_lewis_weights}
Let $A \in \R^{n \times d}$, $1 \leq p < 2$, and let $r = O(d \log d)$ if $p = 1$ and $r = O(d \log d (\log \log d)^2)$ otherwise. There exists a distribution $(p_1, p_2, \ldots, p_n)$ on the rows of $A$ such that if $S$ is a matrix with $r$ rows, each row chosen independently as the $i^{th}$ standard basis vector times $\frac{1}{(rp_i)^{\frac{1}{p}}}$ with probability $p_i$, then with probability $0.9$,
$$\Omega(1)\|Ax\|_p \leq \|SAx\|_p \leq O(1) \|Ax\|_p$$
for all $x \in \R^d$. The distribution $(p_1, p_2, \ldots, p_n)$ can be computed in $\nnz(A) + \poly(d)$ time.
\end{lemma}

\begin{lemma} [$O(1)$ Dilation and Contraction for Lewis Weights - Lemma D.11 $(p = 1)$ and E.11 $(p \in (1, 2))$ of \cite{l1_lower_bound_and_sqrtk_subset}]
\label{lemma:lewis_weights_no_dilation_contraction}
Let $M \in \R^{n \times d}$ and $U \in \R^{n \times t}$. Let $r = O(t \log t)$ if $p = 1$ and $O(t \log t (\log \log t)^2)$ if $p \in (1, 2)$, and suppose $S \in \R^{r \times n}$ is a sampling and rescaling matrix generated according to the $\ell_p$ Lewis weights of $U$. Then, with probability $0.9$, $\|SM\|_p^p \leq O(1) \|M\|_p^p$, and with probability $0.9$, for all $x \in \R^t$, $\|SUx\|_p^p \geq \Omega(1)\|Ux\|_p^p$.
\end{lemma}

\begin{lemma} [$O(1)$ Contraction on Affine Subspace for Lewis Weights - From Lemmas D.11 and D.7 of \cite{l1_lower_bound_and_sqrtk_subset} for $p = 1$, and Lemmas E.11 and E.7 of \cite{l1_lower_bound_and_sqrtk_subset} for $p \in (1, 2)$]
\label{lemma:lewis_weights_affine_contraction}
Let $A \in \R^{n \times d}$, $U \in \R^{n \times k}$, and let $V^* = \argmin_{V \in \R^{k \times d}} \|UV - A\|_p$. Let $r = O(k \log k)$ if $p = 1$ and $r = O(k \log k (\log \log k)^2)$ if $p \in (1, 2)$. Let $S \in \R^{r \times n}$ be a sampling and rescaling matrix generated according to the $\ell_p$ Lewis weights of $U$. Then, with probability $0.9$, simultaneously for all $V \in \R^{k \times d}$,
$$\|SUV - SA\|_p^p \geq \Omega(1) \|UV - A\|_p^p - O(1) \|UV^* - A\|_p^p$$
\end{lemma}

\subsection{$\ell_p$ Column Subset Selection Algorithm and Analysis} \label{subsection:lp_css_algorithm_and_analysis}

In this section, we analyze Algorithm \ref{algorithm:column_subset_sampling}. Note that the algorithm itself is nearly the same as Algorithm 1 of \cite{DBLP:conf/nips/SongWZ19} \footnote{Our pseudocode in Algorithm \ref{algorithm:column_subset_sampling} is based on Algorithm 2 of \url{https://arxiv.org/pdf/1811.01442v1.pdf}, which is version 1 of \cite{DBLP:conf/nips/SongWZ19} on arXiv.}, with the difference being that $2r = O(k \log k)$ columns are sampled per iteration, instead of $2k$ --- however, its analysis is significantly different, and does not use maximum-determinant column subsets. Instead, we rely on the following key existence result, shown in Theorem C.1 of \cite{l1_lower_bound_and_sqrtk_subset}:

\begin{theorem} \label{theorem:subset_existence}
(Existence of a Good Column Subset)
Let $A \in \R^{n \times d}$, $p \in [1, 2)$, $k \in \N$, and $r = O(k \log k)$ if $p = 1$ and $r = O(k \log k (\log \log k)^2)$ if $1 < p < 2$. Then, there exist matrices $U \in \R^{n \times r}$ and $V \in \R^{r \times d}$, such that the columns of $U$ are columns of $A$, and
$$\|UV - A\|_p \leq O\Big(r^{\frac{1}{p} - \frac{1}{2}}\Big) \min_{A_k \text{ rank }k} \|A - A_k\|_p$$
Moreover, this occurs with probability $\frac{9}{10}$ if the columns of $U$ are sampled from the columns of $A$ according to the Lewis weights of $V^*$, where $U^*V^*$ is an optimal rank-$k$ approximation to $A$.
\end{theorem}
\begin{proof}
This result was essentially shown in Theorem C.1 of \cite{l1_lower_bound_and_sqrtk_subset} for the case $p = 1$ --- we give a brief sketch here. The key point is that $U = AR$, where $R$ is given by (column) Lewis weights of $V^*$ (here we first generate a sampling matrix based on the Lewis weights of $(V^*)^T$, then transpose the sampling matrix).

Indeed, if $R$ is chosen according to the column Lewis weights of $V^*$, then letting $U^{i\prime}$ be the minimizer of $\|U^iV^*R - A^iR\|_2$ instead of $\|U^iV^*R - A^iR\|_1$ for all $i \in [n]$ gives an $O(r^{\frac{1}{2}})$-approximation. This minimizer $U^{i\prime}$ is given by $U^{i\prime} = A^iR(V^*R)^{+}$, where $(V^*R)^{+}$ is the pseudo-inverse of $V^*R$. Hence we can choose $U' = AR(V^*R)^{+}$, and hence there exists an $O(r^{\frac{1}{2}})$-approximation for $A$ with left factor $AR$.

For $1 < p < 2$, the proof is nearly identical \cite{l1_lower_bound_and_sqrtk_subset}. Instead of Lemma D.11 of \cite{l1_lower_bound_and_sqrtk_subset} we can use part (III) of Lemma E.11 of the same work to extend the result to $p \neq 1$. Similarly, instead of Lemma D.8, we can use Lemma E.8 of \cite{l1_lower_bound_and_sqrtk_subset}. Finally, we apply Lemma B.10 of \cite{l1_lower_bound_and_sqrtk_subset} to convert between the $\ell_p$-norm and the $\ell_2$-norm, rather than between the $\ell_1$-norm and the $\ell_2$-norm as was done in \cite{l1_lower_bound_and_sqrtk_subset} --- now we obtain a distortion of $r^{\frac{1}{p} - \frac{1}{2}}$.
\end{proof}

\begin{algorithm}
\caption{Randomly sample columns of $A$ repeatedly, to obtain $O(k \cdot \polylog(k) \log(d))$ columns of $A$ spanning a good approximation. This is a variant of Algorithm 1 in \cite{DBLP:conf/nips/SongWZ19}, with the difference being that we sample $k \cdot \polylog(k)$ columns in each round instead of $2k$ columns. Here, \textsc{MultipleRegressionSolver}$(n, d, m, U, B)$ (where $U \in \R^{n \times d}$, $B \in \R^{n \times m}$) is a subroutine which computes $\min_x \|Ux - B_j\|_p^p$ for each $j \in [m]$. The call $\textsc{BottomK}(\textsc{Sort}(cost), \Omega(m))$ serves to find the $\Omega(m)$ column indices in $[m]$ having the smallest regression cost.}
\label{algorithm:column_subset_sampling}
\begin{algorithmic}
\Require $A \in \R^{n \times d}$, $k \in \N$, $p \in [1, 2)$
\Ensure $S \subset [d]$, $|S| = O(k \log k \log d)$ if $p = 1$ and $O(k \log k (\log \log k)^2 \log d)$ otherwise
\Procedure{RandomColumnSubsetSelection}{$A, k, p$}
\State {$\text{samples} \gets O(\log d)$}
\State {$T_0 \gets [d]$}
\State {$r \gets O(k \log k)$ if $p = 1$ and $O(k \log k (\log \log k)^2)$ otherwise}
\For{$i = 1 \to \text{samples}$}
    \State {$m \gets |T_{i - 1}|$}
    \For{$j = 1 \to O(\log d)$} 
        \State {Sample $S^{(j)}$ from $\binom{T_{i - 1}}{2r}$ uniformly at random}
        \State {$m \gets |T_{i - 1} \setminus S^{(j)}|$}
        \State {$\{\text{cost}_t\}_{t \in T_{i - 1} \setminus S_j} \gets \textsc{MultipleRegressionSolver}(n, 2r, m, A_{S^{(j)}}, A_{T_{i - 1} \setminus S^{(j)}})$}
        \State {$R^{(j)} \gets \textsc{BottomK}(\textsc{Sort}(cost), \Omega(m))$}
        \State {$c_j \gets \sum_{t \in R^{(j)}} \text{cost}_t$}
    \EndFor
    \State {$j^* \gets \min_{j \in [O(\log d)]} {c_j}$}
    \State {$S_i \gets S^{(j^*)} \cup R^{(j^*)}$}
    \State {$T_i \gets T_{i - 1} \setminus S_i$}
\EndFor
\State {$S \gets \cup_i S_i$} \\
\Return {$S, S_1, S_2, \ldots, S_{O(\log d)}$}
\EndProcedure
\end{algorithmic}
\end{algorithm}

Now, we analyze Algorithm \ref{algorithm:column_subset_sampling}:

\begin{theorem} [Column Sampling Approximation Factor]
\label{theorem:column_subset_selection_log_d_approximation}
Let $A \in \R^{n \times d}$, $p \in [1, 2)$, $k \in \N$, and $r = O(k \log k)$ if $p = 1$ and $r = O(k \log k (\log \log k)^2)$ otherwise. Let $S \subset [d]$ be the set of columns of $A$ that is returned by Algorithm \ref{algorithm:column_subset_sampling}. Let $U = A_S \in \R^{n \times O(r \log d)}$, and let $V = \argmin_{V' \in \R^{O(r \log d) \times d}} \|UV' - A\|_p$, i.e., $V$ is the optimal right factor for $A_S$. Then, with probability $1 - o(1)$,
$$\|UV - A\|_p \leq O(r^{\frac{1}{p} - \frac{1}{2}} (\log d)^{\frac{1}{p}})) \min_{A_k \text{rank }k} \|A - A_k\|_p$$
\end{theorem}
\begin{proof}
We first show the following claim (whose role is analogous to Lemma 6 of \cite{algorithm3_original}), which is that after sampling $2r$ columns of $A$, at least a constant fraction of the remaining columns are covered, up to our desired approximation factor, with constant probability. The main theorem is then a consequence of this claim, together with a Markov bound (to show that a constant fraction of the columns of $A$ are covered with constant probability) and a union bound over the samples in each of the iterations (note that we perform $O(\log d)$ repetitions in each iteration, and use the result of the best repetition, to boost the probability of a constant fraction of columns being covered).

\begin{claim} \label{claim:each_column_covered_constant probability}
Let $B \in \R^{n \times 2r}$ be a submatrix of $A$, whose columns are a uniformly random subset of those of $A$, of size $2r$. Furthermore, let $A_i$ be an additional, uniformly random, column of $A$ not among those of $B$. Then, with constant probability (where the probability is taken over $B$ and $A_i$),
$$\min_{x \in \R^{2r}} \|Bx - A_i \|_p^p \leq O(r^{1 - \frac{p}{2}}) \frac{OPT^p}{d}$$
where $OPT = \min_{A_k \text{rank-}k} \|A - A_k\|_p$.
\end{claim}

\begin{proof}
The proof is somewhat along the lines of Lemma 6 of \cite{algorithm3_original}, but the analysis is based on the existence result in Theorem \ref{theorem:subset_existence} rather than the existence of an $O(k)$-approximate column subset, as in \cite{algorithm3_original}. Let $B' \in \R^{n \times (2r + 1)}$ be $B$ with $A_i$ adjoined. Then, the columns of $B'$ also form a uniformly random subset of the columns of $A$. Define $\Delta \in \R^{n \times d}$ so that if $U^*V^*$ is the optimal rank-$k$ approximation to $A$, then $\Delta = A - U^*V^*$. For a subset $S \subset [d]$, let $\Delta_S$ denote the submatrix of $\Delta$ containing those columns whose indices are in $S$. Finally, let $T \subset [d]$ such that $B' = A_T$, and let $B'_k$ be the best rank-$k$ approximation to $B'$.

Then, by the definition of $B_k'$, $\|B' - B_k'\|_p^p \leq \|\Delta_T\|_p^p$. Hence, taking expectations gives
\begin{equation} \label{eq:expectation}
\begin{split}
E\Big[\|B' - B_k'\|_p^p\Big] \leq E\Big[\|\Delta_T\|_p^p\Big] = \frac{|T|}{d}OPT^p = O\Big(\frac{r}{d}OPT^p \Big)
\end{split}
\end{equation}
where the last equality is because $T$ is a uniformly random subset of $[d]$ of size $2r + 1$. Hence, by Markov's inequality, with probability $19/20$, $\|B' - B_k'\|_p^p \leq O\Big(\frac{r}{d}OPT^p\Big)$ --- denote this event by $\calE_1$.

Now we use Theorem \ref{theorem:subset_existence}, by which there exists a subset $S$ of column indices of $B'$ of size $r$ such that, if we let $B_S'$ denote the corresponding submatrix of $B'$, and $V = \argmin_{V' \in \R^{r \times |T|}} \|B_S'V' - B'\|_p$, then
$$\|B_S'V - B'\|_p \leq O(r^{\frac{1}{p} - \frac{1}{2}}) \|B_k' - B'\|_p$$
i.e. $\|B_S'V - B'\|_p^p \leq O(r^{1 - \frac{p}{2}}) \|B_k' - B'\|_p^p$. If $\calE_1$ holds, then this implies that
$$\|B_S' V - B'\|_p^p \leq O(r^{1 - \frac{p}{2}}) \frac{r \cdot OPT^p}{d}$$
Finally, let $i \in T$ be a uniformly random column index --- for instance, the uniformly random column of $B'$ that is not in $B$. Then,
\begin{equation} \label{eq:uniformly_random_column}
\begin{split}
E\Big[\|B_S'V_i - B_i'\|_p^p \mid \calE_1 \Big] \leq \frac{1}{2r + 1} O(r^{1 - \frac{p}{2}}) \frac{r \cdot OPT^p}{d} = O(r^{1 - \frac{p}{2}}) \frac{OPT^p}{d}
\end{split}
\end{equation}
and hence, assuming $\calE_1$, with probability at least $19/20$, $\|B_S'V_i - B_i'\|_p^p \leq O(r^{1 - \frac{p}{2}}) OPT^p/d$. Let $\calE_2$ denote the event that this holds --- then, the probability of $\calE_1 \cap \calE_2$ is at least $18/20$. 

Finally, let $\calE_3$ be the event that $A_i$ is not in the $O(r^{\frac{1}{p} - \frac{1}{2}})$-approximate column subset $S$. Note that $S$ has size $r$, while $B'$ has $2r + 1$ columns. Moreover, $S$ is determined entirely by $B'$ --- therefore, since $A_i$ is also a uniformly random column of $B'$, with probability at least $\frac{1}{2}$, $A_i$ is not in $S$, meaning that it is covered well by $B$ (since $S$ is contained entirely in $B$).

In summary, if $\calE_1$, $\calE_2$ and $\calE_3$ all hold, and $x$ is obtained by performing $\ell_p$ regression, using $B$ to fit $A_i$, then
$$\|Bx - A_i\|_p^p \leq \|B_S'V_i - B_i'\|_p^p \leq O(r^{1 - \frac{p}{2}}) \frac{OPT^p}{d}$$
Moreover, the failure probability of $\calE_1 \cap \calE_2 \cap \calE_3$ is at most $\frac{1}{2} + \frac{2}{10} = \frac{7}{10}$. This proves the lemma.
\end{proof}

We can combine the above claim with a Markov bound, as follows. Let $B \in \R^{n \times 2r}$ be a submatrix of $A$, such that the column indices of $B$ form a uniformly random subset of $T_{i - 1}$ (here, we are using the notation of Algorithm \ref{algorithm:column_subset_sampling} --- $T_i$ is the set of indices of columns of $A$ which have not been discarded after $i$ iterations). For $i \in [m]$, let $Z_i$ be equal to $1$ if $A_i$ is approximately covered by $B$, i.e., 
$$\min_{x \in \R^{2r}} \|Bx - A_i\|_p^p \leq O\Big(r^{1 - \frac{p}{2}} \frac{OPT^p}{d}\Big)$$
and $0$ otherwise. Then, $E_{B, i}[Z_i] = c$, where $0 < c < 1$, and the expectation is taken over uniformly random $B$ and $i$. Hence, if we let $Z := \sum_i Z_i$ be the number of approximately covered columns, then $E[Z] = cm$, and $E[m - Z] = (1 - c)m$. Therefore, by Markov's inequality, with constant probability $F > 0$, $(m - Z) \leq (1 - \frac{c}{2})m$, meaning $Z \geq \frac{c}{2}m$. In other words, with probability $F$, there exists $R \subset T_{i - 1}$ with $|R| \geq \frac{c}{2}|T_{i - 1}|$ such that for each $j \in R$,
$$\min_{V_j \in \R^k} \|BV_j - A_j\|_p^p \leq O\Big(r^{1 - \frac{p}{2}} \frac{OPT^p}{|T_{i - 1}|}\Big)$$
meaning
$$\min_{V \in \R^{k \times |R|}} \|BV - A_R\|_p^p \leq O\Big(r^{1 - \frac{p}{2}} \frac{OPT^p \cdot |T_{i - 1}|}{|T_{i - 1}|}\Big) = O\Big(r^{1 - \frac{p}{2}} OPT^p\Big)$$
Denote this event by $\calE$, meaning $\calE$ occurs with probability $F$. Then, as done in Algorithm \ref{algorithm:column_subset_sampling}, it is sufficient to sample the submatrix $B$ in $O(\log d)$ independent iterations, and take the sampled submatrix which minimizes the sum of the lowest $\frac{c}{2}|T_{i - 1}|$ residuals. By choosing $B$ in this way, we ensure that $\calE$ has failure probability at most $(1 - F)^{O(\log d)} = \frac{1}{d^{O(1)}}$. Since we perform $O(\log d)$ iterations, we only have to perform a union bound over $O(\log d)$ such events, meaning the overall failure probability of Algorithm \ref{algorithm:column_subset_sampling} is $O(\log d) \cdot \frac{1}{d^{O(1)}} = o(1)$.
\end{proof}

\begin{remark}
Note that our application of linearity of expectation in the above proof is valid. For a fixed $i \in [m]$, it may not seem that we can take the expectation $E_{B, i}[Z_i]$ over a uniformly random column index $i$, since $i$ is determined by $Z_i$. However, we could for instance shuffle the indices $i \in [m]$, and sum the $Z_i$ in the shuffled order --- then, each index $i \in [m]$ is a uniformly random column index, and we can use linearity of expectation.
\end{remark}

We can also remove the $O((\log d)^{\frac{1}{p}})$ term in the approximation factor from Theorem \ref{theorem:column_subset_selection_log_d_approximation} with the following refined analysis, reminiscent of one performed in \cite{DBLP:conf/nips/SongWZ19} --- we examine the number of times a column of $A$ can remain ``uncovered" before it is removed:

\begin{theorem} [Column Sampling - Better Approximation Factor] 
\label{thm:column_sampling_better_approximation_factor}
Let $A \in \R^{n \times d}$, $p \in [1, 2)$, $k \in \N$ and $r = O(k \log k)$ if $p = 1$ and $r = O(k \log k (\log \log k)^2)$ otherwise. Let $S \subset [d]$ be the set of columns of $A$ that is returned by Algorithm \ref{algorithm:column_subset_sampling}. Let $U = A_S \in \R^{n \times O(r \log d)}$ and let $V = \argmin_{V' \in \R^{O(r \log d) \times d}} \|UV' - A\|_p$, i.e., $V$ is the optimal right factor for $A_S$. Then, with probability $1 - o(1)$,
$$\|UV - A\|_p \leq O(r^{\frac{1}{p} - \frac{1}{2}} (\log k)^{\frac{1}{p}}) \min_{A_k \text{ rank } k} \|A - A_k\|_p$$
\end{theorem}

\begin{proof}
Define $\Delta$ as in the proof of Theorem \ref{theorem:column_subset_selection_log_d_approximation}. Consider the $i^{th}$ iteration of Algorithm \ref{algorithm:column_subset_sampling}, and recall that $T_{i - 1}$ is the set of remaining column indices. Let $m = |T_{i - 1}|$. Finally, let $T_{i - 1, big} \subset T_{i - 1}$ consist of the $\frac{m}{r}$ indices of $T_{i - 1}$ with greatest cost, rounding down if $\frac{m}{r}$ is not an integer (i.e. $j \in T_{i - 1, big}$ if $\|\Delta_j\|_p$ is among the top $\frac{m}{r}$ column norms of $\Delta$).

Note that with constant probability, $S^{(j)}$ will be disjoint from $T_{i - 1, big}$ (provided $r$ is multiplied by a sufficiently large constant and $m \geq \Omega(r)$), since the probability of not selecting an element of $T_{i - 1, big}$ in a uniformly random subset of size $2r$ is at least $(1 - O(\frac{1}{r}))^{2r} \geq \Omega(1)$ for $r$ sufficiently large. Since we take $O(\log d)$ samples per iteration and choose the best one, this occurs on all iterations with probability $1 - o(1)$.

Condition on this event (which we can call $\calE_1$) occurring --- then, for uniformly random $j \in T_{i - 1} \setminus T_{i - 1, big}$, with constant probability, if $B = A_{S^{(i)}}$,
$$\min_{x \in \R^{2r}} \|Bx - A_j\|_p^p \leq O\Big(r^{1 - \frac{p}{2}} \frac{OPT_{T_{i - 1} \setminus T_{i - 1, big}}^p}{m}\Big)$$
This is by Claim \ref{claim:each_column_covered_constant probability}, since $S^{(j)}$ is a uniformly random subset of $T_{i - 1} \setminus T_{i - 1, big}$. Therefore, conditioning on $\calE_1$, with constant probability, the smallest $\Omega(m)$ columns have a cost of $O(r^{1 - \frac{p}{2}} OPT_{T_{i - 1} \setminus T_{i - 1, big}}^p)$. Finally, $\calE_1$ occurs in the $i^{th}$ iteration with probability at least $1 - \frac{1}{\poly(d)}$ since we repeat the sampling process $O(\log d)$ times and take the sample which gives the smallest cost on the lowest $\Omega(m)$ columns --- hence, with probability $1 - \frac{1}{\poly(d)}$, the smallest $\Omega(m)$ columns have a cost of $O(r^{1 - \frac{p}{2}} OPT_{T_{i - 1} \setminus T_{i - 1, big}}^p)$, and this occurs on every iteration $i$ with probability $1 - o(1)$ by a union bound (since there are $O(\log d)$ iterations).

It remains to bound $\sum_i OPT_{T_{i - 1} \setminus T_{i - 1, big}}^p$ where $i$ ranges across all of the iterations --- we show that
$$\sum_i OPT_{T_{i - 1} \setminus T_{i - 1, big}}^p = O((\log k)OPT^p)$$
First assume that on each iteration $i$, the column indices $j$ which are ``covered'' (and hence discarded) have the smallest $\|\Delta_j\|_p$. If we remove this assumption, then this can only decrease $\sum_i OPT_{T_{i - 1} \setminus T_{i - 1, big}}^p$. To see why, suppose $i$ is the index of the latest iteration on which this does not occur, meaning there exist column indices $j_1$ and $j_2$ such that in iteration $i$, column $j_1$ is not removed and column $j_2$ is removed, and $\|\Delta_{j_1}\|_p < \|\Delta_{j_2}\|_p$. Let $(A_1, A_2, \ldots, A_t)$ be the $\|\Delta_j\|_p$s of the remaining columns if column $j_2$ is removed, and $(B_1, B_2, \ldots, B_t)$ be the $\|\Delta_j\|_p$s of the remaining columns if column $j_1$ is removed instead. Then, for all $j \in [t]$, $A_j \leq B_j$, and hence, if on all subsequent iterations, the columns are removed in order of their $\|\Delta_j\|_p$s, then this will only increase $OPT_{T_{i - 1} \setminus T_{i - 1, big}}$ for \textit{all} subsequent $i$.

We can therefore argue that on the last iteration (and by recursing, on all iterations) where the columns are not removed in order of their $\|\Delta_j\|_p$s, removing them in order of their $\|\Delta_j\|_p$s instead can only increase $\sum_i OPT_{T_{i - 1} \setminus T_{i - 1, big}}^p$. Hence, we can assume without loss of generality that in iteration $i$, the indices $j$ in $T_{i - 1}$ with the $\Omega(m)$ lowest values of $\|\Delta_j\|_p$ will be removed. If this holds, then for $j \in T_{i - 1} \setminus T_{i - 1, big}$, after $O(\log r)$ iterations, i.e. for iteration $i'$ where $i' \geq i + O(\log r)$, $j$ will not be in $T_{i' - 1}$, since $T_{i' - 1}$ will be equal to $T_{i - 1, big}$ (since a constant fraction of the columns in $T_{i - 1}$ are removed in each iteration).
\end{proof}

\newpage
\section{$(1 + \eps)$-Approximation in FPT Time with Bicriteria Rank $3k$} \label{section:fpt_approx}

In this section, we give an algorithm for $\ell_p$-low rank approximation, for $p \in [1, 2)$, which runs in $2^{\poly(k/\eps)} + \poly(nd)$ time and outputs a matrix of rank $3k$. Our algorithm discussed in Subsubsection \ref{subsubsection:general_l1_lra_techniques} is a special case of this algorithm, in the case $p = 1$. The analysis here is similar to the analysis given in the introduction, but uses sketching matrices whose entries are $p$-stable random variables, rather than Cauchy random variables.

\subsection{Preliminaries: Median-Based Estimator for $\ell_p$-norm Dimension Reduction from \cite{ptas_for_lplra}}

We first recall some concepts from \cite{ptas_for_lplra} related to sketches based on medians and dense $p$-stable random matrices (Section 2 of \cite{ptas_for_lplra} for the $p = 1$ case and Subsection 3.2 of \cite{ptas_for_lplra} for the $p \in (1, 2)$ case). In the following, we let $B \in \R^{n \times d}$.

\begin{definition}[$p$-Stable Random Variables - As Defined in Section 3.2 of \cite{ptas_for_lplra}]
\label{def:p_stable_rvs}
Suppose $Z, Z_1, Z_2, \ldots, Z_n$ are i.i.d. random variables, and $p \in [1, 2]$. We say $Z$ and $Z_i$ are $p$-stable if, for any $x \in \R^n$, $\|x\|_pZ$ and $\sum_{i = 1}^n x_i Z_i$ have the same distribution.
\end{definition}

Note that $p$-stable random variables only exist for $p \in (0, 2]$ --- we consider $p \in [1, 2)$. $1$-stable random variables are also called Cauchy random variables, as in our description of our algorithm in the introduction. For $p \in [1, 2)$, we use $\med_p$ to denote the median of a half-$p$-stable random variable --- that is, if $Z$ is a $p$-stable random variable, then $|Z|$ is a half-$p$-stable random variable. Note that the median of a half-Cauchy random variable is just $1$, but for $p \in (1, 2)$, there is no simple closed form for $\med_p$. $\med_p$ can be computed up to a $(1 \pm \eps)$-factor, as described in \cite{knw_10_approximating_lp_norm_p_stable} --- this is enough for our purposes. This definition is relevant for the median-based sketch of \cite{ptas_for_lplra}, as we see below.

\begin{definition} [Medians and Quantiles of Vectors (Definition 4 in \cite{ptas_for_lplra})]
For a vector $v \in \R^n$, we let $\med(v)$ be the median of $|v_i|$ for $i \in [n]$. In addition, for $\alpha \in [0, 1]$, we let $q_{\alpha}(v)$ denote the minimum value greater than $\lceil \alpha n \rceil$ of the values $|v_1|, |v_2|, \ldots, |v_n|$.
\end{definition}

The following lemmas from \cite{ptas_for_lplra} allow us to obtain very accurate estimates of the $\ell_p$-norms of matrices after first multiplying by a dense $p$-stable matrix to reduce the dimension.

\begin{lemma}[$p$-stable Matrix + Median Preserves Norms (from \cite{ptas_for_lplra})]
\label{lemma:p_stable_matrix_median_preserves_norm}
Let $S$ be an $m \times n$ matrix with i.i.d. standard $p$-stable entries and let $M$ be an $n \times d$ matrix. For $\eps > 0$, with probability $1 - \frac{1}{\Omega(1)}$,
$$(1 - \eps)\|M\|_p \leq \Big(\sum_i \med(SM_i)^p\Big)^{\frac{1}{p}} \Big/ \med_p \leq (1 + \eps)\|M\|_p$$
as long as $m = \poly(1/\eps)$.
\end{lemma}
\begin{proof}
This is Lemma 6 from \cite{ptas_for_lplra} in the case $p = 1$, and Lemma 12 of \cite{ptas_for_lplra} in the case $p \in (1, 2)$.
\end{proof}

\begin{remark}
From inspecting the proof of the above lemma in the $p = 1$ case, it can be applied as long as $m \geq \frac{1}{\eps^3}$. This is because $\Prob{\med(SM_i) = (1 \pm \eps)\|M_i\|_1} \leq e^{-\Theta(\eps^2)m}$ (by Lemma 4 of \cite{ptas_for_lplra}) and as long as $m \geq \frac{1}{\eps^3}$, this probability is at most $e^{-\frac{1}{\eps}} \leq \eps$, which is sufficient for this lemma.
\end{remark}

\begin{lemma} [$p$-stable Matrix + Top Quantile Does not Cause Dilation (from \cite{ptas_for_lplra}]
\label{lemma:p_stable_top_quantile_no_dilation}
When $S$ is an $m \times n$ matrix with i.i.d. standard $p$-stable entries, $m = \poly(1/\eps)$, and $M$ is an $n \times d$ matrix, then with probability $1 - \frac{1}{\Omega(1)}$,
$$\Big( \sum_i q_{1 - \frac{\eps}{2}} (SM_i)^p \Big)^{\frac{1}{p}} \Big/ \med_p \leq O\Big(\frac{1}{\eps}\Big) \|M\|_p$$
\end{lemma}
\begin{proof}
This is Lemma 7 from \cite{ptas_for_lplra} in the $p = 1$ case, and Lemma 13 of \cite{ptas_for_lplra} in the $p \in (1, 2)$ case.
\end{proof}

\begin{remark}
From inspecting the proof of the above lemma in the $p = 1$ case, we see that it can be applied as long as $m \geq \frac{1}{\eps^2}$. This is because $\Prob{q_{1 - \eps/2}(SM_i) \geq \frac{c}{\eps}\|M_i\|_1} \leq e^{-\Theta(\eps)c \cdot m} \leq e^{-O(1/\eps)} \leq \eps$ (by Lemma 4 of \cite{ptas_for_lplra}) as long as $m \geq \frac{1}{\eps^2}$.
\end{remark}

\begin{remark}
In addition, the failure probabilities in the above two lemmas can be as small as desired, since they are obtained from Markov bounds.
\end{remark}

\begin{lemma} [Quasi-Subspace Embedding with Median Estimator (from \cite{ptas_for_lplra})]
\label{lemma:p_stable_median_subspace_embedding}
Let $X$ be a $k$-dimensional subspace of $\R^n$ and $\eps, \delta > 0$. Let $S$ be an $m \times n$ matrix whose entries are i.i.d. standard $p$-stable random variables, where $m = O(1/\eps^2 \cdot k \log (k/\eps \delta))$. Then, with probability at least $1 - \Theta(\delta)$, for all $x \in X$,
$$(1 - \Theta(\eps)) \|x\|_p \leq q_{\frac{1}{2} - \eps}(Sx/\med_p) \leq q_{\frac{1}{2} + \eps}(Sx/\med_p) \leq (1 + O(\eps))\|x\|_p$$
\end{lemma}
\begin{proof}
This is Lemma 5 of \cite{ptas_for_lplra} in the $p = 1$ case, and Lemma 11 of \cite{ptas_for_lplra} in the $p \in (1, 2)$ case.
\end{proof}

\begin{remark}
Note that we can select any number of rows that is larger than some $O(\frac{1}{\eps^2}k\log(\frac{k}{\eps\delta}))$, by inspecting the proof of Lemma 5 of \cite{ptas_for_lplra} --- any number of rows larger than the specified $O(\frac{1}{\eps^2}k\log(\frac{k}{\eps\delta}))$ allows the net argument in that proof to work.
\end{remark}

The following lemma shown in \cite{ptas_for_lplra} allows this sketch to serve as a ``quasi-affine embedding":

\begin{lemma}
Let $U \in \R^{n \times k}$ and $A \in \R^{n \times d}$. Let $V^*$ be chosen to minimize $\|UV^* - A\|_p$. Suppose $S$ is an $m \times n$ random matrix such that:
\begin{enumerate}
    \item $q_{\frac{1}{2} - \eps}(SUx/\med_p) \geq (1 - \Theta(\eps)) \|Ux\|_p$ for all $x \in \R^k$
    \item For each $i \in [d]$, with probability at least $1 - \eps^3$, $\med(S[U, A_i]x/\med_p) \geq (1 - \eps^3) \|[U, A_i]x\|_p$ for all $x \in \R^{k + 1}$
    \item $(\sum_i \med(SUV_i^* - SA_i)^p)^{\frac{1}{p}} / \med_p \leq (1 + \eps^3) \|UV^* - A\|_p$
    \item $(\sum_i q_{1 - \eps/2}(SUV_i^* - SA_i)^p)^{\frac{1}{p}} / \med_p \leq O(\frac{1}{\eps}) \|UV^* - A\|_p$
\end{enumerate}
If statements 1, 3 and 4 each hold with probability $1 - \frac{1}{\Omega(1)}$, then with probability $1 - \frac{1}{\Omega(1)}$, for all $V \in \R^{k \times d}$,
$$(\sum_i \med(SUV_i - SA_i)^p)^{\frac{1}{p}} \Big/ \med_p \geq (1 - O(\eps)) \|UV - A\|_p$$
\end{lemma}
\begin{proof}
This is Theorem 11 of \cite{ptas_for_lplra} in the $p = 1$ case and Theorem 12 of \cite{ptas_for_lplra} in the $p \in (1, 2)$ case. Note that in the statements of these Theorems in \cite{ptas_for_lplra}, it is not explicitly stated that statements 1, 3 and 4 only need to hold with constant probability. However, this is true because, by inspecting the proof of Theorem 11 of \cite{ptas_for_lplra}, we see that statement 2 is only used to perform a Markov bound, which needs to hold with constant probability --- once that Markov bound is obtained, a union bound can be performed over statements 1, 2, 3 and 4.
\end{proof}

Specializing this lemma to $p$-stable matrices gives:

\begin{lemma} [Lower Bound for One-Sided Embedding with Median Estimator (from \cite{ptas_for_lplra})]
\label{lemma:affine_embedding_lower_side_p_stable_median}
Let $U \in \R^{n \times k}$ and $A \in \R^{n \times d}$. If $S$ is an $m \times n$ random matrix with i.i.d. standard $p$-stable entries, where $m = O(\max(k/\eps^6 \log (k/\eps), 1/\eps^9))$, then with probability $1 - \frac{1}{\Omega(1)}$,
$$(\sum_i \med(SUV_i - SA_i)^p)^{\frac{1}{p}} \Big/ \med_p \geq (1 - O(\eps)) \|UV - A\|_p$$
for all $V \in \R^{k \times d}$.
\end{lemma}
\begin{proof}
This is a corollary of all the above lemmas. For statement 1 of the previous lemma to hold, it is sufficient for $S$ to have at least $O(k/\eps^2 \log (k/\eps))$ rows by Lemma \ref{lemma:p_stable_median_subspace_embedding}. For statement 2 of the previous lemma to hold, it is sufficient for $S$ to have at least $O(k/\eps^6 \log (k/\eps))$ rows, again by Lemma \ref{lemma:p_stable_median_subspace_embedding}. For statement 3 of the previous lemma to hold, it is sufficient for $S$ to have $O(1/\eps^9)$ rows, by Lemma \ref{lemma:p_stable_matrix_median_preserves_norm}. Finally, for statement 4 of the previous lemma to hold, it is sufficient for $S$ to have $O(1/\eps^2)$ rows, by Lemma \ref{lemma:p_stable_top_quantile_no_dilation}. 
\end{proof}

We recall another useful lemma on $p$-stable matrices, which bounds the $\ell_p$-norm of $SM$ if $S$ is a $p$-stable matrix and $M$ is a fixed matrix:

\begin{lemma} [Distortion in $\ell_p$-norm with $p$-stable Matrices - Lemma E.11 of \cite{l1_lower_bound_and_sqrtk_subset}]
\label{lemma:p_stable_lp_norm_distortion}
Let $M \in \R^{n \times d}$, and let $S \in \R^{r \times n}$ be a random matrix whose entries are i.i.d. standard $p$-stable random variables. Then, with probability $1 - \frac{1}{\Omega(1)}$,
$$\|SM\|_p^p \leq O(r\log d) \|M\|_p^p$$
\end{lemma}
Note that in the original statement of this lemma in \cite{l1_lower_bound_and_sqrtk_subset}, the $p$-stable matrices are rescaled by $\Theta(1/r^{1/p})$ --- since our sketch uses $p$-stable matrices that are not rescaled, we include this factor in the distortion.

In the course of our analysis, it will also be useful to note that $\med_p$ is bounded away from $0$ --- that is, there exists a constant $K > 0$ such that $\med_p \geq K$. To our knowledge, this fact was not explicitly shown elsewhere, and we prove it below.

\begin{lemma} [$\med_p$ is $\Omega(1)$]
\label{lemma:p_stable_median_omega_1}
There exists an absolute constant $K > 0$ such that $\med_p \geq K$ for all $p \in [1, 2]$.
\end{lemma}
\begin{proof}
First, we recall the following formula from \cite{nolan_97_pstable_density} for the c.d.f. of a standard $p$-stable random variable for $p \in (1, 2]$. For $x > 0$, if $X$ is a standard $p$-stable random variable, then
$$\Prob{X > x} = 1 - \frac{1}{\pi} \int_0^{\frac{\pi}{2}} e^{-x^{\frac{p}{p - 1}} \cdot V(\theta; p)} d\theta$$
where
$$V(\theta; p) = \Big(\frac{\cos \theta}{\sin p\theta}\Big)^{\frac{p}{p - 1}} \cdot \frac{\cos (p - 1)\theta}{\cos \theta}$$
(This is a corollary of Theorem 1 of \cite{nolan_97_pstable_density}.) Hence, because $X$ is symmetric, $x > 0$ is less than $\med_p$ if and only if
$$\frac{3}{4} > \Prob{X > x} = 1 - \frac{1}{\pi} \int_0^{\frac{\pi}{2}} e^{-x^{\frac{p}{p - 1}}V(\theta; p)} d\theta$$
or equivalently
$$I := \int_0^{\frac{\pi}{2}} e^{-x^{\frac{p}{p - 1}} V(\theta; p)} d\theta > \frac{\pi}{4}$$
(recall that $\med_p$ is the median of $|X|$ rather than $X$). Now, our goal is to bound the integral $I$ on the left-hand side from below. As a first step, we show the following claim.
\begin{claim}
Let $c > 0$ be a sufficiently small absolute constant (to be chosen outside of this claim). There exists a constant $D > 0$ (which may depend on $c$) such that for all $p \in [1, 2]$ and $\theta \in [\frac{\pi}{8}, \frac{\pi}{2} - c]$,
$$\frac{\cos(\theta)}{\sin(p\theta)} \cdot \Big(\frac{\cos((p - 1)\theta)}{\cos(\theta)}\Big)^{\frac{p - 1}{p}} \leq D$$
\end{claim}
\begin{proof}
Observe that for $p \in [1, 2]$ and $\theta \in [\frac{\pi}{8}, \frac{\pi}{2} - c]$, $\frac{\pi}{8} \leq p\theta \leq \pi - 2c$, meaning $\sin(p\theta) > 0$ and is in fact bounded away from $0$ because $[1, 2] \times [\frac{\pi}{8}, \frac{\pi}{2} - c]$ is compact. Let $C_0 > 0$ be such that $\sin(p\theta) \geq C_0$ for all $(p, \theta) \in [1, 2] \times [\frac{\pi}{8}, \frac{\pi}{2} - c]$. Then,
\begin{equation}
\begin{split}
\frac{\cos(\theta)}{\sin(p\theta)} \cdot \Big(\frac{\cos((p - 1)\theta)}{\cos(\theta)}\Big)^{\frac{p - 1}{p}} \leq  \cos(\theta)^{\frac{1}{p}} \cdot (\cos((p - 1)\theta))^{\frac{p - 1}{p}} \cdot \frac{1}{C_0} \leq \frac{1}{C_0}
\end{split}
\end{equation}
where the second inequality holds because $\cos(\theta)$ and $\cos((p - 1)\theta)$ are at most 1. (Note that $\cos(\theta)^{\frac{1}{p}}$ and $\cos((p - 1)\theta)^{\frac{p - 1}{p}}$ are well-defined because $\cos(\theta)$ and $\cos((p - 1)\theta)$ are nonnegative for this choice of $\theta$.)
\end{proof}

Hence, for $p \in [1, 2]$ and $\theta \in [\frac{\pi}{8}, \frac{\pi}{2} - c]$, $$V(\theta; p) \leq \Big(\frac{\cos \theta}{\sin p\theta}\Big)^{\frac{p}{p - 1}} \cdot \frac{\cos (p - 1)\theta}{\cos \theta} = \Big(\frac{\cos(\theta)}{\sin(p\theta)} \cdot \Big(\frac{\cos((p - 1)\theta)}{\cos(\theta)}\Big)^{\frac{p - 1}{p}}\Big)^{\frac{p}{p - 1}} \leq D^{\frac{p}{p - 1}}$$

Hence, we can bound $I$ from below:
\begin{equation}
\begin{split}
\int_0^{\frac{\pi}{2}} e^{-x^{\frac{p}{p - 1}}V(\theta; p)} d\theta
& \geq \int_{\frac{\pi}{8}}^{\frac{\pi}{2} - c} e^{-x^{\frac{p}{p - 1}} D^{\frac{p}{p - 1}}} d\theta \\
& = \Big(\frac{3\pi}{4} - c\Big) e^{-(xD)^{\frac{p}{p - 1}}} 
\end{split}
\end{equation}
We can choose $c = \frac{\pi}{8}$, meaning $I > \frac{\pi}{4}$ as long as $\frac{5\pi}{8} \cdot e^{-(xD)^{\frac{p}{p - 1}}} > \frac{\pi}{4}$ or $e^{-(xD)^{\frac{p}{p - 1}}} > \frac{2}{5}$. This holds as long as $x < \frac{1}{D}(-\log(2/5))^{\frac{p - 1}{p}}$. Letting $C = -\ln(2/5)$, we observe that $0 < C < 1$, meaning $\sqrt{C} \leq C^{\frac{p - 1}{p}} \leq 1$.

In summary, $x \leq \med_p$ as long as $x < \frac{\sqrt{C}}{D}$, and hence we can take $\frac{\sqrt{C}}{2D}$ to be the desired $K$.
\end{proof}

\subsection{$(1 + \eps)$-Approximation Algorithm and Analysis}

We now present and analyze our $(1 + \eps)$-approximation algorithm with bicriteria rank $3k$ for $\ell_p$-low rank approximation, where $p \in [1, 2)$.

\begin{algorithm}
\caption{Guessing a sketched left factor $SU$, and finding an appropriate right factor $V$ with norm at most $\poly(k)$, to obtain a $(1 + \eps)$-approximation with additive $\eps/(f\poly(k)) \|A\|_p^p$ error.}
\label{algorithm:guessing_eps_approximation_general_lp}
\begin{algorithmic}
\Require $A \in \R^{n \times d}$, $k \in \N$, $\eps \in (0, c)$, $f > 1$, $\widehat{OPT} \geq 0$, $p \in [1, 2)$
\Ensure $U \in \R^{n \times k}, V \in \R^{k \times d}$
\Procedure{GuessingAdditiveEpsApproximation}{$A, k, \eps, f, \widehat{OPT}, p$}
\State {If $A$ has rank $k$, return $A$.}
\State {$r \gets O(\max(k/\eps^6 \log (k/\eps), 1/\eps^9))$}
\State {$q \gets \poly(k)$}
\State {$S \gets $ An $r \times n$ matrix of i.i.d. standard $p$-stable random variables}
\State {$\mathcal{I} \gets \{0\} \cup \Big\{\sigma \cdot (1 + \frac{1}{f\poly(k/\eps)})^t \mid t \in \Z, \frac{1}{f \poly(k/\eps)}\|A\|_p \leq (1 + \frac{1}{f\poly(k/\eps)})^t \leq \poly(k/\eps)\|A\|_p ,\, \sigma = \pm 1\Big\}$}
\State {$\calC \gets \Big\{ M \in \R^{r \times k} \mid M_{i, j} \in \mathcal{I} \,\,\, \forall i \in [r], j \in [k]\Big\}$ --- This is the set of (sketched) left factors we will guess.}

\item []
\State {// Guess possible rounded (sketched) left factors and find a good right factor $V$, with $\|V\|_p \leq \poly(k)$.}
\State {$U_{best} \gets 0 ,\, V_{best} \gets 0$}
\For {$M \in \calC$}
    \State {$\textsc{CostBounds} \gets \Big\{\frac{\eps^2}{f}\|A\|_p/d^{\frac{1}{p}} \leq c \leq O(\|A\|_p)$ and $c$ is an integer power of $(1 + \eps)\Big\}$}
    \For {$i \in [d]$, $c \in \textsc{CostBounds}$}
        \State {$V_{i, c} \gets \argmin_{V_i} \|V_i\|_p$ subject to the constraint that $\med(MV_i - SA_i)/\med_p \leq c$}
        \State {$C_{i, c} \gets \med(MV_{i, c} - SA_i)/\med_p$}
    \EndFor
    
    \item[]
    \State {// Create LP to find a good distribution over $c \in \textsc{CostBounds}$ for each $i \in [d]$.}
    \State{$\textsc{Variables} \gets \{x_{i, c} \,\, \forall i \in [d], c \in \textsc{CostBounds}\}$}
    \State{$\textsc{Constraints} \gets \Big\{0 \leq x_{i, c} \,\, \forall i \in [d], c \in \textsc{CostBounds} \text{ and } \sum_{c \in \textsc{CostBounds}} x_{i, c} = 1 \,\, \forall i \in [d]\Big\}$}
    \State{$\textsc{Constraints} \gets \textsc{Constraints} \cup \Big\{ \sum_{i \in [d], c \in \textsc{CostBounds}} x_{i, c}\|V_{i, c}\|_p^p \leq kq^p\Big\}$}
    \State{$\Delta \gets (1 + O(\eps))^p \Big(\widehat{OPT} + \frac{1}{f\poly(k/\eps)}\|A\|_p \Big)^p + \Big(\frac{\eps^2}{f}\Big)^p \|A\|_p^p$}
    \State{$\textsc{Constraints} \gets \textsc{Constraints} \cup \Big\{ \sum_{i \in [d], c \in \textsc{CostBounds}} x_{i, c}C_{i, c}^p \leq \Delta\Big\}$}
    \State{$x_{i, c} \gets $ Solution to the LP given by \textsc{Variables} and \textsc{Constraints}, for all $i \in [d]$, $c \in \textsc{CostBounds}$}
    \State{If LP is infeasible, then \textbf{continue} to next $M \in \calC$.}
    
    \item[]
    \State{// For each column, sample an appropriate cost bound. Do this $O(1/\eps)$ times, then $V'$ meets both} 
    \State{// the cost and norm constraints with constant probability.}
    \For {$t = 1 \to 10/\eps$}
        \State{$c_i \gets $ An element $c \in \textsc{CostBounds}$ selected according to the distribution on $\textsc{CostBounds}$} 
        \State{given by $\{x_{i, c} \mid c \in \textsc{CostBounds}\}$ (note that for $i \in [d]$, the $x_{i, c}$ are nonnegative and sum to $1$)}
        \State{$V_i' \gets V_{i, c_i}$ for all $i \in [d]$}
        \State{\textbf{Break} if $\|V'\|_p^p \leq \frac{2kq^p}{\eps}$ and $\sum_{i = 1}^d \med(MV_i' - SA_i)^p / \med_p^p \leq (1 + 2\eps)\Delta$}
    \EndFor
    \State{$U' \gets \argmin_U \|UV' - A\|_p$}
    \State{If $\|U'V' - A\|_p \leq \|U_{best} V_{best} - A\|_p$ then $U_{best} \gets U'$ and $V_{best} \gets V'$.}
\EndFor \\
\Return {$U', V'$}
\EndProcedure
\end{algorithmic}
\end{algorithm}

\begin{algorithm}
\caption{First apply \textsc{PolyKErrorNotBiCriteriaApproximation} from Algorithm \ref{algorithm:poly_k_rank_exactly_k} to $A$ to obtain a $\poly(k)$-approximation $B$ --- then, apply Algorithm \ref{algorithm:guessing_eps_approximation_general_lp} to the residual to obtain an approximation $UV$ with additive error $1/\poly(k/\eps) \|A - B\|_p \leq \eps \cdot OPT$. Finally, $B + UV$ gives a $(1 + \eps)$-approximation with rank $3k$.}
\label{algorithm:get_eps_approximation_after_initialization_general_lp}
\begin{algorithmic}
\Require $A \in \R^{n \times d}$, $k \in \N$, $\eps \in (0, c)$, $p \in [1, 2)$
\Ensure $\widehat{A} \in \R^{n \times d}$ having rank $3k$
\Procedure{RoundingGuessingEpsApproximation}{$A, k, \eps, p$}
\State {$W, Z \gets \textsc{PolyKErrorNotBiCriteriaApproximation}(A, k, p)$}
\State {$B \gets WZ$}
\State {$C \gets A - B$}
\State {$f \gets \poly(k)$, the approximation factor of Algorithm \ref{algorithm:poly_k_rank_exactly_k}}

\item[]
\State {// Guess all $O((\log nd)/\eps)$ possible values for $\widehat{OPT}$ and try them for Algorithm \ref{algorithm:guessing_eps_approximation_general_lp}} 
\State {// as described in Remark \ref{remark:guess_opt_remark}.}
\State {$C_{SVD, 2k} \gets $ The optimal rank-$2k$ approximation for $C$ under the $\ell_2$ norm.}
\State {$\textsc{SvdError} \gets \|C - C_{SVD, 2k}\|_p$}
\State {$\widehat{A} \gets 0 \in \R^{n \times d}$}
\For {$t = 0 \to O(\frac{\log nd}{\eps})$}
    \State {$\widehat{OPT} \gets \textsc{SvdError}/(1 + \eps)^t$}
    \State {$U, V \gets \textsc{GuessingAdditiveEpsApproximation}(C, 2k, \eps, f, \widehat{OPT}, p)$}
    \State {If $\|(B + UV) - A\|_p \leq \|\widehat{A} - A\|_p$ then $\widehat{A} \gets B + UV$}
\EndFor \\
\Return {$\widehat{A}$}
\EndProcedure
\end{algorithmic}
\end{algorithm}

\begin{theorem} [Correctness and Running Time of Algorithm \ref{algorithm:guessing_eps_approximation_general_lp}]
\label{theorem:correctness_of_guessing_eps_approx}
Let $A \in \R^{n \times d}$, $k \in \N$, $\eps \in (0, C)$ (where $C$ is a sufficiently small absolute constant), $f > 1$ and $p \in [1, 2)$. Furthermore, if $OPT = \min_{A_k \text{ rank }k} \|A_k - A\|_p$, then suppose $(1 - O(\eps)) OPT \leq \widehat{OPT} \leq (1 + O(\eps)) OPT$. Finally, let $U \in \R^{n \times k}, V \in \R^{k \times d}$ be the output of $\textsc{GuessingEpsApproximation}(A, k, \eps, f, \widehat{OPT}, p)$ (shown in Algorithm \ref{algorithm:guessing_eps_approximation_general_lp}). Then,
$$\|UV - A\|_p \leq (1 + O(\eps))OPT + O\Big(\frac{\eps}{f}\Big) \|A\|_p$$
The running time of Algorithm \ref{algorithm:guessing_eps_approximation_general_lp} is at most $f^{O(rk)} + 2^{O(rk\log(k/\eps))} + \poly(fnd/\eps)$, where $r$, the number of rows in the $p$-stable sketching matrix, is $O(\max(k^2/\eps^6\log(k/\eps), 1/\eps^9))$.
\end{theorem}

\begin{remark} \label{remark:guess_opt_remark}
We can obtain $\widehat{OPT}$ efficiently as follows. If $A_{SVD, k}$ is the rank-$k$ SVD of $A$ (which can be computed in polynomial time) then $OPT \leq \|A_{SVD, k} - A\|_p \leq (nd)^{1/p - 1/2}OPT$, meaning we can guess all integer powers of $1 + \eps$ between $\|A_{SVD, k} - A\|_p$ and $\frac{1}{(nd)^{1/p - 1/2}}\|A_{SVD, k} - A\|_p$ --- the number of guesses is $O(\frac{\log (nd)}{\eps})$. We can input all of those guesses to Algorithm \ref{algorithm:guessing_eps_approximation_general_lp}, and one of them will produce the right answer. This is done when applying Algorithm \ref{algorithm:guessing_eps_approximation_general_lp} within Algorithm \ref{algorithm:get_eps_approximation_after_initialization_general_lp}.
\end{remark}

\begin{remark}
Note that $f$ represents the approximation factor of the initialization algorithm used to obtain $A$. In our case $f$ will equal $\poly(k)$, but we will analyze this algorithm for a general $f$.
\end{remark}

\begin{proof}
Let $U^* \in \R^{n \times k}$, $V^* \in \R^{k \times d}$ such that $\|U^*V^* - A\|_p = OPT$. Without loss of generality, assume $V^*$ is a $q$-well-conditioned basis (where $q = \poly(k)$), meaning that for all $x \in \R^k$,
$$\frac{\|x\|_p}{q} \leq \|x^TV^*\|_p \leq q\|x\|_p$$
(Note that well-conditioned bases exist for all $p$, for instance see Lemma 10 of \cite{algorithm3_original}.) In particular, this implies that $\|V^*\|_p^p \leq kq^p$ by letting $x$ be each of the standard basis vectors. Now, let $S \in \R^{r \times n}$ be a random matrix where each entry is an i.i.d. $p$-stable random variable, where $r = O(\max(k/\eps^6 \log (k/\eps), 1/\eps^9))$ as in Algorithm \ref{algorithm:guessing_eps_approximation_general_lp}.

We first analyze the effect of rounding $SU^*$, multiplicatively, so that the absolute value of each entry is rounded to the nearest power of $1 + \frac{1}{f\poly(k/\eps)}$, or set to $0$ if it is too small (here, $\poly(k/\eps)$ in the denominator is $(k/\eps)^c$ for a sufficiently large constant $c$). Note that $\|U^*V^*\|_p \leq O(1)\|A\|_p$ by the triangle inequality. Hence, because $V^*$ is a well-conditioned basis, $O(1)\|A\|_p \geq \frac{1}{q}\|U^*\|_p$, and $\|U^*\|_p \leq O(q)\|A\|_p$. Therefore, $\|SU^*\|_\infty \leq \|SU^*\|_p \leq \poly(k/\eps)\|U^*\|_p \leq \poly(k/\eps)\|A\|_p$, where the second inequality is due to Lemma \ref{lemma:p_stable_lp_norm_distortion}, since $U^*$ has $k$ columns and $S$ has $\poly(k/\eps)$ rows.

Now, let $M_1$ be $SU^*$, but with the absolute value of each entry rounded to the nearest power of $1 + \frac{1}{f\poly(k/\eps)}$. In addition, let $M_2$ be $M_1$, but with each entry having absolute value less than $\frac{1}{f\poly(k/\eps)} \|A\|_p$ being replaced by $0$. Then,
$$\|M_2 - M_1\|_p^p \leq \Big(\frac{1}{f\poly(k/\eps)}\Big)^p\|A\|_p^p \cdot O\Big(\frac{k^2}{\eps^9} \log (k/\eps)\Big) \leq \Big(\frac{1}{f\poly(k/\eps)}\Big)^p\|A\|_p^p$$
where the first inequality holds because $M_1$ and $M_2$ have $rk = O(\max(k^2/\eps^6 \log(k/\eps), 1/\eps^9))$ entries. Moreover,
$$\|SU^* - M_1\|_p^p \leq \Big(\frac{1}{f \poly(k/\eps)}\Big)^p \|SU^*\|_p^p \leq \Big(\frac{1}{f\poly(k/\eps)}\Big)^p \cdot \poly(k/\eps) \|A\|_p^p \leq \Big(\frac{1}{f\poly(k/\eps)}\Big)^p \|A\|_p^p$$
where the first inequality is because for any $a \in \R$, and $t \in (0, 1)$, if $\widehat{a}$ is $a$ with its absolute value rounded to the nearest power of $(1 + t)$, then $|a - \widehat{a}|^p \leq t^p|a|^p$. The second inequality is simply because $\|SU^*\|_p^p \leq \poly(k/\eps) \|A\|_p^p$ as mentioned above. Hence,
\begin{equation} \label{eq:rounding_error}
\begin{split}
\|SU^* - M_2\|_p
& \leq \|SU^* - M_1\|_p + \|M_1 - M_2\|_p \\
& \leq \Big(\Big(\frac{1}{f\poly(k/\eps)}\Big)^p \|A\|_p^p\Big)^{\frac{1}{p}} + \Big(\Big(\frac{1}{f\poly(k/\eps)}\Big)^{p} \|A\|_p^p\Big)^{\frac{1}{p}} \\
& \leq \frac{1}{f\poly(k/\eps)} \|A\|_p + \frac{1}{f\poly(k/\eps)} \|A\|_p \\
& \leq \frac{1}{f\poly(k/\eps)}\|A\|_p
\end{split}
\end{equation}
where the first inequality is due to the triangle inequality.

Now, observe that Algorithm \ref{algorithm:guessing_eps_approximation_general_lp} will guess $M_2$ at some point --- that is, since $M_2 \in \calC$, $M$ will be equal to $M_2$ at some point. Suppose $M = M_2$. Let us condition on the following events involving $S$. Let $\calE_1$ denote the event that for all $V \in \R^{k \times d}$,
$$\Big(\sum_i \med(SU^*V_i - SA_i)^p\Big)^{1/p} \Big/ \med_p \geq (1 - O(\eps)) \|U^*V - A\|_p$$
$\calE_2$ the event that
$$\Big(\sum_i \med(SU^*V^*_i - SA_i)^p\Big)^{1/p} \Big/ \med_p \leq (1 + \eps) \|U^*V^* - A\|_p = (1 + \eps) OPT$$
and $\calE_3$ the event that
$$\|SU^*\|_p \leq \poly(k/\eps)\|U^*\|_p$$
By Lemma \ref{lemma:affine_embedding_lower_side_p_stable_median}, \ref{lemma:p_stable_matrix_median_preserves_norm} and \ref{lemma:p_stable_lp_norm_distortion} respectively, these each occur with probability $1 - O(1)$ (where the constant probability can be made as small as desired by increasing $r$ by a constant factor), and by a union bound they occur simultaneously with probability $1 - O(1)$. First, we use these events to examine the effect of rounding $SU^*$ to $M$, when the right factor is $V^*$, and more generally, when the $\ell_p$ norm of the right factor $V$ is at most $\poly(k/\eps)$.

\begin{claim}
\label{claim:connecting_rounded_median}
Suppose $\calE_1$, $\calE_2$ and $\calE_3$ hold, and suppose $V \in \R^{k \times d}$ such that $\|V\|_p^p \leq \poly(k/\eps)$. Then, 
$$\Big|\Big(\sum_{i = 1}^d \med(MV_i - SA_i)^p\Big)^{\frac{1}{p}} \big/ \med_p - \Big(\sum_{i = 1}^d \med(SU^*V_i - SA_i)^p\Big)^{\frac{1}{p}} \big/\med_p\Big| \leq \frac{1}{f\poly(k/\eps) \med_p}\|A\|_p$$
\end{claim}
\begin{proof}
Let $x \in \R^d$ be the vector whose $i^{th}$ coordinate is $\med(MV_i - SA_i)$, and let $y \in \R^d$ be the vector whose $i^{th}$ coordinate is $\med(SU^
*V_i - SA_i)$. Then, observe that
$$\Big|\Big(\sum_{i = 1}^d \med(MV_i - SA_i)^p\Big)^{\frac{1}{p}} - \Big(\sum_{i = 1}^d \med(SU^*V_i - SA_i)^p\Big)^{\frac{1}{p}}\Big| = |\|x\|_p - \|y\|_p|$$
Now, by the triangle inequality, this is at most $\|x - y\|_p$, which we can bound from above:
\begin{equation}
\begin{split}
\|x - y\|_p^p
& = \sum_{i = 1}^d |\med(MV_i - SA_i) - \med(SU^*V_i - SA_i)|^p \\
& \leq \sum_{i = 1}^d \|(M - SU^*)V_i\|_\infty^p \\
& \leq \sum_{i = 1}^d \|(M - SU^*)V_i\|_p^p \\
& = \|(M - SU^*)V\|_p^p
\end{split}
\end{equation}
Here, the first inequality is because $|\med(v_1 + v_2) - \med(v_1)| \leq \|v_2\|_\infty$ for any two vectors $v_1, v_2 \in \R^n$, and the second is because $\|v\|_\infty \leq \|v\|_p$ for any vector $v$.

Hence, 
$$\|x - y\|_p \leq \|(M - SU^*)V\|_p \leq \poly(k/\eps) \|M - SU^*\|_p \leq \frac{1}{f\poly(k/\eps)}\|A\|_p$$
Here, the second inequality holds because, even though $V$ is not necessarily a well-conditioned basis, for any vector $x \in \R^k$,
$$\|x^TV\|_p = \Big\|\sum_{i = 1}^k x_iV^i\Big\|_p \leq \sum_{i = 1}^k |x_i| \|V^i\|_p \leq k\|x\|_p\|V\|_p = \poly(k/\eps)\|x\|_p$$
where the first inequality is due to the triangle inequality, and the last equality is because $\|V\|_p \leq \poly(k/\eps)$. Hence, for any matrix $D \in \R^{n \times k}$,
$$\|DV\|_p^p = \sum_{i = 1}^n \|D^iV\|_p^p \leq \sum_{i = 1}^n \poly(k/\eps)\|D^i\|_p^p = \poly(k/\eps) \|D\|_p^p$$
and taking $p^{th}$ roots gives $\|DV\|_p \leq \poly(k/\eps) \|D\|_p$.
\end{proof}

\begin{claim}
Suppose $\calE_1$, $\calE_2$ and $\calE_3$ occur. Then, the linear program constructed in Algorithm \ref{algorithm:guessing_eps_approximation_general_lp} with variables in the set \textsc{Variables}, and constraints in the set \textsc{Constraints}, is feasible when $M = M_2$ is guessed.
\end{claim}
\begin{proof}
Consider a particular $i \in [d]$. Note that
\begin{equation}
\begin{split}
\frac{\med(MV_i^* - SA_i)}{\med_p}
& \leq \frac{\med(SU^*V_i^* - SA_i)}{\med_p} + \frac{\|(M - SU^*)V_i^*\|_\infty}{\med_p} \\
& \leq \frac{\med(SU^*V_i^* - SA_i)}{\med_p} + \frac{\|(M - SU^*)V_i^*\|_p}{\med_p} \\
& \leq \frac{\med(SU^*V_i^* - SA_i)}{\med_p} + \poly(k)\frac{\|(M - SU^*)\|_p}{\med_p} \\
& \leq \frac{\med(SU^*V_i^* - SA_i)}{\med_p} + \frac{1}{f\poly(k/\eps)}\|A\|_p \\
& \leq \frac{\Big(\sum_{j \in [d]} \med(SU^*V_j^* - SA_j)^p\Big)^{\frac{1}{p}}}{\med_p} + \frac{1}{f\poly(k/\eps)}\|A\|_p \\
& \leq (1 + \eps)OPT + \frac{1}{f\poly(k/\eps)}\|A\|_p \\
& \leq O(1)\|A\|_p
\end{split}
\end{equation}
The first inequality is true because for $v_1, v_2 \in \R^n$, $|\med(v_1 + v_2) - \med(v_1)| \leq \|v_2\|_\infty$, meaning $|\med(MV_i^* - SA_i) - \med(SU^*V_i^* - SA_i)| \leq \|(M - SU^*)V_i^*\|_\infty$. Here, the second inequality is by $\|x\|_p \geq \|x\|_\infty$, and the third is because $V^*$ is a $\poly(k)$ well-conditioned basis. The fourth inequality is because $\med_p$ is nonnegative and bounded away from $0$ by Lemma \ref{lemma:p_stable_median_omega_1}, and because $\|M - SU^*\|_p \leq \frac{1}{f\poly(k/\eps)}\|A\|_p$ by Equation \ref{eq:rounding_error}. Finally, the sixth inequality is because $\calE_2$ holds and the seventh is because $OPT \leq \|A\|_p$.

As a summary, we have shown that for each $i \in [d]$, $\med(MV_i^* - SA_i)/\med_p \leq O(1)\|A\|_p$. Hence, there are two cases --- either there exists $c_0 \in \textsc{CostBounds}$ such that $c_0 \geq \med(MV_i^* - SA_i)/\med_p \geq \frac{c_0}{1 + \eps}$, or $\med(MV_i^* - SA_i)/\med_p \leq (\frac{\eps^2}{f}) \|A\|_p/d^{\frac{1}{p}}$.

\begin{itemize}
    \item In the first case, we can assign $x_{i, c_0} = 1$ and $x_{i, c'} = 0$ for $c' \neq c_0$. Define $V_{i, c_0}$ as in Algorithm \ref{algorithm:guessing_eps_approximation_general_lp}, meaning it is equal to $\argmin_{V_i} \|V_i\|_p$ subject to the constraint that $\med(MV_i - SA_i)/\med_p \leq c_0$. Then, $\|V_{i, c_0}\|_p \leq \|V_i\|_p$ and moreover, since $\med(MV_i^* - SA_i)/\med_p \geq \frac{c_0}{1 + \eps}$, $\med(MV_{i, c_0} - SA_i)/\med_p \leq c_0 \leq (1 + O(\eps)) \med(MV_i^* - SA_i)/\med_p$.
    \item In the second case, we can assign $x_{i, c_0} = 1$ and $x_{i, c'} = 0$ for $c' \neq c_0$, where $c_0$ is now $(\frac{\eps^2}{f}) \|A\|_p/d^{\frac{1}{p}}$. Again, $V_{i, c_0}$ is equal to $\argmin_{V_i} \|V_i\|_p$ subject to the constraint that $\med(MV_i - SA_i)/\med_p \leq c_0$. Hence, since $\med(MV_i^* - SA_i)/\med_p \leq c_0$, $\|V_{i, c_0}\|_p \leq \|V_i^*\|_p$. Moreover, $\med(MV_{i, c_0} - SA_i)/\med_p \leq (\frac{\eps^2}{f})\|A\|_p/d^{\frac{1}{p}}$, which is sufficient for our purposes, as we see now.
\end{itemize}

We now conclude that this assignment to the $x_{i, c}$ satisfies all the constraints of the linear program. Clearly the constraints $x_{i, c} \geq 0$ (for all $i \in [d]$, $c \in \textsc{CostBounds}$) and $\sum_{c \in \textsc{CostBounds}} x_{i, c} = 1$ (for all $i \in [d]$) are satisfied (since for each $i \in [d]$, exactly one $x_{i, c}$ is set to $1$ and the rest are $0$). For each $i \in [d]$, let $c_i$ be the unique element of $\textsc{CostBounds}$ such that $x_{i, c_i} = 1$ according to our assignment. Then, for all $i \in [d]$, $\|V_{i, c_i}\|_p^p \leq \|V_i^*\|_p^p$ as mentioned above (where we wrote $\|V_{i, c_0}\|_p^p \leq \|V_i^*\|_p^p$), meaning
$$\sum_{i = 1}^d\sum_{c \in \textsc{CostBounds}} x_{i, c} \|V_{i, c}\|_p^p = \sum_{i = 1}^d \|V_{i, c_i}\|_p^p \leq \sum_{i = 1}^d \|V_i^*\|_p^p = \|V^*\|_p^p \leq kq^p = \poly(k)$$
where the last inequality is because each row of $V^*$ has norm at most $q$. In addition,
\begin{equation}
\begin{split}
\sum_{i = 1}^d\sum_{c \in \textsc{CostBounds}} x_{i, c}C_{i, c}^p
& = \sum_{i = 1}^d \med(MV_{i, c_i} - SA_i)^p/\med_p^p \\
& \leq \sum_{i = 1}^d \Big((1 + O(\eps))^p\med(MV_i^* - SA_i)^p/\med_p^p + (\eps^2/f)^p\|A\|_p^p/d\Big) \\
& = (1 + O(\eps))^p \frac{\sum_{i = 1}^d \med(MV_i^* - SA_i)^p}{\med_p^p} + \Big(\frac{\eps^2}{f}\Big)^p \|A\|_p^p
\end{split}
\end{equation}
By Claim \ref{claim:connecting_rounded_median}, this is at most
$$(1 + O(\eps))^p \Big( \Big(\sum_{i = 1}^d \med(SU^*V_i^* - SA_i)^p\Big)^{\frac{1}{p}} \big/\med_p + \frac{1}{f\poly(k/\eps) \med_p}\|A\|_p \Big)^p + \Big(\frac{\eps^2}{f}\Big)^p \|A\|_p^p$$
and this is in turn at most
$$(1 + O(\eps))^p \Big(\widehat{OPT} + \frac{1}{f\poly(k/\eps)}\|A\|_p \Big)^p + \Big(\frac{\eps^2}{f}\Big)^p \|A\|_p^p$$
Here, we used the fact that $\calE_2$ occurs, meaning $\Big(\sum_{i = 1}^d \med(SU^*V_i^* - SA_i)^p\Big)^{\frac{1}{p}} \big/\med_p \leq (1 + \eps)OPT$, $OPT \leq \frac{1}{1 - O(\eps)}\widehat{OPT} \leq (1 + O(\eps))\widehat{OPT}$, and moreover, $\med_p = \Omega(1)$ (regardless of $p$), meaning $\frac{1}{\med_p}$ is bounded above.

Hence, we can conclude that
$$\sum_{i = 1}^d \sum_{c \in \text{CostBounds}} x_{i, c}C_{i, c}^p \leq (1 + O(\eps))^p \Big(\widehat{OPT} + \frac{1}{f\poly(k/\eps)}\|A\|_p \Big)^p + \Big(\frac{\eps^2}{f}\Big)^p \|A\|_p^p$$
and we have shown that this assignment to the $x_{i, c}$ satisfies all the constraints of the LP.
\end{proof}

\begin{claim}
Suppose $\calE_1$, $\calE_2$ and $\calE_3$ occur. Then, when $M = M_2$ is guessed, with probability $1 - O(1)$, Algorithm \ref{algorithm:guessing_eps_approximation_general_lp} finds $V'$ such that $\|V'\|_p^p \leq \frac{2kq^p}{\eps}$ and $\sum_{i = 1}^d \med(MV_i' - SA_i)^p \Big/\med_p^p \leq (1 + 2\eps)\Delta$, where $\Delta$ is defined in Algorithm \ref{algorithm:guessing_eps_approximation_general_lp}.
\end{claim}

\begin{proof}

The proof is by Markov's inequality. First, by the previous claim, if $\calE_1$, $\calE_2$ and $\calE_3$ hold, then solving the LP within Algorithm \ref{algorithm:guessing_eps_approximation_general_lp} gives $x_{i, c}$ for $i \in [d]$ and $c \in \textsc{CostBounds}$ such that
$$\sum_{i \in [d], c \in \textsc{CostBounds}} x_{i, c} \|V_{i, c}\|_p^p \leq kq^p \,\, \text{ and } \sum_{i \in [d], c \in \textsc{CostBounds}} x_{i, c} C_{i, c}^p \leq \Delta$$
Now, for each column $i \in [d]$, suppose we sample a single $c_i \in \textsc{CostBounds}$ according to the distribution on $\textsc{CostBounds}$ given by the $x_{i, c}$, and we let $V_i' = V_{i, c_i}$. Denote this distribution by $\pi_i$. Then, conditioning on a \textit{fixed} value of the $p$-stable matrix $S$ such that $\calE_1$, $\calE_2$ and $\calE_3$ hold,
$$E_{c_i \sim \pi_i ,\, \forall i \in [d]}\Big[\|V'\|_p^p \mid S\Big] = \sum_{i = 1}^d \sum_{c \in \textsc{CostBounds}} x_{i, c} \|V_{i, c}\|_p^p \leq kq^p$$
The first equality is by the linearity of expectation. Note that prior to sampling $c_i$ for each $i \in [d]$, the only source of randomness is $S$, and the appropriate rounded matrix $M$, as well as the minimizers $V_{i, c}$ and their costs $C_{i, c}$ are all determined by $S$. The second inequality holds if $S$ is such that $\calE_1$, $\calE_2$ and $\calE_3$ are satisfied. By the same argument
$$E_{c_i \sim \pi_i ,\, \forall i \in [d]}\Big[\sum_{i = 1}^d \med(MV_i' - SA_i)^p / \med_p^p \mid S\Big] = \sum_{i \in [d], c \in \textsc{CostBounds}} x_{i, c} C_{i, c}^p \leq \Delta$$
Hence, letting $\calE = \calE_1 \cap \calE_2 \cap \calE_3$,
$$E_{c_i \sim \pi_i ,\, \forall i \in [d]}\Big[\|V'\|_p^p \mid \calE\Big] = \frac{\int_{\calE} E_{c_i \sim \pi ,\, \forall i \in [d]}\Big[\|V'\|_p^p \mid S\Big] p(S) dS}{P[\calE]} \leq \frac{\int_{\calE} kq^p \cdot p(S) dS}{P[\calE]} = kq^p$$
Here, the first equality is by the identity $E[X] = E[E[X \mid Y]]$, where $X = \|V'\|_p^p$ and $Y = S$, the expectation is taken using the probability measure obtained by conditioning on $\calE$, and $p(S)$ denotes the p.d.f. of the $p$-stable matrix $S$. The second is because the integral is over $S$ for which $\calE_1$, $\calE_2$ and $\calE_3$ hold. Note that we take the expectation conditioned on a fixed value of the continuous random variable $S$. To make this fully rigorous, we can assume the entries of $S$ are rounded to the nearest integer multiples of an arbitrarily small $\delta > 0$. This would then allow us to treat the integral as a discrete sum, and it is easy to see that rounding $S$ in this way does not affect the analysis --- see the remark below on how we can justify the above steps, such as taking the expectation conditioned on a fixed value of $S$, without discretizing $S$. Similarly,
\begin{equation}
\begin{split}
& E_{c_i \sim \pi_i ,\, \forall i \in [d]}\Big[\sum_{i = 1}^d \med(MV_i' - SA_i)^p / \med_p^p \mid \calE\Big] \\
& = \frac{\int_{\calE} E_{c_i \sim \pi ,\, \forall i \in [d]}\Big[\sum_{i = 1}^d \med(MV_i' - SA_i)^p / \med_p^p \mid S\Big] p(S) dS}{P[\calE]} \\
& \leq \frac{\int_{\calE} \Delta \cdot p(S) dS}{P[\calE]} \\
& = \Delta
\end{split}
\end{equation}
Now, by Markov's inequality, if we sample $c_i$ according to $\pi_i$ for $i \in [d]$, then
$$P\Big[\|V'\|_p^p \geq \frac{2kq^p}{\eps} \mid \calE \Big] \leq \frac{\eps}{2}$$
and
$$P\Big[\sum_{i = 1}^d \med(MV_i' - SA_i)^p / \med_p^p \geq (1 + 2\eps)\Delta \mid \calE\Big] \leq \frac{1}{1 + 2\eps} \leq 1 - \eps$$
where the second inequality holds because for $\eps \in (0, \frac{1}{2})$, $(1 + 2x)(1 - x) = 1 + x - 2x^2 = 1 + x(1 - 2x) \geq 1$, meaning $1 - x \geq \frac{1}{1 + 2x}$. Hence, by a union bound, $$P\Big[\sum_{i = 1}^d \med(MV_i' - SA_i)^p / \med_p^p \geq (1 + 2\eps)\Delta \text{ or } \|V'\|_p^p \geq \frac{2kq^p}{\eps} \mid \calE\Big] \leq 1 - \frac{\eps}{2}$$
Finally, if we repeatedly sample $c_i$ for $i \in [d]$ over the course of $\frac{10}{\eps}$ trials, then the probability that none of the $V'$ satisfy the desired properties is at most
$$\Big(1 - \frac{\eps}{2}\Big)^{10/\eps} = \Big(1 - \frac{\eps}{2}\Big)^{(2/\eps) \cdot 5} \leq \frac{1}{e^5} \leq \frac{1}{100}$$
This completes the proof of the claim --- if we condition on $\calE_1$, $\calE_2$ and $\calE_3$, then with probability $\frac{99}{100}$, we obtain $V'$ such that $\|V'\|_p^p \leq \frac{2kq^p}{\eps}$ and $$\sum_{i = 1}^d \med(MV_i' - SA_i)^p \Big/ \med_p^p \leq (1 + 2\eps)\Delta$$
\end{proof}

\begin{remark}
In the proof of the above lemma, we take the expectation of $\|V'\|_p^p$ and $\sum_{i = 1}^d \med(MV_i' - SA_i)^p/\med_p^p$ conditioned on a fixed value of the $p$-stable matrix $S$, and then use this to compute $E[\|V'\|_p^p \mid \calE]$ and $E[\sum_{i = 1}^d \med(MV_i' - SA_i)^p/\med_p^p \mid \calE]$, using the identity $E[X] = E[E[X \mid Y]]$. Note that the inner expectation here is conditioned on a \textit{fixed} value of the continuous random variable $S$. We can make this fully rigorous by discretizing $S$, i.e. rounding the entries of $S$ to the nearest integer multiple of $\delta > 0$, where $\delta$ is arbitrarily small, as mentioned above (in fact, it suffices to round the entries of $S$ to a multiple of $\frac{\eps}{\poly(nd)}$). Alternatively, we can formalize this using disintegrations, by applying Theorem 3 of \cite{cp97_conditional_expectation}, with the following parameters:
\begin{itemize}
    \item The sample space $\Omega$ is $\R^{\poly(k/\eps) \times n} \times \mathcal{D}$, where $\R^{\poly(k/\eps) \times n}$ is the set of all possible $S$ and $\mathcal{D}$ is the collection of all possible choices of $c_i$, for $i \in [d]$.
    \item $\lambda$, in the statement of Theorem 3 of \cite{cp97_conditional_expectation}, would be the product measure of the Lebesgue measure on $\R^{\poly(k/\eps) \times n}$, with the counting measure on $\mathcal{D}$.
    \item $\rho$ would be the probability distribution on $\Omega$ which is absolutely continuous with respect to $\lambda$, such that its density with respect to $\lambda$, at a point $(S, c)$ in the sample space, is $p(S) q_{c \mid S}(c \mid S)$, where $p(S)$ is the p.d.f. of a $p$-stable matrix $S$, and $q_{c \mid S}(c \mid S)$ is the distribution on the choices $c_i$, for $i \in [d]$, that is given by Algorithm \ref{algorithm:guessing_eps_approximation_general_lp}.  (Note that by Fubini's Theorem, $\rho(\Omega) = \int_\Omega p(S)q_{c \mid S}(c \mid S) d\lambda = 1$, so $\rho$ is a probability distribution.)
    \item $\mu$ is the probability distribution on $\R^{\poly(k/\eps) \times n}$, such that $S \sim \mu$ has i.i.d. standard $p$-stable entries.
\end{itemize}
Note that $\lambda$ is $\sigma$-finite since the Lebesgue measure is $\sigma$-finite, and $\rho$ and $\mu$ are finite measures since they are probability distributions. Thus, the relevant measures are all $\sigma$-finite, and $\rho$ is absolutely continuous w.r.t. $\lambda$, meaning the hypotheses of Theorem 3 of \cite{cp97_conditional_expectation} are satisfied. With these parameters, from part (v) of Theorem 3 in \cite{cp97_conditional_expectation}, we find that for $S \in \R^{\poly(k/\eps) \times n}$, the collection of probability measures on $\mathcal{D}$ given by
$$\rho_S(A) = \sum_{c \in A} q_{c \mid S}(c \mid S)$$
for $A \subset \mathcal{D}$ is a valid disintegration of $\rho$ with respect to $\mu$. Thus, using this disintegration, we can rigorously define $E[\|V'\|_p^p \mid S]$ and $E[\sum_{i = 1}^d \med(MV_i' - SA_i)/\med_p^p \mid S]$ for a fixed value of the $p$-stable matrix $S$, and use these to compute $E[\|V'\|_p^p \mid \calE]$ and $E[\sum_{i = 1}^d \med(MV_i' - SA_i)/\med_p^p \mid \calE]$, using part (iii) of Definition 1 of \cite{cp97_conditional_expectation} (which is equivalent to the identity $E[X] = E[E[X|Y]]$).
\end{remark}

Now, suppose $\calE_1$, $\calE_2$ and $\calE_3$ hold, and that we have obtained $V'$ with $$\|V'\|_p^p \leq \frac{O(kq^p)}{\eps} \text{ and } \sum_{i = 1}^d \med(MV_i' - SA_i)^p \Big/ \med_p^p \leq (1 + 2\eps)\Delta$$ 
By the previous claim, this occurs with constant probability. By Claim \ref{claim:connecting_rounded_median}, since $\|V'\|_p \leq \poly(k/\eps)$, we know that
\begin{equation}
\begin{split}
\Big(\sum_{i = 1}^d \med(MV_i' - SA_i)^p\Big)^{\frac{1}{p}} \Big/ \med_p
& \geq \Big(\sum_{i = 1}^d \med(SU^*V_i' - SA_i)^p \Big)^{\frac{1}{p}} \Big/ \med_p - \frac{1}{f\poly(k/\eps)\med_p} \|A\|_p \\
& \geq (1 - O(\eps)) \|U^*V' - A\|_p - \frac{1}{f\poly(k/\eps) \med_p} \|A\|_p \\
& \geq (1 - O(\eps)) \|U'V' - A\|_p - \frac{1}{f\poly(k/\eps) \med_p} \|A\|_p
\end{split}
\end{equation}
Here, the first inequality is by Claim \ref{claim:connecting_rounded_median}, the second is because we are conditioning on $\calE_1$ (meaning the median-based sketch and $S$ provide a one-sided embedding for $U^*$ and $A$), and the third inequality is because $U' = \argmin_U \|UV' - A\|_p$. Therefore,
\begin{equation}
\begin{split}
\|U'V' - A\|_p
& \leq (1 + O(\eps)) \Big(\sum_{i = 1}^d \med(MV_i' - SA_i)^p\Big)^{\frac{1}{p}} \Big/ \med_p + \frac{1}{f\poly(k/\eps)\med_p} \|A\|_p \\
& \leq (1 + O(\eps)) (1 + 2\eps)^{1/p} \Delta^{1/p} + \frac{1}{f\poly(k/\eps)\med_p}\|A\|_p \\
& \leq (1 + O(\eps)) \Delta^{1/p} + \frac{1}{f\poly(k/\eps)\med_p}\|A\|_p \\
& \leq (1 + O(\eps)) \Delta^{1/p} + \frac{1}{f\poly(k/\eps)}\|A\|_p
\end{split}
\end{equation}
where the last inequality is because $\med_p = \Omega(1)$ as $p$ ranges through $[1, 2]$. Finally, we can bound $\Delta^{1/p}$ from above by observing that the function $f: x \to |x|^{1/p}$ is subadditive:
\begin{equation}
\begin{split}
\Delta^{1/p}
& = \Big((1 + O(\eps))^p \Big(\widehat{OPT} + \frac{1}{f\poly(k/\eps)}\|A\|_p \Big)^p + \Big(\frac{\eps^2}{f}\Big)^p \|A\|_p^p\Big)^{1/p} \\
& \leq \Big((1 + O(\eps))^p \Big(\widehat{OPT} + \frac{1}{f\poly(k/\eps)}\|A\|_p \Big)^p\Big)^{1/p} + \Big(\Big(\frac{\eps^2}{f}\Big)^p \|A\|_p^p\Big)^{1/p} \\
& = (1 + O(\eps)) \Big(\widehat{OPT} + \frac{1}{f\poly(k/\eps)}\|A\|_p \Big) + \frac{\eps^2}{f} \|A\|_p \\
& \leq (1 + O(\eps)) \Big(OPT + \frac{1}{f\poly(k/\eps)}\|A\|_p \Big) + \frac{\eps^2}{f} \|A\|_p
\end{split}
\end{equation}
where the first inequality holds because $(|x| + |y|)^{1/p} \leq |x|^{1/p} + |y|^{1/p}$, and the last is because $\widehat{OPT} \leq (1 + O(\eps))OPT$. In summary,
\begin{equation}
\begin{split}
\|U'V' - A\|_p
& \leq (1 + O(\eps))\Delta^{1/p} + \frac{1}{f\poly(k/\eps)} \|A\|_p \\
& = (1 + O(\eps))OPT + O\Big(\frac{\eps^2}{f}\Big) \|A\|_p
\end{split}
\end{equation}

Finally, we analyze the running time of Algorithm \ref{algorithm:guessing_eps_approximation_general_lp}.

\begin{claim}
The running time of Algorithm \ref{algorithm:guessing_eps_approximation_general_lp} is at most $f^{O(rk)}  + 2^{O(rk\log(k/\eps))} + \poly(fnd/\eps)$, where $r$, the number of rows in the $p$-stable matrix $S$, is at most $O(\max(k/\eps^6 \log(k/\eps), 1/\eps^9))$
\end{claim}
\begin{proof}
First, let us find the number of matrices in $|\calC|$. Note that each $M \in \calC$ has $rk$ entries (recall that $r$ is the number of rows of $S$). For each entry, there are $|\calI|$ choices. The magnitude of the largest possible guess is $\poly(k/\eps)\|A\|_p$, while that of the smallest possible guess is $\frac{1}{f\poly(k/\eps)}\|A\|_p$. Therefore, the number of guesses for each entry is $\frac{O(\log (fk/\eps))}{\log(1 + \frac{1}{f\poly(k/\eps)})}$. Since for $x < 1$, $\log(1 + x) \geq \frac{x}{2}$, the number of guesses per each entry of $M$ is in fact $O(f\poly(k/\eps) \log(fk/\eps))$. In summary,
$$|\calC| = (f\poly(k/\eps)\log(fk/\eps))^{rk} = f^{O(rk)} 2^{O(rk\log(k/\eps))}$$
Let us now calculate the running time needed for each guess. The size of $\textsc{CostBounds}$ is
$$\log\Big(\frac{fd}{\eps}\Big) \cdot \frac{1}{\log(1 + \eps)} \leq O\Big(\frac{1}{\eps} \log\Big(\frac{fd}{\eps}\Big)\Big)$$
For each $i \in [d]$ and $c \in \textsc{CostBounds}$, we solve the problem of minimizing $\|V_i\|_p$ subject to the constraint that $\med(MV_i - SA_i)/\med_p \leq c$. Computing $SA_i$ takes $O(nr)$ time. If $V_{i, c, best}$ is the solution to this problem, then there are at most $r!$ possible orderings for the coordinates of $MV_{i, c, best} - SA_i$. Therefore, we can make this optimization problem into a convex program by trying all orderings of the coordinates of $MV_i - SA_i$ (meaning that the orderings are added in the form of $\poly(k/\eps)$ additional constraints). Once we do this, $\med(MV_i - SA_i) /\med_p \leq c$ becomes a linear constraint. Hence, we solve $r!$ distinct convex programs. Since $M$ is an $r \times k$ matrix, and $SA_i$ is an $r$-dimensional vector, finding $V_{i, c, best}$ takes at most $r! \cdot \poly(r)$ running time. In summary, the total time taken to find $V_{i, c}$ and $C_{i, c}$ for all $i \in [d]$ and $c \in \textsc{CostBounds}$ is
\begin{equation}
\begin{split}
d \cdot |\textsc{CostBounds}| \cdot r! \cdot \poly(r) + O(nr)
& \leq O\Big(\frac{d}{\eps}\log\Big(\frac{fd}{\eps}\Big)\Big) \cdot r^{O(r)} + O(nr) \\
& \leq \poly(fnd/\eps) 2^{O(r\log r)}
\end{split}
\end{equation}
Next comes the running time for the convex program used to find the $x_{i, c}$. As calculated above, the number of variables is $d \cdot |\textsc{CostBounds}| = O\Big(\frac{d}{\eps}\log\Big(\frac{fd}{\eps}\Big)\Big)$ and the number of constraints is $d|\textsc{CostBounds}| + d + 2 = O\Big(\frac{d}{\eps}\log\Big(\frac{fd}{\eps}\Big)\Big)$. Therefore, the convex program can be solved in $\poly(fd/\eps)$ time.

Finally, we calculate the running time of the stage where we sample $c_i$ for every $i \in [d]$ in order to find a right factor $V'$. The time needed to sample $c_i$ is $d \cdot |\textsc{CostBounds}| = O\Big(\frac{d}{\eps}\log\Big(\frac{fd}{\eps}\Big)\Big)$, and for each sample of the $c_i$, computing the norm of the new $V'$ takes $O(nk)$ time, and computing $\sum_{i \in [d]} \med(MV_i' - SA_i)^p / \med_p^p$ takes $n \cdot \poly(k)$ time. The number of samples is $O(1/\eps)$, meaning the total running time of this stage is $\poly(fnd/\eps)$. After $V'$ has been found, the running time needed to find $U'$ by solving $\min_U \|UV' - A\|_p$ using convex programming is $\poly(nd)$. In summary, the total running time of Algorithm \ref{algorithm:guessing_eps_approximation_general_lp} is
$$f^{O(rk)}2^{O(rk\log(k/\eps))}\Big(\poly(fnd/\eps) 2^{O(r \log r)} + \poly(fnd/\eps)\Big)$$
and this is at most $f^{O(rk)}2^{O(rk\log(k/\eps))}\poly(fnd/\eps)$. The runtime in the theorem statement is due to the inequality $abc \leq \frac{a^3 + b^3 + c^3}{3}$, for $a, b, c \geq 0$.
\end{proof}
\end{proof}

We use this to obtain a $(1 + \eps)$-approximation algorithm with bicriteria rank $3k$.

\begin{theorem}[Correctness and Running Time of Algorithm \ref{algorithm:get_eps_approximation_after_initialization_general_lp}]
Let $A \in \R^{n \times d}$, $k \in \N$, and $\eps \in (0, c)$ where $c$ is a sufficiently small absolute constant. Then, Algorithm \ref{algorithm:get_eps_approximation_after_initialization_general_lp}, with these inputs, returns $\widehat{A} \in \R^{n \times d}$ such that $\widehat{A}$ has rank $3k$ and
$$\|\widehat{A} - A\|_p \leq (1 + O(\eps))\min_{A_k \text{ rank }k} \|A - A_k\|_p$$
The running time of Algorithm \ref{algorithm:get_eps_approximation_after_initialization_general_lp} is $2^{O(rk \log(k/\eps))} + \poly(nd/\eps)$, where $r$ is the same as in the statement of Theorem \ref{theorem:correctness_of_guessing_eps_approx}.
\end{theorem}

\begin{proof}
Let $A_k$ be the optimal rank-$k$ approximation for $A$, and suppose $$(1 - O(\eps))OPT_{C, 2k} \leq \widehat{OPT} \leq (1 + O(\eps)) OPT_{C, 2k}$$ where $OPT_{C, 2k}$ is the error from the optimal rank-$2k$ approximation for $C$ (as the algorithm will indeed guess such an $\widehat{OPT}$ at some point). If this is the case, then when we call $$\textsc{GuessingAdditiveEpsApproximation}(C, 2k, \eps, f, \widehat{OPT}, p)$$ we obtain $U, V$ such that
\begin{equation}
\begin{split}
\|UV - C\|_p
& \leq (1 + O(\eps)) \min_{C_{2k} \text{ rank } 2k} \|C_{2k} - C\|_p + O\Big(\frac{\eps}{f}\Big) \|C\|_p \\
& \leq (1 + O(\eps)) \|(A_k - B) - (A - B)\|_p + O\Big(\frac{\eps}{\poly(k)}\Big) \cdot \poly(k)\|A - A_k\|_p \\
& \leq (1 + O(\eps)) \|A_k - A\|_p + O(\eps) \|A - A_k\|_p \\
& \leq (1 + O(\eps)) \|A_k - A\|_p
\end{split}
\end{equation}
where the second inequality is because $A_k - B$ has rank at most $2k$, and because $B$ is the result of a $\poly(k)$-approximation for $A$. Since $\|UV - C\|_p = \|UV - (A - B)\|_p = \|(UV + B) - A\|_p$, this completes the proof of correctness for Algorithm \ref{algorithm:get_eps_approximation_after_initialization_general_lp}. Note that $UV + B$ has rank at most $3k$ because $UV$ has rank at most $2k$ and $B$ has rank at most $k$.

Now, we analyze the running time of Algorithm \ref{algorithm:get_eps_approximation_after_initialization_general_lp}. First, the running time of Algorithm \ref{algorithm:poly_k_rank_exactly_k} (used to find $B$) is $\poly(nd)$. The running time of $\textsc{GuessingAdditiveEpsApproximation}$ is at most $f^{O(rk)} + 2^{O(rk \log(k/\eps))} + \poly(fnd/\eps)$, and since $f = \poly(k)$, this is
$$2^{O(rk\log k)} + 2^{O(rk\log(k/\eps))} + \poly(nd/\eps) = \boxed{2^{O(rk \log(k/\eps))} + \poly(nd/\eps)}$$
The number of times Algorithm \ref{algorithm:get_eps_approximation_after_initialization_general_lp} is called is $O((\log nd)/\eps)$, meaning the overall running time of Algorithm \ref{algorithm:get_eps_approximation_after_initialization_general_lp} is also the above. All other steps in Algorithm \ref{algorithm:get_eps_approximation_after_initialization_general_lp} can be done in polynomial time.
\end{proof}

\newpage
\appendix
\section{Hardness for $\ell_p$ Low Rank Approximation with Additive Error, based on \cite{ptas_for_lplra}} \label{appendix:additive_error_hardness}

\subsection{Background: Small Set Expansion Hypothesis}

The hardness proof in \cite{ptas_for_lplra} proceeds by a reduction from the Small Set Expansion problem --- our presentation of this problem follows that of \cite{ptas_for_lplra}.

\begin{problem} [Small Set Expansion Problem - As Presented in \cite{ptas_for_lplra}]
\label{problem:small_set_expansion_problem}
Let $G = (V, E)$ be a regular graph. For any subset $S \subset V$, the measure of $S$ is defined to be $\mu(S) := |S|/|V|$. The distribution $G(S)$ over the vertices of $G$ is generated as follows: first a uniformly random vertex $x \in S$ is selected, then a uniformly random neighbor $y$ of $x$ is selected (as a sample of $G(S)$). For $S \subset V$, the expansion of $S$ is defined to be $\Phi_G(S) := \text{Pr}_{y \sim G(S)}[y \not\in S]$. Finally, for $\delta \in (0, 1)$, $\Phi_G(\delta) := \min_{S \subset V \mid \mu(S) \leq \delta} \Phi_G(S)$.

The Small Set Expansion Problem is as follows: given a graph $G = (V, E)$ and $\eps, \delta \in (0, 1)$, the goal is to decide whether $\Phi_G(\delta) \leq \eps$ or $\Phi_G(\delta) \geq 1 - \eps$.
\end{problem}

In other words, $\Phi_G(S)$ is the proportion of neighbors that $S$ has, which do not belong to $S$ --- this can be considered as the ``expansion'' of $S$ --- and $\Phi_G(\delta)$ is the smallest expansion among all subsets of $V$ which have at most $\delta |V|$ vertices. The Small Set Expansion Hypothesis is as follows:

\begin{conjecture} [Small Set Expansion Hypothesis --- Conjecture 1.3 of \cite{rs_2010_sseh_original}]
For any fixed $\eps \in (0, 1)$, there exists $\delta \in (0, 1)$ such that it is NP-hard to decide whether $\Phi_G(\delta) \leq \eps$ or $\Phi_G(\delta) \geq 1 - \eps$.
\end{conjecture}

\subsection{Background: Hardness Proof from \cite{ptas_for_lplra} --- $\ell_p$ Low Rank Approximation for $p \in (1, 2)$}

Now, we summarize the reduction due to \cite{ptas_for_lplra} from the Small Set Expansion Problem to $\ell_p$ low rank approximation for $p \in (1, 2)$. The reduction in Section 5 of \cite{ptas_for_lplra} shows that given $k \in \N$ and $p \in (1, 2)$, and given a matrix $A \in \R^{n \times d}$, it is NP-hard to find an $O(1)$-approximation for the error from the best rank-$k$ approximation for $A$ in the $\ell_p$-norm. Ultimately, \cite{ptas_for_lplra} reduces from Problem \ref{problem:small_set_expansion_problem} to finding the best rank-$(k - 1)$ approximation for a $k \times k$ matrix. The first step is to reduce from Problem \ref{problem:small_set_expansion_problem} to computing the $2\to p^*$ norm, where $p^*$ is the H{\"o}lder conjugate of $p$:

\begin{theorem}[Theorem 2.4 of \cite{bbhksz_2012_stoc_original_2_to_q_norm} for $q \in 2\Z \setminus \{2\}$, and Theorem 21 of \cite{ptas_for_lplra} for general $q \in (2, \infty)$]
\label{theorem:2_to_q_norm_expansion_reduction}
Let $G$ be a regular graph, $\lambda \in (0, 1)$ and $q \in (2, \infty)$. Let $M$ be the normalized adjacency matrix of $G$, and let $V_{\geq \lambda}(G)$ be the subspace spanned by the eigenvectors of $M$ with eigenvalue at least $\lambda$. Finally, let $P_{\geq \lambda}(G)$ be the orthogonal projection matrix onto $V_{\geq \lambda}(G)$. Then,
\begin{itemize}
    \item For all $\delta > 0$, $\eps > 0$, $\|P_{\geq \lambda}(G)\|_{2 \to q} \leq \eps/\delta^{(q - 2)/2q}$ implies that $\Phi_G(\delta) \geq 1 - \lambda - \eps^2$.
    \item There is a constant $a = a(q)$ such that for all $\delta > 0$, $\Phi_G(\delta) > 1 - a\lambda^{2q}$ implies $\|P_{\geq \lambda}(G)\|_{2\to q} \leq 2/\sqrt{\delta}$.
\end{itemize}
\end{theorem}

This implies that the $2 \to q$ norm is hard to approximate --- for completeness, we include the proof of this from \cite{ptas_for_lplra}. In the next subsection, we show that this proof can be modified to obtain hardness for a multiplicative $O(1)$-approximation for $\ell_p$ low rank approximation with additive $\frac{1}{2^{\poly(k)}}\|A\|_p$ error.

\begin{theorem}[Theorem 7 of \cite{ptas_for_lplra}]
\label{theorem:2_to_q_norm_hard}
Assuming the Small Set Expansion Hypothesis, for any $q \in (2, \infty)$, and $r > 1$, it is NP-hard to approximate the $\|\cdot \|_{2 \to q}$ norm within a factor $r$.
\end{theorem}
\begin{proof}
(From \cite{ptas_for_lplra}, Page 46). Using \cite{rst_2010}, the Small Set Expansion Hypothesis implies that for any sufficiently small numbers $0 < \delta \leq \delta'$, there is no polynomial time algorithm that can distinguish between the following cases for a given graph $G$:
\begin{itemize}
    \item Yes case: $\Phi_G(\delta) < 0.1$
    \item No case: $\Phi_G(\delta') > 1 - 2^{-a'\log(1/\delta')}$
\end{itemize}
In particular, for all $\eta > 0$, if we let $\delta' = \delta^{(q - 2)/8q}$ and make $\delta$ small enough, then in the No case $\Phi_G(\delta^{(q - 2)/8q}) > 1 - \eta$. (Since $q > 2$, $\delta' \to 0$ as $\delta \to 0$.)

Using Theorem \ref{theorem:2_to_q_norm_expansion_reduction}, in the Yes case we know $\|P_{\geq 1/2}(G)\|_{2 \to q} \geq 1/(10\delta^{(q - 2)/2q})$, while in the No case, if we choose $\delta$ sufficiently small so that $\eta$ is smaller than $a(1/2)^{2q}$, then we know that $\|P_{\geq 1/2}(G)\|_{2 \to q} \leq 2/\sqrt{\delta'} = 2/\delta^{(q - 2)/4q}$. The gap between the Yes case and the No case is at least $\delta^{-(q - 2)/4q}/20$, which goes to $\infty$ as $\delta$ decreases.
\end{proof}

This proof shows that computing $\|P_{\geq 1/2}(G)\|_{2 \to p^*}$ within a constant factor is sufficient to decide the Small Set Expansion Problem. Now (if $G$ is assumed to be a graph with $k$ vertices) the following steps are used to reduce the problem of computing the $2 \to p^*$ norm of the $k \times k$ matrix $P_{\geq 1/2}(G)$ to the problem of rank-$(k - 1)$ $\ell_p$ low rank approximation. For convenience, let $P := P_{\geq 1/2}(G)$.

First, computing the $2 \to p^*$ norm of $P$ is equivalent to computing the $2 \to p^*$ norm of some invertible matrix $P_1$ which can be constructed in $\poly(k)$ time:

\begin{lemma} [Claim 14 of \cite{ptas_for_lplra}]
\label{lemma:reduction_step_1}
Let $A$ be a nonzero $n \times d$ matrix. For any $p, q \in (1, \infty)$ and any $\eps > 0$, there is an invertible and polynomial-time computable $\max(n, d) \times \max(n, d)$ matrix $B$ such that $(1 - \eps) \|A\|_{p \to q} \leq \|B\|_{p \to q} \leq (1 + \eps) \|A\|_{p \to q}$.
\end{lemma}

The proof of the above lemma proceeds as follows: first, $A$ is made into a square matrix by adding rows/columns of zeros --- then, $\frac{\eps M}{\|I\|_{p \to q}}$ is added to every diagonal entry to make it invertible, where $M$ is the absolute value of the largest entry of $A$. Hence, in our case, we will add $\frac{\eps M}{\|I\|_{2 \to p^*}}$ to every diagonal entry of $P$ to obtain $P_1$. The next step, after obtaining the invertible matrix $P_1$, is to compute a matrix $P_2$ such that finding $\|P_2\|_{p \to p^*}$ allows us to find $\|P_1\|_{2 \to p^*}$:

\begin{lemma} [Claim 13 of \cite{ptas_for_lplra}]
\label{lemma:reduction_step_2}
For any $A \in \R^{n \times d}$ and $p \in (1, \infty)$, if $p^*$ is its H{\"o}lder conjugate, then $\|AA^T\|_{p \to p^*} = \|A\|_{2 \to p^*}^2$.
\end{lemma}

Hence, by this claim, the appropriate $P_2$ will simply be $P_1P_1^T$. Finally, it is useful to reduce this problem to computing $\min_{p^* \to p}(P_3)$ of some well-chosen matrix $P_3$, since computing $\min_{p^* \to p}$ was shown by \cite{ptas_for_lplra} to be equivalent to rank-$(k - 1)$ $\ell_p$-low rank approximation:

\begin{lemma} [Fact 4 of \cite{ptas_for_lplra}]
\label{lemma:reduction_step_3}
For $p, q \in (1, \infty)$, if $A$ is an invertible matrix, then $\min_{p \to q}(A^{-1}) = (\|A\|_{q \to p})^{-1}$.
\end{lemma}

\begin{lemma} [Lemma 1/Lemma 27 of \cite{ptas_for_lplra} - Equivalence of $\min_{p^* \to p}$ and Rank-$(k - 1)$ $\ell_p$ Low Rank Approximation]
\label{lemma:equivalence_of_lp_norm_lra_min_p_to_p}
Let $p \in (1, \infty)$ and let $p^*$ be the H{\"o}lder conjugate of $p$. Let $A \in \R^{n \times d}$ with $n \geq d$ and $k = d - 1$. Then,
$$\min_{U \in \R^{n \times k}, V \in \R^{k \times d}} \|UV - A\|_p = \min_{x \in \R^d, \|x\|_{p^*} = 1} \|Ax\|_p$$
\end{lemma}

Hence, $P_3$ is in fact $P_2^{-1}$, and the $2 \to p^*$ norm of $P_{\geq 1/2}(G)$ can be approximated up to a constant factor if $\min_{p^* \to p}(P_3) = \min_{M \text{ rank }k - 1} \|M - P_3\|_p$ can be approximated up to a constant factor. In \cite{ptas_for_lplra}, it is also mentioned that assuming the Exponential-Time Hypothesis \cite{ip_2001_eth_original} together with SSEH implies that the running time required to obtain an $O(1)$-approximation is $2^{k^{\Omega(1)}}$ --- this is also true for our hardness result, as we mention in the proof of Theorem \ref{theorem:hardness_lp_additive_error}.

\subsection{Deterministic Embedding from $\ell_p^n$ to $\ell_1^{n^{O(\log n)}}$}

We can construct a matrix that can be used to embed $\ell_p^n$ into $\ell_1^{r}$, where $r = n^{O(\log n)}$ --- the matrix can be constructed deterministically, in $n^{O(\log n)}$ time. This embedding is then used to reduce $\ell_p$ low rank approximation to $\ell_{1, p}$ low rank approximation, and then to reduce $\ell_{1, p}$ low rank approximation to constrained $\ell_1$ low rank approximation. Throughout this section, we will use the following notation: for two functions $f, g > 0$, $f = O_p(g)$ if $f \leq C_p g$ for some constant $C_p$ depending on $p$. We also define $\Omega_p$ and $\Theta_p$ similarly. In addition, using the notation of \cite{knw_10_approximating_lp_norm_p_stable}, we say that $x \approx_\eps y$ for two $x, y \in \R$ if $|x - y| \leq \eps$.

\begin{theorem}[Deterministic Embedding of $\ell_p^n$ into $\ell_1^{r}$]
\label{theorem:embedding_lp_into_l1_deterministic}
Let $n \in \N$, and $p \in (1, 2)$. Then, there exists a matrix $R \in \R^{r \times n}$, where $r = n^{O(\log n)}$, such that for all $x \in \R^n$,
$$\Omega_p(1) \|x\|_p \leq \|Rx\|_1 \leq O_p(1) \|x\|_p$$
$R$ can be constructed deterministically in $n^{O(\log n)}$ time.
\end{theorem}

Our proof of Theorem \ref{theorem:embedding_lp_into_l1_deterministic} will be based on the following observation from \cite{matousek_metric_embeddings}:

\begin{fact}[Observation 2.7.2 of \cite{matousek_metric_embeddings}]
\label{fact:ell1_embedding_expected_value_observation}
Let $R_1, R_2, \ldots, R_n$ be real random variables on a probability space that has $k$ elements $\omega_1, \omega_2, \ldots, \omega_k$, and let $A \in \R^{k \times n}$ such that $A_{i, j} = \Prob{\omega_i} R_j(\omega_i)$. For $x \in \R^n$ let us set $X := \sum_{j = 1}^n R_j x_j$. Then, $E[ |X| ] = \|Ax\|_1$.
\end{fact}

Our embedding matrix $R$ will be constructed in the same way $A$ is constructed in the above observation. The random variables $R_i$ will be standard $p$-stable random variables (discretized appropriately) and will be selected so that they are $O(\log n)$-wise independent, rather than fully independent, so that the sample space only has size $n^{O(\log n)}$. We use lemmas from \cite{knw_10_approximating_lp_norm_p_stable} which are useful when dealing with identically distributed random variables that have limited independence. The first of these lemmas gives smooth approximations to indicator functions of an interval:

\begin{lemma}[Lemma 2.5 of \cite{knw_10_approximating_lp_norm_p_stable}]
\label{lemma:smooth_approximations_to_interval}
There exist constants $c', \eps_0 > 0$ such that for all $c > 0$ and $0 < \eps < \eps_0$, and for all $[a, b] \subset \R$, there exists a function $J_{[a, b]}^c: \R \to \R$ satisfying:
\begin{itemize}
    \item $\|(J_{[a, b]}^c)^{(\ell)} \|_\infty = O(c^\ell)$ for all $\ell \geq 0$.
    \item For all $x$ such that $a, b \not\in [x - \eps, x + \eps]$, and as long as $c > c'\eps^{-1} \log^3(1/\eps)$, $|J^c_{[a, b]}(x) - I_{[a, b]}(x)| < \eps$.
\end{itemize}
Here, in the first item above, for a function $f$, $f^{(\ell)}$ denotes the $\ell^{th}$ derivative of $f$.
\end{lemma}

The following lemma from \cite{knw_10_approximating_lp_norm_p_stable} allows us to show that, if $X$ is a linear combination of $n$ fully independent $p$-stable random variables, $Y$ is a linear combination of $n$ $p$-stable random variables which are only $O(1)$-wise independent, and $f$ is a smooth function, then $E[f(X)]$ and $E[f(Y)]$ are not too far apart. When we apply this lemma, $f$ will be $J^c_{[a, b]}$ for some $c$ and $\eps$.

\begin{lemma}[Lemma 2.2 of \cite{knw_10_approximating_lp_norm_p_stable}]
\label{lemma:smooth_function_removing_full_independence}
There exists an $\eps_0 > 0$ such that the following holds. Let $n$ be a positive integer and $0 < \eps < \eps_0$, $0 < p < 2$ be given. Let $f: \R \to \R$ satisfy $\|f^{(\ell)}\|_\infty = O(\alpha^\ell)$ for all $\ell \geq 0$, for some $\alpha$ satisfying $\alpha^p \geq \log(1/\eps)$. Let $k = \alpha^p$. Let $a \in \R^n$ satisfy $\|a\|_p = O(1)$. Let $X_i$ be a $3Ck$ independent family of $p$-stable random variables for $C$ a suitably large even constant. Let $Y_i$ be a fully independent family of $p$-stable random variables. Let $X = \sum_i a_iX_i$ and $Y = \sum_i a_iY_i$. Then, $E[f(X)] = E[f(Y)] + O(\eps)$.
\end{lemma}

\begin{remark}
Note that this lemma only requires that the $X_i$ are $3Ck$-wise independent for some constant $C$ and some $k = O(\log 1/\eps)$ --- since we will anyways choose $\eps = \Theta(1)$, this does not affect the amount of independence that we need. 
\end{remark}

It is also useful to note that $p$-stable random variables have finite expectation for $p \in (1, 2)$:

\begin{lemma}
Let $p \in (1, 2)$ and suppose $X$ is a standard $p$-stable random variable. Then, $E[|X|] = O_p(1)$.
\end{lemma}
\begin{proof}
Note that the p.d.f. of $X$, which we denote $p_X(x)$, is $\Theta(1/x^{p + 1})$ for large $x$:

\begin{theorem}[Theorem 1.12 of \cite{nolan_18_book_chapter_1}]
Let $X$ be a $p$-stable random variable with p.d.f. $p_X(x)$. Then,
$$\lim_{x \to \infty} \frac{p_X(x)}{(pC_p)/x^{p + 1}} = 1$$
for some constant $C_p$ depending on $p$.
\end{theorem}

\begin{remark}
Theorem 1.12 in \cite{nolan_18_book_chapter_1} is more general and also applies when $X$ is not symmetric.
\end{remark}

Hence, $p_X(x) \leq O_p(\frac{1}{x^{p + 1}})$ (for some constant $K$ and all $x$ with $|x| \geq K$), meaning
\begin{equation}
\begin{split}
E[|X|]
& \leq \int_{-K}^{K} |x| p_X(x) dx + 2 \int_{K}^{\infty} xp_X(x) dx \\
& \leq 2K + 2 \int_K^\infty x \cdot O\Big(\frac{1}{x^{p + 1}}\Big) dx \\
& \leq 2K + 2 \int_K^\infty O\Big(\frac{1}{x^p}\Big) dx \\
& \leq O_p(1)
\end{split}
\end{equation}
where the last inequality holds because the integral $\int_K^\infty \frac{1}{x^p} dx$ converges, since $p > 1$.
\end{proof}

Also, we note that linear combinations of $O(1)$-wise independent $p$-stable random variables are anti-concentrated:

\begin{lemma}[Anti-Concentration of $O(1)$-wise Independent $p$-stable Random Variables (Due to \cite{knw_10_approximating_lp_norm_p_stable})]
\label{lemma:anti_concentration_of_p_stables}
Let $R_1, R_2, \ldots, R_n$ be an $O(\log(1/\eps))$-wise independent family of $p$-stable random variables, and $x \in \R^n$ with $\|x\|_p = O(1)$. Then, for any $t \in \R$,
$$\ProbBig{\Big|\Big(\sum_{i = 1}^n R_i x_i\Big) - t\Big| \leq \eps} \leq O(\eps)$$
In particular, if $X$ is a standard $p$-stable random variable, then $\Prob{X \in [t - \eps, t + \eps]} \leq O(\eps)$ for $\eps \in (0, \eps_0)$ and any $t \in \R$.
\end{lemma}
\begin{proof}
Let $R = \sum_{i = 1}^n R_ix_i$. The main tool in showing this lemma is the following result from \cite{knw_10_approximating_lp_norm_p_stable}:

\begin{lemma}[(Shown in Section A.4 of \cite{knw_10_approximating_lp_norm_p_stable})]
\label{lemma:anti_concentration_functions}
For some constant $\eps_0$ and any $t \in \R$ and $\eps \in (0, \eps_0)$, there exists a nonnegative function $f_{t, \eps} : \R \to \R$ such that:
\begin{itemize}
    \item For some $\alpha = O(1/\eps)$, $\|f_{t, \eps}^{(\ell)}\|_\infty = O(\alpha^\ell)$ for all $\ell \geq 0$.
    \item If $X$ is a standard $p$-stable random variable, then $E[f_{t, \eps}(X)] = O(\eps)$.
    \item $f_{t, \eps}(t + \eps) = \Omega(1)$
    \item As $|x - t|$ increases, $f_{t, \eps}(x)$ decreases.
\end{itemize}
\end{lemma}

The rest of the proof follows that of a similar claim within Theorem 2.1 of \cite{knw_10_approximating_lp_norm_p_stable}. Let $S_1, S_2, \ldots, S_n$ be a fully independent family of $p$-stable random variables. Note that $S = \sum_{i = 1}^n S_ix_i$ is a standard $p$-stable random variable, and by Lemma \ref{lemma:anti_concentration_functions}, $E[f_{t, \eps}(S)] = O(\eps)$. Now, by Lemma \ref{lemma:smooth_function_removing_full_independence}, since the $R_i$ are $O(\log(1/\eps))$-wise independent,
$$E[f_{t, \eps}(R)] \leq E[ f_{t, \eps}(S)] + O(\eps) \leq O(\eps)$$
where $E[f_{t, \eps}(S)] \leq O(\eps)$ by Lemma \ref{lemma:anti_concentration_functions}. Finally, for all $r \in [t - \eps, t + \eps]$, $f_{t, \eps}(r) \geq f_{t, \eps}(t + \eps)$, since $f_{t, \eps}(x)$ is decreasing in $|x - t|$. Hence,
$$E[f_{t, \eps}(R)] \geq \int_{t - \eps}^{t + \eps} f_{t, \eps}(r) p_R(r) dr \geq f_{t, \eps}(t + \eps) \ProbBig{R \in [t - \eps, t + \eps]}$$
where $p_R$ is the p.d.f. of $R$. Since $f_{t, \eps}(t + \eps) = \Omega(1)$, combining this with the previous chain of inequalities gives
$$O(\eps) \geq \Omega(1) \ProbBig{R \in [t - \eps, t + \eps]}$$
and this gives the desired result.
\end{proof}

We will study $E[|\langle Z, x \rangle|]$ for a fixed $x \in \R^n$, where $Z$ is a vector whose entries are $O(\log n)$-wise independent $p$-stable random variables, which are truncated at $\poly(n)$, i.e. $Z_i$ follows a $p$-stable distribution conditioned on $|Z_i| \leq \poly(n)$. To analyze the tails of $|\langle Z, x \rangle|$, we use the following lemma from \cite{knpw11_pth_moment}:

\begin{lemma}[Lemma 23 of \cite{knpw11_pth_moment}]
\label{lemma:tail_probability_small_independence_p_stable}
Let $p \in (0, 2]$. Suppose $x \in \R^n$, $\|x\|_p = 1$, $0 < \eps < 1$ is given, and $R_1, \ldots, R_n$ are $k$-wise independent $p$-stable random variables for $k \geq 2$. Let $Q$ be a standard $p$-stable random variable. Then, for all $t \geq 0$, $R = \sum_{i = 1}^n R_ix_i$ satisfies
$$\Big| \Prob{|Q| \geq t} - \Prob{|R| \geq t} \Big| = O\Big( \frac{k^{-1/p}}{1 + t^{p + 1}} + \frac{k^{-2/p}}{1 + t^2} + 2^{-\Omega(k)}\Big)$$
\end{lemma}

Now we show that $E[|\langle Z, x \rangle|] = \Theta(1) \|x\|_p$:

\begin{lemma}[Expected Value of $|\langle Z, x \rangle|$]
\label{lemma:p_stable_inner_product_log_independence_expected}
Let $n \in \N$ and $x \in \R^n$. For $i \in [n]$, let $Z_i$ be a standard $p$-stable random variable truncated at $D = \poly(n)$, and suppose the $Z_i$ are $O(\log n)$-wise independent. Let $Z \in \R^n$ so that the $i^{th}$ coordinate of $Z$ is $Z_i$. Then, $E[|\langle Z, x \rangle|] = \Theta_p(1) \|x\|_p$.
\end{lemma}

\begin{proof}
Let $x \in \R^n$ with $\|x\|_p = 1$, and define $Z$ as in the statement of this lemma. Here, the amount of independence of the coordinates of $Z$ is actually $O(\max(\log n, \log(1/\eps)))$, where $\eps$ is a constant that will be chosen to be sufficiently small later. We will prove the lemma in two steps:

\begin{itemize}
    \item With constant probability, $|\langle Z, x \rangle| \geq \Omega(1) \med_p$ which is $\Omega(1)$ by Lemma \ref{lemma:p_stable_median_omega_1}. Therefore, $E[|\langle Z, x \rangle|] \geq \Prob{|\langle Z, x \rangle \geq \Omega(1)} \cdot \Omega(1) = \Omega(1)$.
    \item $E[|\langle Z, x \rangle|] \leq O_p(1)$.
\end{itemize}
Then, we can extend to general $x$ by multiplying both sides by $\|x\|_p$.

The proof of the first step is similar to that of Theorem 2.1 of \cite{knw_10_approximating_lp_norm_p_stable}, and Lemma I.3 of \cite{l1_lower_bound_and_sqrtk_subset}. First assume for simplicity that the coordinates of $Z$ are standard $p$-stable random variables that are not truncated at $D = \poly(n)$ --- for $D = n^c$ and $c$ sufficiently large, the probability of $\|Z\|_\infty \leq D$ is $1 - o(1)$, so this assumption does not affect the probability calculations in the proof of the first step. Let $b > 0$ be a small constant, chosen so that if $X$ is a $p$-stable random variable, then with probability $0.999$, $|X| \geq b$. In addition, define $Y \in \R^n$ so that the coordinates of $Y$ are i.i.d. $p$-stable random variables. Note that $\langle Y, x \rangle$ is equal to $\|x\|_pX = X$, where $X$ is a $p$-stable random variable. Hence, if $1_{[-b, b]}$ is the indicator function of $[-b, b]$, then
$$E[1_{[-b, b]}(\langle Y, x \rangle)] = \Prob{|\langle Y, x \rangle| \leq b} \leq 0.001$$
For convenience, let $f_0 = 1_{[-b, b]}$, and define $f = J_{[-b, b]}^c$, as in Lemma \ref{lemma:smooth_approximations_to_interval}, for some sufficiently small $\eps > 0$ and sufficiently large $c$. Now, we wish to show that $E[1_{[-b, b]}(\langle Z, x \rangle)] = \Prob{|\langle Z, x \rangle| \leq b}$ is small. Following the proofs of Theorem 2.1 of \cite{knw_10_approximating_lp_norm_p_stable} and Lemma I.3 of \cite{l1_lower_bound_and_sqrtk_subset}, we will proceed by showing that
$$E[f_0(\langle Y, x \rangle)] \approx_{O(\eps)} E[f(\langle Y, x \rangle)] \approx_{O(\eps)} E[f(\langle Z ,x \rangle)] \approx_{O(\eps)} E[f_0(\langle Z, x \rangle)]$$
First, we show that $|E[f_0(\langle Y, x \rangle)] - E[f(\langle Y, x \rangle)]| \leq O_p(\eps)$ as follows. If $\langle Y, x \rangle$ is not in the intervals $[-b - \eps, -b + \eps]$ or $[b - \eps, b + \eps]$, then $|f_0(\langle Y, x \rangle) - f(\langle Y, x \rangle)| \leq \eps$ by the construction in Lemma \ref{lemma:smooth_approximations_to_interval}. On the other hand, since $\|x\|_p = 1$, $\langle Y, x \rangle$ is equal to $X$, where $X$ is a standard $p$-stable random variable. By Lemma \ref{lemma:anti_concentration_of_p_stables}, $X$ is anti-concentrated, meaning
$$\Prob{X \in [-b - \eps, -b + \eps]} = \Prob{X \in [b - \eps, b + \eps]} = O_p(\eps)$$
Moreover, $|f_0| \leq 1$ and $|f| \leq O(1)$ by the construction in Lemma \ref{lemma:smooth_approximations_to_interval}. Hence, if we let $\calE_1$ and $\calE_2$ be the events that $\langle Y, x \rangle$ is in $[-b - \eps, -b + \eps]$ and $[b - \eps, b + \eps]$ respectively, then
\begin{equation}
\begin{split}
\Big|E[f_0(\langle Y, x \rangle)] - E[f(\langle Y, x \rangle)]\Big|
& \leq E\Big[ |f_0(\langle Y, x \rangle) - f(\langle Y, x \rangle)|\Big] \\
& \leq E\Big[ |f_0(\langle Y, x \rangle) - f(\langle Y, x \rangle)| \mid \neg \calE_1 \cap \neg \calE_2 \Big] \\
& \,\, + E\Big[ |f_0(\langle Y, x \rangle) - f(\langle Y, x \rangle)| \mid \calE_1 \cup \calE_2 \Big] \\
& \leq O_p(\eps) + \Prob{\calE_1 \cup \calE_2} \cdot O(1) \\
& \leq O_p(\eps)
\end{split}
\end{equation}
where the first inequality is by Jensen's inequality. Next, observe that $|E[f(\langle Y, x \rangle)] - E[f(\langle Z, x \rangle)]| \leq O(\eps)$ due to Lemma \ref{lemma:smooth_function_removing_full_independence}.

Finally, we show that $|E[f(\langle Z, x \rangle)] - E[f_0(\langle Z, x \rangle)]| \leq O(\eps)$. By Lemma \ref{lemma:anti_concentration_of_p_stables}, $\langle Z, x \rangle$ is anti-concentrated, and in particular, $\ProbBig{\langle Z, x \rangle \in [-b - \eps, -b + \eps]} \leq O(\eps)$ and $\ProbBig{\langle Z, x \rangle \in [b - \eps, b + \eps]} \leq O(\eps)$. Moreover, outside of the intervals $[-b - \eps, -b + \eps]$ and $[b - \eps, b + \eps]$, $|f - f_0| \leq O(\eps)$. Hence, by the same reasoning used to show that $|E[f(\langle Y, x \rangle)] - E[f_0(\langle Y, x \rangle)]| \leq O(\eps)$, we know that $|E[f(\langle Z, x \rangle)] - E[f_0(\langle Z, x \rangle)]| \leq O(\eps)$. In summary, we have shown that
$$E[f_0(\langle Y, x \rangle)] \approx_{O(\eps)} E[f(\langle Y, x \rangle)] \approx_{O(\eps)} E[f(\langle Z, x \rangle)] \approx_{O(\eps)} E[f_0(\langle Z, x \rangle)]$$
meaning
$$E[1_{[-b, b]}(\langle Z, x \rangle)] \leq E[1_{[-b, b]}(\langle Y, x \rangle)] + O(\eps) \leq 0.001 + O(\eps)$$
Letting $\eps$ be a sufficiently small constant, we find that $\Prob{|\langle Z, x \rangle| \geq b} \geq \Omega(1)$, meaning $E[|\langle Z, x \rangle|] \geq \Omega_p(1)$.

Now, we show that $E[|\langle Z, x \rangle|] \leq O_p(1)$. We now drop the assumption that the $Z_i$ are not truncated --- in other words, let the $Z_i$ follow a standard $p$-stable distribution, conditioned on $|Z_i| \leq D = \poly(n)$. Also, define $Z' \in \R^n$, so that the coordinates of $Z'$ are standard (un-truncated) $p$-stable random variables, and the coordinates of $Z'$ are $O(\log n)$-wise independent. Finally, let $X$ be a standard $p$-stable random variable.

First, we write the expectation of $|\langle Z, x \rangle|$ in terms of the tails of $|\langle Z', x \rangle|$ --- then, we can relate the tails of $|\langle Z', x \rangle|$ to those of $|X|$ using Lemma \ref{lemma:tail_probability_small_independence_p_stable}. Note that
\begin{equation}
\begin{split}
E\Big[|\langle Z, x \rangle|\Big]
& = E\Big[|\langle Z', x \rangle| \mid \|Z'\|_\infty \leq D \Big] \\
& = \int_0^\infty \ProbBig{|\langle Z', x \rangle| \geq t \mid \|Z'\|_\infty \leq D} dt \\
& = \int_0^{Dn} \ProbBig{|\langle Z', x \rangle| \geq t \mid \|Z'\|_\infty \leq D} dt \\
& = \int_0^{Dn} \frac{\Prob{|\langle Z', x \rangle| \geq t \text{ and } \|Z'\|_\infty \leq D}}{\Prob{\|Z'\|_\infty \leq D}} dt \\
& \leq O(1) \int_0^{Dn} \ProbBig{|\langle Z', x \rangle| \geq t} dt
\end{split}
\end{equation}
Here, the third equality is because, if $\|Z'\|_\infty \leq D$, then $\|Z'\|_2 \leq D\sqrt{n}$, and by the Cauchy-Schwarz inequality, $|\langle Z', x \rangle| \leq \|Z'\|_2 \|x\|_2 \leq D\sqrt{n}\sqrt{n} = Dn$. The first inequality holds because with probability $1 - o(1)$, $\|Z'\|_\infty \leq D$ as long as $D = n^c$ for sufficiently large $c$.

Hence, our task is now to show that the integral $\int_0^{Dn} \Prob{|\langle Z', x \rangle | \geq t} dt$ is at most $O_p(1)$. Since the coordinates of $Z'$ are $O(\log n)$-wise independent $p$-stable random variables (which are in particular not truncated), Lemma \ref{lemma:tail_probability_small_independence_p_stable} is now applicable, and we can use it to relate the integral to the expectation of $|X|$. First, note that
$$0 \leq \int_0^{Dn} \Prob{|X| \geq t} dt \leq \int_0^\infty \Prob{|X| \geq t} dt = E[|X|] \leq O_p(1)$$
and therefore, it suffices to show that
$$\Big| \int_0^{Dn} \ProbBig{|\langle Z', x \rangle| \geq t} dt - \int_0^{Dn} \ProbBig{|X| \geq t} dt \Big| \leq O_p(1)$$
as follows:
\begin{equation}
\begin{split}
& \Big| \int_0^{Dn} \ProbBig{|\langle Z', x \rangle| \geq t} dt - \int_0^{Dn} \ProbBig{|X| \geq t} dt \Big| \\
& \leq \int_0^{Dn} \Big| \ProbBig{|\langle Z', x \rangle| \geq t} dt - \ProbBig{|X| \geq t} \Big| dt \\
& \leq \int_0^{Dn} O\Big(  \frac{(\log n)^{-1/p}}{1 + t^{p + 1}} + \frac{(\log n)^{-2/p}}{1 + t^2} + 2^{-\Omega(\log n)} \Big) dt \\
& \leq o(1) \int_0^{Dn} \frac{1}{1 + t^{p + 1}} dt + o(1) \int_0^{Dn} \frac{1}{1 + t^2} dt + \int_0^{Dn} 2^{-\Omega(\log n)} dt \\
& \leq o(1) \int_0^{\infty} \frac{1}{1 + t^{p + 1}} dt + o(1) \int_0^{\infty} \frac{1}{1 + t^2} dt + \int_0^{Dn} 2^{-\Omega(\log n)} dt \\
& \leq o(1) + \int_0^{Dn} 2^{-\Omega(\log n)} dt \\
& \leq o(1)
\end{split}
\end{equation}
Here, the second inequality is by Lemma \ref{lemma:tail_probability_small_independence_p_stable}, and the last inequality holds if the coordinates of $Z$ and $Z'$ are $(C \log n)$-wise independent for some sufficiently large $C$, since $Dn = \poly(n)$.
\end{proof}

In order to have finitely many values for the $Z_i$ (and hence a finite sample space) we can round the $Z_i$ to the nearest integer multiple of $\frac{1}{\poly(n)}$:

\begin{lemma}[Expected Value of $|\langle Z, x \rangle|$ After Rounding $Z$]
\label{lemma:rounded_p_stable_inner_product_log_independence_expected}
Let $n \in \N$ and $x \in \R^n$. For $i \in [n]$, let $Z_i$ be a standard $p$-stable random variable truncated at $D = \poly(n)$, and rounded to the nearest multiple of $\frac{1}{F}$, where $F = \poly(n)$. Suppose the $Z_i$ are $O(\log n)$-wise independent. Let $Z \in \R^n$ so that the $i^{th}$ coordinate of $Z$ is $Z_i$. Then, $E[|\langle Z, x \rangle|] = \Theta_p(1) \|x\|_p$.
\end{lemma}
\begin{proof}
Let $Y \in \R^n$ so that the coordinates of $Y$ are $O(\log n)$-wise independent $p$-stable random variables truncated at $D = \poly(n)$, and suppose $Z \in \R^n$ so that the coordinates of $Z$ are those of $Y$, but rounded to the nearest multiple of $\frac{1}{F}$ where $F = \poly(n)$. Finally, fix $x \in \R^n$ with $\|x\|_p = 1$. Then, by Lemma \ref{lemma:p_stable_inner_product_log_independence_expected}, $E[|\langle Y, x \rangle|] = \Theta_p(1)$. Therefore,
\begin{equation}
\begin{split}
\Big| E[|\langle Z, x \rangle|] - E[|\langle Y, x \rangle|] \Big| 
& \leq E\Big[\Big| |\langle Z, x \rangle| - |\langle Y, x \rangle|\Big| \Big] \\
& \leq E\Big[ |\langle Z - Y, x \rangle|\Big] \\
& \leq \frac{n}{F}
\end{split}
\end{equation}
and to have $E[|\langle Z, x \rangle|] = \Theta_p(1)$, it suffices to choose $F = Cn$ for a sufficiently large absolute constant $C$.
\end{proof}

Hence, we can assume that each of the $Z_i$ has at most $\poly(n)$ values. Now, since we only need the $Z_i$ to be $O(\log n)$-wise independent, a sample space of size $n^{O(\log n)}$ suffices to construct the embedding matrix in Fact \ref{fact:ell1_embedding_expected_value_observation}, and this sample space can be efficiently constructed by the following theorem of \cite{km93_small_sample_space}:

\begin{theorem}[Theorem 3.2 of \cite{km93_small_sample_space}]
\label{theorem:constructing_k_wise_independent_sample_spaces}
Let $X = (X_1, X_2, \ldots, X_n)$ be a random vector, such that the $X_i$ are identically distributed and take values in $[r]$. In addition, suppose the $X_i$ satisfy all independence constraints belonging to a set $\calC$, i.e., $\calC$ consists of subsets of $[n]$ such that, for all $S \in \calC$,
$$\ProbBig{X_i = x_i \text{ }\forall i \in S} = \prod_{i \in S} \Prob{X_i = x_i}$$
Then, it is possible to construct a joint distribution over the $X_i$ with sample space $\Omega$, such that $|\Omega| \leq |\calC|$, in running time $O(rn |\calC|^{2.62})$. 

In particular, if we wish for the $X_i$ to be $k$-wise independent (corresponding to the case when $\calC = \binom{[n]}{k}$) then $|\Omega| \leq \binom{n}{k}$ and the running time is $O(rn^{O(k)})$.
\end{theorem}

In summary, combining Fact \ref{fact:ell1_embedding_expected_value_observation}, Lemma \ref{lemma:rounded_p_stable_inner_product_log_independence_expected} and Lemma \ref{theorem:constructing_k_wise_independent_sample_spaces}, we obtain the following:

\begin{theorem}[Deterministic Embedding of $\ell_p^n$ into $\ell_1^{n^{O(\log n)}}$]
Let $n \in \N$, and $p \in (1, 2)$. Then, there exists a matrix $R \in \R^{r \times n}$, where $r = n^{O(\log n)}$, such that for all $x \in \R^n$,
$$\Omega_p(1) \|x\|_p \leq \|Rx\|_1 \leq O_p(1) \|x\|_p$$
$R$ can be constructed deterministically in $n^{O(\log n)}$ time.
\end{theorem}

\begin{remark}
The running time depends on $p$ as well, and goes to infinity as $p \to 1$. In particular, the number of terms in the Taylor series of $e^x$ needed below depends on $p$. This does not affect our reductions later on, since we can treat $p$ as a constant when reducing $\ell_{1, p}$ low rank approximation to constrained $\ell_1$ low rank approximation for a fixed $p$ (and similarly when reducing $\ell_p$ low rank approximation to $\ell_{1, p}$ low rank approximation).
\end{remark}

\begin{proof}
We apply Theorem \ref{theorem:constructing_k_wise_independent_sample_spaces} with the $X_i$ being $O(\log n)$-wise independent. Here, $X$ has the same distribution as the $Z$ specified in the statement of Lemma \ref{lemma:rounded_p_stable_inner_product_log_independence_expected}. Each of the coordinates $Z_i$ of $Z$ takes on values between $-D$ and $D$, where $D = \poly(n)$, in increments of $\frac{1}{F}$ where $F = \poly(n)$. To run the procedure described in Theorem \ref{theorem:constructing_k_wise_independent_sample_spaces}, we must approximately compute the distribution of $Z_i$, meaning we must find $\Prob{t \leq Z_i \leq t + F}$, for all integer multiples $t$ of $F$ that are between $-D$ and $D$.

\begin{claim}
Each of the probabilities $\Prob{t \leq Z_i \leq t + F}$ can be computed with error at most $\frac{1}{\poly(n)}$ in $\poly(n)$ time.
\end{claim}
\begin{proof}

Computing these probabilities reduces to the problem of computing the c.d.f. of a standard $p$-stable random variable at $t$, where the numerator and denominator of $t$ are at most $\poly(n)$. Let $X$ be a standard $p$-stable random variable. Then, as noted in the proof of Lemma \ref{lemma:p_stable_median_omega_1}, it is known due to \cite{nolan_97_pstable_density} that for $x > 0$
$$\Prob{X > x} = 1 - \frac{1}{\pi} \int_0^{\frac{\pi}{2}} e^{-x^{\frac{p}{p - 1}} \cdot V(\theta; p)} d\theta$$
where 
$$V(\theta; p) = \Big(\frac{\cos \theta}{\sin p\theta}\Big)^{\frac{p}{p - 1}} \cdot \frac{\cos (p - 1)\theta}{\cos \theta}$$
Hence,
$$\Prob{X \leq x} = \frac{1}{\pi} \int_0^{\frac{\pi}{2}} e^{-x^{\frac{p}{p - 1}} \cdot V(\theta; p)} d\theta$$

Hence, we must compute the integral $\int_0^{\frac{\pi}{2}} e^{-x^{\frac{p}{p - 1}} \cdot V(\theta; p)} d\theta$ within an accuracy of $O(\frac{1}{K})$, where $K = \poly(n)$. Observe that for $\theta \in [0, \frac{\pi}{2}]$, $\cos \theta \geq 0$, and for $p \in [1, 2]$, $0 \leq p\theta \leq \pi$, meaning $\sin p\theta \geq 0$, and $0 \leq (p - 1)\theta \leq \frac{\pi}{2}$, meaning $\cos (p - 1)\theta \geq 0$. Therefore, $e^{-x^{\frac{p}{p - 1}}V(\theta; p)} \leq 1$, meaning it suffices to evaluate
$$I := \int_{\frac{1}{K}}^{\frac{\pi}{2} - \frac{1}{K}} e^{-x^{\frac{p}{p - 1}} \cdot V(\theta ; p)} d\theta$$
which is at most $\frac{2}{K}$ away from the original integral.

As a first step, we show how to compute the exponent, $-x^{\frac{p}{p - 1}}V(\theta; p)$, within an error of $\frac{1}{2^{\poly(n)}}$ in $\poly(n)$ time. First, note that $|x^{\frac{p}{p - 1}}| \leq \poly(n)$, so we must simply compute $V(\theta; p)$ within an error of $\frac{1}{2^{\poly(n)}}$. However, observe that $\cos \theta$, $\sin p\theta$, $\cos ((p - 1)\theta)$, and $\cos \theta$ are bounded away from $0$ on $[\frac{1}{K}, \frac{\pi}{2} - \frac{1}{K}]$, since their only singularities are at $0$ and $\frac{\pi}{2}$. Furthermore, we can show that $\cos \theta, \sin p\theta \geq \Omega(\frac{1}{K})$ on this interval, by noting that $\cos \theta \geq \frac{1}{2}(\frac{\pi}{2} - x)$, $\sin p\theta \geq \min\Big(\frac{x}{2}, \frac{\pi/2 - x}{2}\Big)$, and $\cos (p - 1)\theta \geq \cos \theta$.

This means that all of the factors in $V(\theta; p)$ are bounded above by $1$ and below by $\frac{1}{\poly(n)}$. Hence, it suffices to compute each of those factors up to an additive $\frac{1}{2^{\poly(n)}}$ error, in order to compute them up to a multiplicative $(1 \pm \frac{1}{2^{\poly(n)}})$ error, since they are all at least $\frac{1}{\poly(n)}$. This suffices to compute $V(\theta; p)$ up to a multiplicative $(1 \pm \frac{1}{2^{\poly(n)}})$ error, and hence up to an additive $\frac{1}{2^{\poly(n)}}$ error (since it is at most $\poly(n)$).

Now, let us compute each of the factors in $V(\theta; p)$ by approximating them with their Taylor series. First, let $P_T$ be the degree-$T$ Taylor polynomial for $f(\theta) = \sin p\theta$. Then, by Taylor's theorem, for all $x \in [0, C]$ (where $C = O(1)$)
$$|P_T(x) - \sin (px)| \leq \max_{x \in [0, C]} |f^{(T)}(x)| \cdot \Big| \frac{x^{T + 1}}{(T + 1)!}\Big| \leq \frac{p^T x^T}{(T + 1)!}$$
and this expression is at most $\frac{1}{2^{O(n \log n)}}$ for $T = O(n)$. We can analyze the other factors of $V(\theta; p)$ similarly.

Hence, for any $\theta \in [\frac{1}{K}, \frac{\pi}{2} - \frac{1}{K}]$, if we let $E := x^{\frac{p}{p - 1}} V(\theta; p)$, then we can compute $E$ up to an additive $\frac{1}{2^{\poly(n)}}$ error. Since the function $f(x) = e^{-x}$ is Lipschitz continuous on $(-\infty, 0]$ with Lipschitz constant $1$, $|e^{-E} - e^{-E'}| \leq \frac{1}{2^{\poly(n)}}$, where $E'$ is our estimate of $E$. In addition, we can compute $e^{-E'}$ up to additive $\frac{1}{2^{\poly(n)}}$ error, by using the first $O(n)$ terms of the Taylor series of $e^{-x}$, through similar reasoning as above: again let $P_T$ be the degree $T$ Taylor polynomial of $e^x$, centered at $0$. We know that $0 \leq x^{\frac{p}{p - 1}}V(\theta; p) \leq \poly(n)$. Now, for $x \in [-\poly(K), 0]$,
$$|e^x - P_T(x)| \leq \Big(\max_{y \in [-\poly(n), 0]} |e^y|\Big) \Big|\frac{x^{T + 1}}{(T + 1)!}\Big| \leq \Big|\frac{x^{T + 1}}{(T + 1)!}\Big|$$
Since $x \leq \poly(n)$, the right-hand side is at most $\frac{1}{2^{\poly(n)}}$ if $T = n^c$ for a sufficiently large $c > 0$.

In summary, we can evaluate the integrand $e^{-x^{\frac{p}{p - 1}}V(\theta; p)}$ at any $\theta \in [\frac{1}{K}, \frac{\pi}{2} - \frac{1}{K}]$, up to an additive $\frac{1}{2^{\poly(n)}}$ error, in $\poly(n)$ time. The final step is to show that we can compute $\int_{\frac{1}{K}}^{\frac{\pi}{2} - \frac{1}{K}} e^{-x^{\frac{p}{p - 1}}V(\theta; p)} d\theta$ by evaluating the integrand at $\poly(n)$ many points. We can do this with the trapezoidal rule, by bounding the second derivative of the integrand. Let
$$f(\theta) = e^{-x^{\frac{p}{p - 1}} V(\theta; p)}$$
Then,
$$f'(\theta) = e^{-x^{\frac{p}{p - 1}} V(\theta; p)} \cdot (-x^{\frac{p}{p - 1}} \cdot V'(\theta; p))$$
and
$$f''(\theta) = e^{-x^{\frac{p}{p - 1}} V(\theta; p)} \cdot (-x^{\frac{p}{p - 1}} \cdot V'(\theta; p))^2 + e^{-x^{\frac{p}{p - 1}} V(\theta; p)} \cdot (-x^{\frac{p}{p - 1}} \cdot V''(\theta; p))$$
Observe that $e^{-x^{\frac{p}{p - 1}} V(\theta; p)} \leq 1$, and $x^{\frac{p}{p - 1}} \leq \poly(n)$. Finally, note that the numerators of $V'(\theta; p)$ and $V''(\theta; p)$ will be at most $1$, and the denominators will only involve $\sin p\theta$ and $\cos \theta$, which are at least $\frac{1}{\poly(n)}$ --- therefore, $V'(\theta; p)$ and $V''(\theta; p)$ are at most $\poly(n)$. Therefore, the second derivative of the integrand is at most $\poly(n)$, and it suffices to evaluate the integrand at $\poly(n)$ many points to get an estimate of $\int_{\frac{1}{K}}^{\frac{\pi}{2} - \frac{1}{K}} e^{-x^{\frac{p}{p - 1}} V(\theta; p)} d\theta$ with $\frac{1}{\poly(n)}$ additive error.
\end{proof}

Observe that if we compute the distribution of the $Z_i$ approximately rather than exactly, then this may affect $E[|\langle Z, x \rangle|]$. We now show that this is not an issue:

\begin{claim}
Suppose $Z_1, Z_2, \ldots, Z_n$ are as specified in the statement of Lemma \ref{lemma:rounded_p_stable_inner_product_log_independence_expected}, and $W_1, W_2, \ldots, W_n$ are $O(\log n)$-wise independent random variables, such that for all $z \in \R$, $|\Prob{Z_i = z} - \Prob{W_i = z}| \leq \frac{1}{\poly(n)}$, and $W_i$ takes on the same values as $Z_i$. Then, $E[|\langle W, x \rangle|] = \Theta_p(1)$.
\end{claim}
\begin{proof}
Let $Z$ be a real discrete random variable whose distribution is as specified in the statement of Lemma \ref{lemma:rounded_p_stable_inner_product_log_independence_expected}, and $W$ be a random variable so that for all $z \in \R$, $|\Prob{Z = z} - \Prob{W = z}| \leq \frac{1}{F}$, where $F = n^c$ for a sufficiently large constant $c$. Suppose $Z$ takes on the values $z_1, z_2, \ldots, z_r$ (and $W$ takes on these values). For convenience, define $p_{Z, i} = \Prob{Z = z_i}$ and $p_{W, i} = \Prob{W = z_i}$, and $s_{Z, i} = \sum_{j = 1}^i p_{Z, i}$ and $s_{W, i} = \sum_{j = 1}^i p_{W, i}$.

Let $P_U$ be the uniform distribution on $[0, 1]$. Define the random variables $Z_U$ and $W_U$ as follows. Suppose $U$ is drawn from $P_U$. If $U$ is in the interval $[s_{Z, i - 1}, s_{Z, i}]$ (where $S_{Z, 0} = 0$), then define $Z_U = z_i$. Similarly, if $U$ is in the interval $[s_{W, i - 1}, s_{W, i}]$, then define $W_U = z_i$. Note that $Z_U$ and $Z$ have the same distributions, and $W_U$ and $W$ have the same distributions. Finally, we say $U$ is \textit{bad} if $U \in [s_{Z, i - 1}, s_{Z, i}]$ and $U \in [s_{W, j - 1}, s_{W, j}]$ for $i \neq j$. Note that $\Prob{U \text{ is bad}} \leq \frac{\poly(n)}{F}$, since for each $i \in [n]$,
$$|s_{Z, i} - s_{W, i}| \leq \sum_{j = 1}^i |p_{Z, i} - p_{W, i}| \leq \frac{n}{F}$$
meaning $\Prob{U \in [s_{Z, i - 1}, s_{Z, i}] \cap [s_{W, j - 1}, s_{W, j}]}$ is $O(\frac{n}{F})$ if $i \neq j$, and we can perform a union bound over all such pairs $i, j$ to find that $\Prob{U \text{ is bad}} \leq O(\frac{n^3}{F})$.

Now, let $U_1, U_2, \ldots, U_n$ be $O(\log n)$-wise independent random variables drawn from $P_U$, and define $Z_{U_1}, Z_{U_2}, \ldots, Z_{U_n}$ and $W_{U_1}, W_{U_2}, \ldots, W_{U_n}$ as described above. Let $\textbf{Z} = (Z_{U_1}, Z_{U_2}, \ldots, Z_{U_n})$ and $\textbf{W} = (W_{U_1}, W_{U_2}, \ldots, W_{U_n})$. Finally, let $x \in \R^n$ so that $\|x\|_p = 1$. By Lemma \ref{lemma:rounded_p_stable_inner_product_log_independence_expected}, we know that $E[|\langle \textbf{Z}, x \rangle|] = \Theta_p(1)$, and we wish to show that $E[|\langle \textbf{W}, x \rangle|] = \Theta_p(1)$. Let $\calE$ be the event that $U_i$ is bad for some $i$ --- observe that if none of the $U_i$ were bad, then $\textbf{Z} = \textbf{W}$. However, 
$$\Prob{\calE} \leq \sum_{i = 1}^n \Prob{U_i \text{ is bad}} \leq O\Big(\frac{n^4}{F}\Big) \leq \frac{1}{n^{10}}$$ 
if the constant $c$ is chosen large enough. Hence,
$$\Theta_p(1) = E[|\langle \textbf{Z}, x \rangle|] = E[|\langle \textbf{Z}, x \rangle| \mid \calE] \Prob{\calE} + E[|\langle \textbf{Z}, x \rangle| \mid \neg\calE] \Prob{\neg \calE}$$
and
$$E[|\langle \textbf{W}, x \rangle|] = E[|\langle \textbf{W}, x \rangle| \mid \calE] \Prob{\calE} + E[|\langle \textbf{W}, x \rangle| \mid \neg\calE] \Prob{\neg \calE}$$
Note that with probability $1$, $\|\mathbf{Z}\|_2, \|\mathbf{W}\|_2 \leq \poly(n)$, meaning $E[|\langle \textbf{W}, x \rangle| \mid \calE] \Prob{\calE} \leq \frac{1}{\poly(n)}$ and $E[|\langle \textbf{Z}, x \rangle| \mid \calE] \Prob{\calE} \leq \frac{1}{\poly(n)}$ if $F$ is chosen large enough. In addition, if $\calE$ does not hold, then $\mathbf{W} = \mathbf{Z}$. Hence,
$$\Big||E[|\langle \mathbf{W}, x \rangle|] - E[|\langle \mathbf{Z}, x \rangle|] \Big| \leq \frac{1}{\poly(n)}$$
and $E[|\langle \mathbf{W}, x \rangle|] = \Theta_p(1)$. This proves the claim.
\end{proof}

If we let $Z \in \R^n$ have coordinates given by the statement of Lemma \ref{lemma:rounded_p_stable_inner_product_log_independence_expected}, then in $n^{O(\log n)}$ time, we can construct a sample space $\Omega$ for the $Z_i$, of size at most $n^{O(\log n)}$, such that the $O(\log n)$-wise independence of the $Z_i$ holds. We can then construct the embedding matrix $R$ as described in Fact \ref{fact:ell1_embedding_expected_value_observation}, using the $Z_i$ in place of the random variables $R_i$ in the statement of Fact \ref{fact:ell1_embedding_expected_value_observation}, and letting $R$ be the resulting matrix $A$ given in Fact \ref{fact:ell1_embedding_expected_value_observation}.
\end{proof}

\subsection{Hardness for $\ell_p$ and $\ell_{1, p}$ Low Rank Approximation with Additive Error, $p \in (1, 2)$}

The goal of this section is to use the techniques of \cite{ptas_for_lplra} to show the following stronger hardness result.

\begin{theorem} [Hardness for $\ell_p$ and $\ell_{1, p}$ Low Rank Approximation with Additive Error]
\label{theorem:hardness_lp_additive_error}
Suppose the SSEH and ETH hold. Then, for $p \in (1, 2)$, at least $2^{k^{\Omega(1)}}$ time is required to achieve the following guarantees:
\begin{itemize}
    \item Given a matrix $A \in \R^{n \times d}$, $k \in \N$, find a matrix $\widehat{A} \in \R^{n \times d}$ of rank at most $k$ such that
    $$\|\widehat{A} - A\|_p \leq O(1) \min_{A_k \text{ rank }k} \|A - A_k\|_p + \frac{1}{2^{\poly(k)}} \|A\|_p$$
    \item Given a matrix $A \in \R^{n \times d}$, $k \in \N$, find a matrix $\widehat{A} \in \R^{n \times d}$ of rank at most $k$ such that
    $$\|\widehat{A} - A\|_{1, p} \leq O(1) \min_{A_k \text{ rank }k} \|A - A_k\|_{1, p} + \frac{1}{2^{\poly(k)}}\|A\|_{1, p}$$
    where for a matrix $M \in \R^{n \times d}$, $\|M\|_{1, p} = \sum_{j = 1}^d \|M_j\|_p$.
\end{itemize}
\end{theorem}

Using the hardness of the second guarantee, we will show hardness for constrained $\ell_1$ low rank approximation as well. It is worth noting that Algorithm \ref{algorithm:guessing_eps_approximation_general_lp} can in fact be used to achieve the first guarantee in $2^{\poly(k)} + \poly(nd)$ time, by letting $f = 2^{\poly(k)}$ (instead of letting $f = \poly(k)$ as is done when Algorithm \ref{algorithm:guessing_eps_approximation_general_lp} is called by Algorithm \ref{algorithm:get_eps_approximation_after_initialization_general_lp}).

First we show the following lemma, similar to Lemma \ref{lemma:equivalence_of_lp_norm_lra_min_p_to_p}, which will be used to show hardness for $\ell_{1, p}$-norm low rank approximation:

\begin{lemma}[Equivalence of $\min_{p^* \to 1}$ and Rank---$(k - 1)$ $\ell_{1, p}$ Low Rank Approximation]
\label{lemma:equivalence_of_sum_of_p_norms_min_p_to_1}
Let $p \in (1, \infty)$ and let $p^*$ be the H{\"o}lder conjugate of $p$. Let $A \in \R^{n \times d}$ with $n \geq d$ and $k = d - 1$. Then,
$$\min_{U \in \R^{d \times k}, V \in \R^{k \times n}} \|UV - A^T\|_{1, p} = \min_{x \in \R^d, \|x\|_{p^*} = 1} \|Ax\|_1 =: \min_{p^* \to 1}(A)$$
\end{lemma}
\begin{proof}
The proof is the same as that of Lemma 27 of \cite{ptas_for_lplra}, but we make small modifications now that the objective is the $\ell_{1, p}$-norm. First, assume the rank of $A$ is $d$, since both sides are $0$ otherwise. In addition, for a matrix $M \in \R^{n \times d}$, define
$$\|M\|_{p, 1} = \sum_{i = 1}^n \|M^i\|_p$$
In other words, $\|M\|_{p, 1} = \|M^T\|_{1, p}$. Now, let $x \in \R^d$ such that $\|x\|_{p^*} = 1$, and take $V \in \R^{k \times d}$ so that the rows of $V$ span a $(k - 1)$---dimensional subspace of $\R^d$ orthogonal to $x$. In addition, define $U = \argmin_{U \in \R^{n \times k}} \|UV - A\|_{p, 1}$, meaning that each row $U^i$ of $U$ is the minimizer $\argmin_{U^i} \|U^iV - A^i\|_p$. Then, we can write
$$\|U^iV - A^i\|_p = \min_{\langle y, x \rangle = 0} \|y - A^i\|_p = \min_{\langle z, x \rangle = - \langle A^i, x \rangle} \|z\|_p$$
\begin{claim}
$\min_{\langle z, x \rangle = -\langle A^i, x \rangle} \|z\|_p = |\langle A^i, x \rangle|$
\end{claim}
\begin{proof}
First, by H{\"o}lder's inequality, for all $z \in \R^d$,
$$|\langle x, z \rangle | \leq \|x\|_{p^*} \|z\|_p = \|z\|_p$$
since $\|x\|_p^* = 1$. Hence, if $\langle z, x \rangle = - \langle A^i, x \rangle$, then
$$\|z\|_p \geq | \langle x, z \rangle| = | \langle A^i, x \rangle |$$
On the other hand, if we choose $z$ so that the $j^{th}$ coordinate of $z$ is $z_j = - \langle x, A^i \rangle (\text{sgn}(x_j)|x_j|^{\frac{p^*}{p}})$ then
$$\langle z, x \rangle = -\langle x, A^i \rangle \sum_j (\text{sgn}(x_j)|x_j|^{\frac{p^*}{p}}) \cdot x_j = -\langle x, A^i \rangle \sum_j |x_j|^{\frac{p^* + p}{p}} = -\langle x, A^i \rangle \sum_j |x_j|^{p^*} = -\langle x, A^i \rangle$$
where the third equality is because $\frac{1}{p^*} + \frac{1}{p} = 1$, meaning $p^* + p = pp^*$, and the last equality is because $\|x\|_{p^*} = 1$. In addition,
$$\|z\|_p^p = |\langle x, A^i \rangle|^p \sum_j |x_j|^{p^*} = |\langle x, A^i \rangle|^p$$
meaning $\|z\|_p = |\langle x, A^i \rangle|$. This proves the claim.
\end{proof}

By this claim, $\|U^i V - A^i\|_p = |\langle A^i, x \rangle|$, meaning
$$\|UV - A\|_{p, 1} = \sum_{i = 1}^n \|U^i V - A^i\|_p = \sum_{i = 1}^n |\langle A^i, x \rangle| = \|Ax\|_1$$
Hence, for every $x \in \R^d$ with $\|x\|_{p^*} = 1$, if the rows of $V \in \R^{k \times d}$ span the subspace orthogonal to $x$, then
$$\|Ax\|_1 = \|UV - A\|_{p, 1} \geq \min_{U \in \R^{n \times k}, V \in \R^{k \times d}} \|UV - A\|_{p, 1}$$
Similarly, if $V \in \R^{k \times d}$, we can assume without loss of generality that $V$ has rank exactly $k$, and choose $x$ perpendicular to the rows of $V$ such that $\|x\|_{p^*} = 1$. Then, as we have shown,
$$\|UV - A\|_{p, 1} = \|Ax\|_1 \geq \min_{p^*\to 1}(A)$$
meaning
$$\min_{U \in \R^{n \times k}, V \in \R^{k \times d}} \|UV - A\|_{p, 1} \geq \min_{p^* \to 1}(A)$$
Finally, to complete the proof, note that $\|UV - A\|_{p, 1}$ is equal to $\|(UV - A)^T\|_{1, p} = \|V^TU^T - A^T\|_{1, p}$.
\end{proof}

Now we show the main result of this subsection.

\begin{theorem} [Hardness for $\ell_p$ and $\ell_{1, p}$ Low Rank Approximation with Additive Error]
Suppose the SSEH and ETH hold. Then, for $p \in (1, 2)$, at least $2^{k^{\Omega(1)}}$ time is required to achieve the following guarantees:
\begin{itemize}
    \item Given a matrix $A \in \R^{n \times d}$, $k \in \N$, find a matrix $\widehat{A} \in \R^{n \times d}$ of rank at most $k$ such that
    $$\|\widehat{A} - A\|_p \leq O(1) \min_{A_k \text{ rank }k} \|A - A_k\|_p + \frac{1}{2^{\poly(k)}} \|A\|_p$$
    \item Given a matrix $A \in \R^{n \times d}$, $k \in \N$, find a matrix $\widehat{A} \in \R^{n \times d}$ of rank at most $k$ such that
    $$\|\widehat{A} - A\|_{1, p} \leq O(1) \min_{A_k \text{ rank }k} \|A - A_k\|_{1, p} + \frac{1}{2^{\poly(k)}}\|A\|_{1, p}$$
\end{itemize}
\end{theorem}

\begin{proof}
The main idea for showing hardness of the first guarantee is that while performing the reduction due to \cite{ptas_for_lplra}, we can ensure that the matrix $A$ has entries with at most $\poly(k)$ bits. As a result of this, we show that the first guarantee is in fact equivalent to achieving an $O(1)$-approximation. To show hardness for the second guarantee, we use this observation, and also combine the reduction of \cite{ptas_for_lplra} with Theorem \ref{theorem:embedding_lp_into_l1_deterministic} above. We use Theorem \ref{theorem:embedding_lp_into_l1_deterministic} to reduce the computation of $\min_{p^*\to p}(A)$ to that of $\min_{p^*\to 1}(RA)$, and this is equivalent to finding the error from the best rank---$(k - 1)$ approximation to $A^TR^T$ in the $\ell_{1, p}$-norm, by Lemma \ref{lemma:equivalence_of_sum_of_p_norms_min_p_to_1}. We now begin the proof.

As in the reduction of \cite{ptas_for_lplra}, let $G$ be a regular graph that has $k$ vertices, so that $P := P_{\geq 1/2}(G)$ is a $k \times k$ matrix. Note that we can write $P = UU^T$, where $U \in \R^{k \times t}$ (for some $t \in \N$) has orthonormal columns. The entries of $U$ have absolute value at most $1$, meaning the entries of $P$ have absolute value at most $k$ by the triangle inequality. In addition, the entries of $P$ can be rounded to the nearest integer multiple of $\frac{1}{D}$, where $D = \poly(k)$ is a sufficiently large power of $2$, without significantly changing $\|P\|_{2 \to p^*}$. We can see this as follows. Let $\widehat{P}$ be $P$ with its entries rounded to the nearest integer multiple of $\frac{1}{D}$. Then,
\begin{equation}
\begin{split}
\Big|\|P\|_{2 \to p^*} - \|\widehat{P}\|_{2 \to p^*}\Big|
& \leq \|P - \widehat{P}\|_{2 \to p^*} \\
& \leq \|P -\widehat{P}\|_{2\to2} \\
& \leq \|P - \widehat{P}\|_F \\
& \leq \frac{k}{D}
\end{split}
\end{equation}
Here, the first inequality is by the triangle inequality. The second is because $\|x\|_{p^*} \leq \|x\|_2$ for any $x \in \R^n$, and the third is because the spectral norm is at most the Frobenius norm. Finally, the last inequality is because $\|P - \widehat{P}\|_\infty \leq \frac{1}{D}$, and $P$ and $\widehat{P}$ are $k \times k$ matrices.

Now, we must simply choose a large enough $D$ so that we can distinguish between the Yes and No cases of the Small Set Expansion Problem using a constant-factor approximation to $\|\widehat{P}\|_{2 \to p^*}$, where the Yes and No cases are as specified in the proof of Theorem \ref{theorem:2_to_q_norm_hard}. As mentioned in the proof of Theorem \ref{theorem:2_to_q_norm_hard}, in the Yes case, $$\|P\|_{2 \to p^*} \geq 1/(10\delta^{(p^* - 2)/2p^*}) =: C_1$$ while in the No case $$\|P\|_{2 \to p^*
} \leq 2/\delta^{(p^* - 2)/4p^*} =: C_2$$ (where $\delta \in (0, 1)$ can be any sufficiently small number). Hence, it suffices to choose $D \geq 50k\delta^{(p
^*- 2)/4p^*}$. To see why, note that in the No case,
$$\|\widehat{P}\|_{2 \to p^*} \leq \|P\|_{2 \to p^*} + \frac{k}{D} \leq \frac{2}{\delta^{(p^* - 2)/4p^*}} + \frac{1}{50\delta^{(p^* - 2)/4p^*}} \leq \frac{3}{\delta^{(p^* - 2)/4p^*}}$$
while in the Yes case,
$$\|\widehat{P}\|_{2 \to p^*} \geq \|P\|_{2 \to p^*} - \frac{k}{D} \geq \frac{1}{10\delta^{(p^* - 2)/2p^*}} - \frac{1}{50\delta^{(p^*- 2)/4p^*}} \geq \frac{1}{20\delta^{(p^* - 2)/2p^*}}$$
where the last inequality holds because $\frac{1}{\delta} > 1$ and $(p^* - 2)/2p^* > (p^* - 2)/4p^*$, meaning $\frac{1}{20\delta^{(p^* - 2)/2p^*}} \geq \frac{1}{50\delta^{(p^* - 2)/4p^*}}$. Note that the gap between the Yes and No case is still
$$\frac{1/(20\delta^{(p^* - 2)/2p^*})}{3/\delta^{(p^* - 2)/4p^*}} = \frac{1}{60\delta^{(p^* - 2)/4p^*}}$$
and this can be made arbitrarily large as $\delta \to 0$, meaning computing $\|\widehat{P}\|_{2 \to p^*}$ within an arbitrary constant factor is NP-hard. We now reduce the problem of computing $\|\widehat{P}\|_{2 \to p^*}$ to that of computing the optimal rank---$(k - 1)$ approximation of a well-chosen matrix where each entry has at most $\poly(k)$ bits.

First, observe that each entry of $\widehat{P}$ has at most $O(\log k)$ bits in its numerator and denominator, since each entry is at most $k$ in absolute value, and each entry has denominator $D$, which has $O(\log k)$ bits. Now, we show that each of the steps in the reduction of \cite{ptas_for_lplra} preserves the bit complexity of $\widehat{P}$, meaning that none of the $P_i$ mentioned in the previous subsection will have more than $\poly(k)$ bits. First, we can construct $P_1$ so that $(1 - \eps)\|\widehat{P}\|_{2 \to p^*} \leq \|P_1\|_{2 \to p^*} \leq (1 + \eps) \|\widehat{P}\|_{2 \to p^*}$ by Lemma \ref{lemma:reduction_step_1} --- recall that this can be done by adding $\frac{\eps M}{\|I\|_{2 \to p^*}}$ to each diagonal entry of $\widehat{P}$, where $I$ is the identity matrix. Since only a constant-factor approximation to $\|\widehat{P}\|_{2 \to p^*}$ is needed to decide the Small Set Expansion problem, we can let $\eps = \Theta(1)$. In addition, recall that $M$ has at most $O(\log k)$ bits in both its numerator and denominator (where $M$ is the largest entry of $\widehat{P}$). Finally, $\|I\|_{2 \to p^*} = 1$ since $p^* > 2$, meaning $\|x\|_p^* \leq \|x\|_2$ for all $x$. Hence, all entries of $P_1$ have at most $O(\log k)$ bits in the numerator and denominator.

Furthermore, recall that all entries of $\widehat{P}$ are integer multiples of $\frac{1}{D}$, meaning all entries of $P_1$ are integer multiples of $\frac{1}{K}$ for some $K = \poly(k)$. Hence, we can compute $P_2 = P_1P_1^T$, and all of the entries of $P_2$ are at most $\poly(k)$ and are integer multiples of $\frac{1}{K^2}$, meaning all entries of $P_2$ have at most $O(\log k)$ bits in their numerators and denominators. Note that by Lemma \ref{lemma:reduction_step_2}, $\|P_2\|_{p \to p^*} = \|P_1\|_{2 \to p^*}^2$, meaning in order to get a constant-factor approximation to $ \|P_1\|_{2 \to p^*}$ it suffices to get a constant-factor approximation to $\|P_2\|_{p \to p^*}$.

Finally, we compute $P_3 = P_2^{-1}$ as described in the previous subsection. Each entry of $P_3$ is a cofactor of $P_2$, divided by the determinant of $P_2$. Each cofactor is the sum of at most $k!$ terms, each of which is a product of $(k - 1)$ entries of $P_2$. Similarly, the determinant of $P_2$ is the sum of $k!$ products of $k$ entries of $P_2$. Since each entry of $P_2$ is an integer multiple of $\frac{1}{D^2}$, each of these products has a denominator of $D^{2k}$, which is at most $2^{\poly(k)}$ (meaning it has $\poly(k)$ bits). The numerator of each of these products has absolute value at most $k! \cdot \poly(k)$, meaning it has at most $\poly(k)$ bits. In summary, for each entry of $P_3$, its numerator and denominator have at most $\poly(k)$ bits, and our argument so far establishes the following claim:

\begin{claim}
There exists a matrix $P_3 \in \R^{k \times k}$, which can be deterministically computed in $\poly(k)$ time, such that its entries have at most $\poly(k)$ bits in their numerators and denominators, and computing a constant-factor approximation for $\|\widehat{P}\|_{2 \to p^*}$ reduces to computing a constant-factor approximation for $\min_{p^* \to p}(P_3)$.
\end{claim}

We now show the following claim, which we use for showing hardness both for $\ell_p$ low rank approximation and $\ell_{1, p}$ low rank approximation:

\begin{claim} \label{claim:lower_bound_on_objective}
$\min_{p^* \to p}(P_3) \geq \frac{1}{\poly(k)}$
\end{claim}
\begin{proof}
Since each entry of $P$ has absolute value at most $k$, $\|P\|_{2 \to p^*} \leq \poly(k)$, meaning $\|\widehat{P}\|_{2 \to p^*} \leq \poly(k)$. Next, note that by Lemma \ref{lemma:reduction_step_1},
$$\|P_1\|_{2 \to p^*} \leq O(1) \|\widehat{P}\|_{2 \to p^*} \leq \poly(k)$$
Moreover, by Lemma \ref{lemma:reduction_step_2},
$$\|P_2\|_{p\to p^*} = \|P_1\|_{2 \to p^*}^2 \leq \poly(k)$$
and finally, by Lemma \ref{lemma:reduction_step_3},
$$\min_{p^* \to p}(P_3) = \min_{p^* \to p}(P_2^{-1}) = ( \|P_2\|_{p \to p^*})^{-1} \geq \frac{1}{\poly(k)}$$
\end{proof}

Now we show hardness for $\ell_p$ low rank approximation. Let $A = P_3$ --- by Lemma \ref{lemma:equivalence_of_lp_norm_lra_min_p_to_p}, $\min_{p^* \to p}(A) = \min_{M_{k - 1} \text{ rank } (k - 1)} \|M_{k - 1} - A\|_p$. To show the desired result, it suffices to show the following claim:

\begin{claim}
$\frac{1}{2^{\poly(k)}}\|A\|_p \leq OPT_p := \min_{M_{k - 1} \text{ rank }(k - 1)} \|M_{k - 1} - A\|_p$
\end{claim}

\begin{proof}
Note that by the previous claim, $OPT_p = \min_{p^* \to p}(A) \geq \frac{1}{\poly(k)}$. On the other hand, $\|A\|_p \leq 2^{\poly(k)}$, since the entries of $A$ have at most $\poly(k)$ bits, meaning they are at most $2^{\poly(k)}$. Therefore, $\|A\|_p \leq 2^{\poly(k)}$, and this proves the claim. Alternatively, it is possible to prove this claim using an argument similar to Claim 1 on page 15 of \cite{ptas_for_lplra}.
\end{proof}

Hence, if we find $\widehat{A}$ of rank at most $k - 1$ such that
$$\|\widehat{A} - A\|_p \leq O(1) \min_{M_{k - 1} \text{ rank }(k - 1)} \|M_{k - 1} - A\|_p + \frac{1}{2^{\poly(k)}}\|A\|_p$$
then the claim above in fact shows that $\|\widehat{A} - A\|_p \leq O(1) \min_{M_{k - 1} \text{ rank }(k - 1)} \|M_{k - 1} - A\|_p$ --- this is enough to obtain a constant-factor approximation to $\|\widehat{P}\|_{2 \to q}$, and decide the Small Set Expansion instance. We now show a similar hardness result for low rank approximation in the $\ell_{1, p}$ norm by using Theorem \ref{theorem:embedding_lp_into_l1_deterministic}.

\begin{claim} \label{claim:reduction_to_l1p_low_rank_approx}
There exists a matrix $B \in \R^{k \times k^{O(\log k)}}$ which can be computed in $k^{O(\log k)}$ time such that the following properties hold:
\begin{itemize}
    \item The entries of $B$ have at most $\poly(k)$ bits in their numerators and denominators.
    \item Computing a constant-factor approximation for $\|\widehat{P}\|_{2 \to p^*}$ reduces to computing a constant-factor approximation for $OPT_{1, p} := \min_{M_{k - 1} \text{ rank }(k - 1)} \|M_{k - 1} - B\|_{1, p}$, and
    \item $OPT_{1, p} \geq \frac{1}{\poly(k)}$
\end{itemize}
\end{claim}

\begin{proof}
Recall that to compute a constant-factor approximation for $\|\widehat{P}\|_{2 \to p^*}$, it suffices to find a constant-factor approximation for
$$\min_{p^* \to p}(P_3) = \min_{\|x\|_{p^*} = 1} \|P_3x\|_p$$
By Theorem \ref{theorem:embedding_lp_into_l1_deterministic}, there exists a matrix $R \in \R^{k^{O(\log k)} \times k}$ such that for all $y \in \R^k$,
$$\Omega(1) \|y\|_p \leq \|Ry\|_1 \leq O(1)\|y\|_p$$
and $R$ can be constructed in $k^{O(\log k)}$ time so that the above holds with constant probability. Hence,
$$\min_{p^* \to 1}(RP_3) = \min_{\|x\|_{p^*} = 1} \|RP_3x\|_1$$
is within a constant factor of $\min_{p^* \to p}(P_3)$, and it suffices to compute a constant factor approximation to $\min_{p^* \to 1}(RP_3)$, which by Lemma \ref{lemma:equivalence_of_sum_of_p_norms_min_p_to_1} is equal to $\min_{M_{k - 1} \text{ rank }(k - 1)} \|M_{k - 1} - P_3^TR^T\|_{1, p}$.

Finally, to complete the proof of the claim, we define $B$ to be $P_3^TR^T$, but with each of its entries rounded to the nearest integer multiple of $\frac{1}{2^{\poly(k)}}$. Let
$$OPT_{1, p}^{0} = \min_{M_{k - 1} \text{ rank }(k - 1)} \|M_{k - 1} - P_3^TR^T\|_{1, p}$$
with $M^{0}$ being the corresponding minimizer, and
$$OPT_{1, p} = \min_{M_{k - 1} \text{ rank }(k - 1)} \|M_{k - 1} - B\|_{1, p}$$
with $M$ being the corresponding minimizer. Then, note that 
$$\|M - B\|_{1, p} \leq \|M^{0} - B\|_{1, p} \leq \|M^0 - P_3^TR^T\|_{1, p} + \frac{1}{2^{\poly(k)}}$$
where the last inequality is because $$\|B - P_3^T R^T\|_{1, p} = \sum_{j = 1}^{k^{O(\log k)}} \|B_j - (P_3^TR^T)_j\|_p \leq k^{O(\log k)} \cdot \frac{1}{2^{\poly(k)}} \leq \frac{1}{2^{\poly(k)}}$$
Similarly,
$$\|M^0 - P_3^TR^T\|_{1, p} \leq \|M - P_3^TR^T\|_{1, p} \leq \|M - B\|_{1, p} + \frac{1}{2^{\poly(k)}}$$
meaning that $\Big| \|M - B\|_{1, p} - \|M^0 - P_3^TR^T\|_{1, p} \Big| \leq \frac{1}{2^{\poly(k)}}$. Since $\|M^0 - P_3^TR^T\|_{1, p} \geq \frac{1}{\poly(k)}$ by Claim \ref{claim:lower_bound_on_objective}, this means that $\|M - B\|_{1, p}$ is within a constant factor of $\|M^0 - P_3^TR^T\|_{1, p}$.

In summary, to obtain a constant-factor approximation to $\|\widehat{P}\|_{2\to p^*}$ with constant probability, it suffices to obtain a constant-factor approximation to $\min_{M_{k - 1} \text{ rank }(k - 1)} \|M_{k - 1} - B\|_{1, p}$. Moreover, the entries of $B$ have at most $\poly(k)$ bits in their numerators and denominators. We can see this as follows. First, recall that $P_3$ has entries with at most $\poly(k)$ bits in their numerators and denominators, meaning its entries are at most $2^{\poly(k)}$. In addition, the entries of $R$ are at most $O(1)$, since for any standard basis vector $e_i$, $\|Re_i\|_1 = \Theta(1)$. Therefore, the entries of $P_3^TR^T$ are at most $2^{\poly(k)}$. Finally, after we round the entries of $P_3^TR^T$ to the nearest integer multiple of $\frac{1}{2^{\poly(k)}}$ to obtain $B$, the entries of $B$ are at most $2^{\poly(k)}$ and are multiples of $\frac{1}{2^{\poly(k)}}$, meaning they have at most $\poly(k)$ bits. Finally, $OPT_{1, p} = \Theta(1) \cdot OPT_{1, p}^0 \geq \frac{1}{\poly(k)}$.
\end{proof}

Since the entries of $B$ have at most $\poly(k)$ bits in their numerators and denominators, $\|B\|_{1, p} \leq 2^{\poly(k)}$. In addition, $OPT_{1, p} \geq \frac{1}{\poly(k)}$, meaning $\|B\|_{1, p} \leq 2^{\poly(k)} OPT_{1, p}$, and if we find $\widehat{B}$ of rank at most $k - 1$ such that
$$\|\widehat{B} - B\|_{1, p} \leq O(1) \min_{M_{k - 1} \text{ rank }(k - 1)} \|M_{k - 1} - B\|_{1, p} + \frac{1}{2^{\poly(k)}} \|B\|_{1, p}$$
then $\|\widehat{B} - B\|_{1, p}$ is a constant factor approximation to $OPT_{1, p}$.

In summary, we have shown that achieving either of the guarantees in Theorem \ref{theorem:hardness_lp_additive_error} is sufficient to decide Problem \ref{problem:small_set_expansion_problem}. Now, suppose we assume the Exponential-Time Hypothesis (ETH) in addition to the Small Set Expansion Hypothesis --- ETH is the assumption that any algorithm for solving $3$-SAT takes at least $2^{\Omega(n)}$ time, where $n$ is the number of variables in the $3$-SAT instance. Since SSEH is the assumption that the Small Set Expansion Problem is NP-hard, we can assume there is a reduction from 3-SAT to Small Set Expansion that takes an instance with $n$ variables to a graph with at most $m = n^{\frac{1}{c}}$ vertices for some $c > 0$ (since the reduction takes polynomial time, it creates a graph with size at most $\poly(n)$). Hence, assuming SSEH and ETH, there is no algorithm which can decide any Small Set Expansion instance in $2^{o(m^c)}$ time, since this would imply that any 3-SAT instance could be decided in $2^{o(n)}$ time. Therefore, an algorithm achieving either of the guarantees in Theorem \ref{theorem:hardness_lp_additive_error} requires $2^{k^c}$ time.
\end{proof}

\subsection{Hardness for Constrained $\ell_1$ Low Rank Approximation with Additive Error}

We recall the constrained $\ell_1$ low rank approximation problem:

\begin{problem}[Constrained $\ell_1$ Low Rank Approximation]
\label{problem:constrained_l1_low_rank_approximation}
Given a matrix $A \in \R^{n \times d}$ and a subspace $V \subset \R^n$, find a matrix $\widehat{A}$ of rank at most $k$ minimizing $\|\widehat{A} - A\|_1$, such that the columns of $\widehat{A}$ are in $V$.
\end{problem}

\begin{theorem}[Hardness for Constrained $\ell_1$ Low Rank Approximation with Additive Error]
\label{theorem:constrained_l1_lra_randomized_hardness}
Suppose the SSEH and ETH hold. Then, for $k \in \N$, at least $2^{k^{\Omega(1)}}$ time is required to achieve the following guarantee with constant probability: given a matrix $A \in \R^{n \times d}$, $Y \in \R^{n \times t}$ for some $t \leq n$, and $k \in \N$, find a matrix $\widehat{A}$ of rank at most $k$ such that
$$\|\widehat{A} - A\|_1 \leq O(1) \min_{A_k} \|A_k - A\|_1 + \frac{1}{2^{\poly(k)}}\|A\|_1$$
and the columns of $\widehat{A}$ are contained in the column span of $Y$ --- here, the minimum on the right-hand side is taken over $A_k$ having rank $k$, such that the columns of $A_k$ are contained in the column span of $Y$.
\end{theorem}
\begin{proof}
We reduce $\ell_{1, p}$ low rank approximation (more specifically, the problem of obtaining the guarantee in Theorem \ref{theorem:hardness_lp_additive_error}) to the above guarantee for constrained $\ell_1$ low rank approximation. For convenience, let $p = \frac{3}{2}$. Recall that by Claim \ref{claim:reduction_to_l1p_low_rank_approx} and the discussion following its proof, given any instance of Problem \ref{problem:small_set_expansion_problem} on a graph of size $k$, a matrix $B \in \R^{k \times k^{O(\log k)}}$ can be computed in $k^{O(\log k)}$ time such that, in order to decide the instance of Problem \ref{problem:small_set_expansion_problem}, it suffices to compute a matrix $\widehat{B} \in \R^{k \times k^{O(\log k)}}$ such that
\begin{equation} \label{eq:l1_p_guarantee}
\begin{split}
\|\widehat{B} - B\|_{1, p} \leq O(1) \min_{M_{k - 1} \text{ rank }(k - 1)} \|M_{k - 1} - B\|_{1, p} + \frac{1}{2^{\poly(k)}} \|B\|_{1, p}
\end{split}
\end{equation}
Now, by Theorem \ref{theorem:embedding_lp_into_l1_deterministic}, there exists a matrix $R \in \R^{k^{O(\log k)} \times k}$, which can be computed in $k^{O(\log k)}$ time, such that for all $x \in \R^k$,
$$\Omega(1) \|x\|_p \leq \|Rx\|_1 \leq O(1)\|x\|_p$$
Now, suppose we could find $M \in \R^{k \times k^{O(\log k)}}$, with rank at most $k - 1$, such that
\begin{equation} \label{eq:constrained_l1_guarantee}
\begin{split}
\|RM - RB\|_1 \leq O(1) \min_{M_{k - 1} \text{ rank }(k - 1)} \|RM_{k - 1} - RB\|_1 + \frac{1}{2^{\poly(k)}} \|RB\|_1
\end{split}
\end{equation}
Note that this is a solution to the constrained $\ell_1$ low rank approximation problem, that achieves the guarantee described in the statement of this theorem (replacing $A$ with $RB$, the subspace $Y \in \R^{n \times t}$ with $R$, and the target rank $k$ with $k - 1$). However, for any matrix $M \in \R^{k \times k^{O(\log k)}}$,
$$\|RM\|_1 = \sum_{i = 1}^d \|RM_i\|_1 = \sum_{i = 1}^d \Theta(1) \|M_i\|_p = \Theta(1) \|M\|_{1, p}$$
where the second equality is because for any $x \in \R^n$, $\|Rx\|_1$ is within a constant factor of $\|x\|_p$. Hence, Equation \ref{eq:constrained_l1_guarantee} implies that
$$\|M - B\|_{1, p} \leq O(1) \min_{M_{k - 1} \text{ rank }(k - 1)} \|M_{k - 1} - B\|_{1, p} + \frac{1}{2^{\poly(k)}} \|B\|_{1, p}$$
i.e. $M$ satisfies the guarantee of Equation \ref{eq:l1_p_guarantee}. Therefore, achieving the guarantee in Equation \ref{eq:constrained_l1_guarantee} is enough to decide the original instance of Problem \ref{problem:small_set_expansion_problem}, and therefore requires $2^{k^{\Omega(1)}}$ time to achieve.
\end{proof}

\newpage
\section{$\ell_p$ Column Subset Selection Lower Bound}
\label{appendix:lp_css_lower_bound_full_proofs}

We use a construction similar to \cite{l1_lower_bound_and_sqrtk_subset}, to show that column subset selection algorithms cannot give an approximation factor better than $O(k^{\frac{1}{p} - \frac{1}{2}})$ for $\ell_p$-low rank approximation, for $1 < p < 2$, if $k \cdot \polylog(k)$ columns are chosen. The hard distribution for $\ell_1$ column subset selection from \cite{l1_lower_bound_and_sqrtk_subset}, which we use here for $\ell_p$ column subset selection, is as follows: $A \in \R^{(k + n) \times n}$ is a random matrix where each of the first $k$ rows has i.i.d. $N(0, 1)$ entries, and the remaining $n \times n$ submatrix is the identity matrix. We show that if $n = k \cdot \poly(\log k)$, then any subset of $r \leq k \cdot \poly(\log k)$ columns cannot give an approximation factor better than $O(k^{\frac{1}{p} - \frac{1}{2}})$, for sufficiently large $k$, unless $n - r = o(n)$. The proofs in this section follow those of the analogous lemmas in Section G.3 of \cite{l1_lower_bound_and_sqrtk_subset} with slight modifications --- for each of our lemmas, we note the corresponding lemma in \cite{l1_lower_bound_and_sqrtk_subset}. First, we state some definitions.

\begin{definition} \label{def:flat_vectors}
Suppose $\beta, \gamma \in \R$ and $p \in (1, 2)$. Then, we define
$$Y_{\beta, \gamma, p} = \Big\{y \in \R^n \mid \|y\|_p \leq O(k^{\gamma}), |y_i| \leq \frac{1}{k^\beta} \leq 1 \,\, \forall i \in [n] \Big\}$$
\end{definition}

\begin{definition} \label{def:event_flat_vectors_stay_flat}
(Similar to Definition G.19 of \cite{l1_lower_bound_and_sqrtk_subset}) Let $V \in \R^{n \times r}$ have orthonormal columns, and $A \in \R^{k \times r}$ be a random matrix, for which each entry is i.i.d. $N(0, 1)$. Then, we define the event $\calEHat(A, V, \beta, \gamma, p)$ as follows: $\forall y \in Y_{\beta, \gamma, p}$, $AV^Ty$ has at most $O(\frac{k}{\log k})$ coordinates with absolute value at least $\Omega(\frac{1}{\log k})$, and $\|A\|_2 = O(\sqrt{r})$.
\end{definition}

\begin{lemma} [Gaussian Matrices Have Small Operator Norm - Due to \cite{rv_10_random_matrices}, as stated in \cite{l1_lower_bound_and_sqrtk_subset}]
\label{lemma:gaussian_matrix_operator_norm}
Let $A \in \R^{r \times k}$ be a random matrix for which each entry is i.i.d. $N(0, 1)$. With probability at least $1 - e^{-\Theta(r)}$, the maximum singular value of $\|A\|_2$ is at most $O(\sqrt{r})$.
\end{lemma}

\begin{lemma} [Flat Vectors Stay Flat - Similar to Lemma G.20 of \cite{l1_lower_bound_and_sqrtk_subset}]
\label{lemma:high_prob_flat_vectors_stay_flat}
Suppose $k \geq 1$, $c_2 \geq c_1 \geq 1$, and $k \leq r = O(k (\log k)^{c_1})$, $r \leq n = O(k (\log k)^{c_2})$. Let $V \in \R^{n \times r}$ have orthonormal columns, and let $A \in \R^{k \times r}$ be a random matrix with each entry being i.i.d. $N(0, 1)$. Finally, suppose $\beta, \gamma \in \R$ such that $0 < \gamma \leq p \gamma < \beta(2 - p) \leq \beta$. Then,
$$\ProbBig{\calEHat(A, V, \beta, \gamma, p)} \geq 1 - 2^{-\Theta(k)}$$
\end{lemma}

\begin{proof}
The proof is the same as that of Lemma G.20 in \cite{l1_lower_bound_and_sqrtk_subset} --- it uses a net argument and a union bound. Rather than constructing a net for all of $Y_{\beta, \gamma, p}$, the coordinates of a point $y \in Y_{\beta, \gamma, p}$ are first divided between points $y^j$ with disjoint supports, such that $y^j$ has coordinates between $\frac{1}{2^{j + 1}}$ and $\frac{1}{2^j}$. Then, for each $j$, an $\eps$-net $\calN_j$ (for a suitable $\eps$) is constructed for all points of the same form as $y^j$, with coordinates between $\frac{1}{2^{j + 1}}$ and $\frac{1}{2^j}$. It is shown that with high probability, all points in $\calN_j$ have at most $O(k/\log^2 k)$ coordinates which are at least $\Omega(\frac{1}{\log^2 k})$ --- if this occurs, then this implies that $\calEHat(A, V, \beta, \gamma, p)$ occurs. Moreover, it is only necessary to consider $O(\log k)$ distinct values of $j$ --- for sufficiently large values of $j \geq j^*$, $AV^T y^j$ makes a very small contribution to the coordinates of $AV^Ty$, and in fact, $\sum_{j \geq j^*} AV^Ty^j$ makes a very small contribution to the coordinates of $AV^Ty$ (the proof of this last statement is where Lemma \ref{lemma:gaussian_matrix_operator_norm} is used).

We now begin the proof. Let $Y_{\beta, \gamma, p}$ be as in Definition \ref{def:flat_vectors}, and take $y \in Y_{\beta, \gamma, p}$. As mentioned before, write $y = \sum_{j = j_0}^\infty y^j$, for $j \geq j_0$, $y^j$ has coordinates in the interval $[\frac{1}{2^{j + 1}}, \frac{1}{2^j})$. Note that the $y^j$ have disjoint supports, and we can take $j_0$ so that $\frac{1}{2^{j_0}} = \frac{1}{k^\beta}$ by the definition of $Y_{\beta, \gamma}$. Let $s_j$ be the support size of $y^j$. Then,
$$\Big(\frac{s_j}{2^{p(j + 1)}}\Big)^{\frac{1}{p}} \leq \|y^j\|_p \leq \|y\|_p \leq O(k^\gamma)$$
meaning
$$s_j \leq O(2^{p(j + 1)}k^{p\gamma})$$
In addition, it will later be useful to bound from above $\|y^j\|_2$:
$$\|y^j\|_2^2 \leq \frac{1}{2^{2j}} \cdot s_j \leq O(k^{p\gamma} \cdot 2^{p(j + 1) - 2j})$$

We use this to construct an $\eps$-net $\calN_j$ for points of the same form as $y^j$, and bound its size:
$$\calN_j := \{x \in \R^n \mid \exists x' \in \Z^n ,\, x = \eps x' ,\, \|x\|_p \leq O(k^\gamma), \, \forall i \in [n], \text{ either } \frac{1}{2^{j + 1}} \leq |x_i| < \frac{1}{2^j} \text{ or } x_i = 0\}$$
In other words, it is the integer grid, but contracted by a factor of $\eps$, and with coordinates and $\ell_p$-norm in the same range as $x$. This is an $\eps$-net in the $\ell_\infty$ norm, meaning it is an $\eps \sqrt{n}$-net in the $\ell_2$ norm. We will choose $\eps = O(\frac{1}{nrk^3}) = O(\frac{1}{k^{c_1 + c_2 + 3}})$, meaning $\calN_j$ is an $O(\frac{1}{k^{c_1 + c_2/2 + 3}})$-net for the $\ell_2$ norm. Now, we bound the size of $\calN_j$ --- for each coordinate $x_i$, since it is between $\frac{1}{2^{j + 1}}$ and $\frac{1}{2^j}$ and the net has a granularity of $\eps$, the number of choices for $x_i$ is $1 + \frac{1}{2^j\eps} \leq 1 + \frac{2}{\eps}$ since $\frac{1}{2^j} \leq \frac{1}{k^\beta} \leq 2$. The number of coordinates of $y^j$ is at most $n = O(k (\log k)^{c_2})$, meaning
\begin{equation}
\begin{split}
|\calN_j|
& \leq \Big(1 + \frac{2}{\eps}\Big)^{O(k (\log k)^{c_2})} \\
& \leq 2^{O(k (\log k)^{c_2} \log \frac{1}{\eps})} \\
& \leq 2^{O(k (\log k)^{c_2 + 1})}
\end{split}
\end{equation}

In preparation for the union bound over points in $\calN_j$, we consider the event $\calE(y^j)$ that for a single $y^j$, $AV^Ty^j$ has at least $O(\frac{k}{\log^2 k})$ coordinates which are at least $O(\frac{1}{\log^2 k})$ in absolute value. Notice that a single coordinate of $AV^Ty^j$ is $N(0, \|V^Ty^j\|_2^2)$ since the rows of $A$ are i.i.d. $N(0, 1)$. Hence, the probability $q$ that a particular coordinate of $AV^Ty^j$ is greater than $O(\frac{1}{\log^2 k})$ is, by properties of the Gaussian distribution, at most
$$\exp\Big(-\frac{1}{\|V^Ty^j\|_2^2} \cdot \frac{1}{\log^4 k}\Big)$$
since the probability that a Gaussian random variable $N(0, \sigma^2)$ has absolute value greater than $t$ is at most $e^{O(-\frac{t^2}{\sigma^2})}$. Hence, letting $i_0 = O(\frac{k}{\log^2 k})$, the probability that $AV^Ty^j$ has at least $O(\frac{k}{\log^2 k})$ coordinates greater than $O(\frac{1}{\log^2 k})$ in absolute value is at most
\begin{equation}
\begin{split}
\sum_{i = i_0}^{k} q^i(1 - q)^i \binom{k}{i}
& \leq k2^k q^i \\
& \leq k2^k \exp\Big(-\frac{i}{\|V^Ty^j\|_2^2} \cdot \frac{1}{\log^4 k}\Big) \\
& \leq k2^k \exp\Big(-\Theta\Big(\frac{k}{\|V^Ty^j\|_2^2 \log^6k}\Big)\Big) \\
& \leq k2^k \exp\Big(-\Theta\Big(\frac{k^{1 - p\gamma}}{2^{p - (2 - p)j} \cdot \log^6k}\Big)\Big) \\
& \leq k2^k \exp\Big(-\Theta\Big(\frac{k^{1 + \beta(2 - p) - p\gamma}}{\log^6k}\Big)\Big) \exp\Big(-\Theta\Big(\frac{k^{1 - p\gamma}}{2^{p - (2 - p)j} \cdot \log^6k}\Big)\Big) \\
& \leq \exp\Big(-\Theta\Big(\frac{k^{1 - p\gamma}}{2^{p - (2 - p)j} \cdot \log^6k}\Big)\Big)
\end{split}
\end{equation}
The first inequality is because $(1 - q)^i \leq 1$, $\binom{k}{i} \leq 2^k$, and there are at most $k$ summands. The third inequality is because $i \geq i_0 = \Theta(\frac{k}{\log^2 k})$. The fourth inequality is because $\|V^Ty^j\|_2^2 \leq \|V^T\|_2^2 \|y^j\|_2^2 \leq O(k^{p\gamma} \cdot 2^{p(j + 1) - 2j})$ because $V$ has orthonormal columns and because of the upper bound on $\|y^j\|_2$ mentioned above. The fifth inequality is because $2^{(2 - p)j} \geq 2^{(2 - p)j_0} = \frac{1}{k^{\beta(2 - p)}}$. Finally, the sixth inequality is because $\beta(2 - p) > p\gamma$, meaning $k^{\beta(2 - p) - p\gamma}/(\log^6 k) = \omega(1)$, and $k2^k$ is absorbed in the remaining exponential --- this is where we use the hypothesis that $p\gamma < (2 - p)\beta$.

Now we show that we can ignore the effect of $y^j$ for $j \geq \Theta(\log k)$ --- we then perform a union bound over the remaining $O(\log k)$ nets $\calN_j$ to show that $\calE(y^j)$ holds for all $y^j$ in those $\calN_j$. In particular, as in \cite{l1_lower_bound_and_sqrtk_subset}, let $j_1 = \lceil 100(c_1 + c_2 + 1) \log k \rceil$. Then, for $j \geq j_1$, $\|\sum_{j = j_1}^\infty y^j \|_2 \leq \Theta(2^{-j_1} \sqrt{n}) \leq \frac{1}{k^{100(1 + c_1)}}$ since all coordinates of $\sum_{j = j_1}^\infty y^j$ are at most $\frac{1}{2^{j_1}}$ and the $y^j$ have disjoint supports. Hence,
\begin{equation}
\begin{split}
\|AV^Ty^j\|_\infty
& \leq \|AV^Ty^j\|_2 \\
& \leq \|A\|_2 \|V^T\|_2 \|y^j\|_2 \\
& \leq O(\sqrt{r}) \cdot \frac{1}{k^{100(1 + c_1)}} \\
& \leq O\Big(\frac{\sqrt{k}(\log k)^{c_1/2}}{k^{100(1 + c_1)}}\Big) \\
& \leq \frac{1}{k^{100}}
\end{split}
\end{equation}
where the third inequality is because $A$ has top singular value at most $O(\sqrt{r})$ and $V$ has orthonormal columns. Hence, if $y^{---j_1} := \sum_{j = j_0}^{j_1} y^j$, then $y$ has at least $O(\frac{k}{\log k})$ coordinates with absolute value at least $O(\frac{1}{\log k})$, if and only if this is the case for $y^{---j_1}$.

Finally, by a union bound over $O(\log k)$ nets $\calN_{j_0}, \ldots, \calN_{j_1}$, we can show that for all $j$ between $j_0$ and $j_1$, with high probability, $AV^Ty^j$ has at most $O(k/\log^2 k)$ coordinates with absolute value greater than or equal to $O(1/\log^2 k)$ --- therefore, outside of a set of $O(\log k) \cdot O(k/\log^2 k) = O(k/\log k)$ coordinates, each coordinate of $AV^Ty$ is at most $O(\log k) \cdot O(1/\log^2 k) = O(1/\log k)$.

First, the probability that there exists $\widehat{y^j} \in \calN_j$ for some $j$ between $j_0$ and $j_1$, such that $\calE(\widehat{y^j})$ occurs, is

\begin{equation}
\begin{split}
P\Big[ \exists y^j \in \bigcup_{j = j_0}^{j_1} \calN_j,\, \calE(y^j) \text{ happens}\Big]
& \leq \sum_{j = j_0}^{j_1} |\calN_j| \cdot \exp\Big(-\Theta\Big(\frac{k^{1 - p\gamma}}{2^{p - (2 - p)j} \cdot \log^6k}\Big)\Big) \\
& \leq \sum_{j = j_0}^{j_1} 2^{O(k (\log k)^{c_2 + 1})} \cdot \exp\Big(-\Theta\Big(\frac{k^{1 - p\gamma} \cdot 2^{(2 - p)j}}{\log^6k}\Big)\Big) \\
& \leq O(\log k) \cdot 2^{O(k (\log k)^{c_2 + 1})} \cdot \exp\Big(-\Theta\Big(\frac{k^{1 - p\gamma} \cdot k^{\beta(2 - p)}}{\log^6k}\Big)\Big) \\
& \leq O(\log k) \cdot 2^{-\Theta(k)} \\
& \leq 2^{-\Theta(k)}
\end{split}
\end{equation}
The first inequality above is by a union bound and our upper bound on the probability of $\calE(\widehat{y^j})$ for an individual $\widehat{y^j}$ in $\calN_j$. The third inequality is because $2^j \geq k^\beta$. The fourth inequality is because $\beta(2 - p) - p\gamma > 0$, meaning that $k \poly(\log k) = o(k^{1 - p\gamma + \beta(2 - p)}/\log^6 k)$.

Finally, given a vector $y^j$ with each entry having absolute value in $[\frac{1}{2^{j + 1}}, \frac{1}{2^j})$ such that $\|y^j\|_p \leq k^\gamma$ and each entry is at most $\frac{1}{k^\beta}$, there exists $\widehat{y^j} \in \calN_j$ such that $\|y^j - \widehat{y^j}\|_2 \leq O(\frac{1}{k^{c_1 + c_2/2 + 3}})$. Hence, assuming that $\calE(\widehat{y^j})$ does not occur for any $\widehat{y^j}$ in $\calN_j$ for $j$ between $j_0$ and $j_1$,
\begin{equation}
\begin{split}
\|AV^Ty^j - AV^T\widehat{y^j}\|_\infty
& \leq \|AV^Ty^j - AV^T\widehat{y^j}\|_2 \\
& \leq O(\sqrt{r}) \|y^j - \widehat{y^j}\|_2 \\
& \leq O\Big(\frac{\sqrt{r}}{k^{c_1 + c_2/2 + 3}}\Big) \\
& \leq O\Big(\frac{1}{k^3}\Big)
\end{split}
\end{equation}
since $A$ has operator norm $O(\sqrt{r})$ and $V^T$ has operator norm $\leq 1$. Hence, with probability $1 - 2^{-\Theta(k)}$, for all $y \in Y_{\beta, \gamma, p}$, and all $j$ between $j_0$ and $j_1$, $AV^Ty^j$ has at most $O(\frac{k}{\log^2 k})$ coordinates with absolute value more than $O(\frac{1}{\log^2 k})$.

In summary, as discussed above, this implies that with probability $1 - 2^{-\Theta(k)}$, for all $y \in Y_{\beta, \gamma, p}$, $AV^Ty$ has at most $O(k/\log k)$ coordinates with absolute value more than $O(1/\log k)$. This completes the proof.
\end{proof}

\begin{lemma}[Paying with Either Regression Cost or Norm - Similar to Lemma G.22 of \cite{l1_lower_bound_and_sqrtk_subset}]
\label{lemma:paying_with_either_regression_cost_or_norm}
For any $t, k \geq 1$, and any constants $c_2 \geq c_1 \geq 1$, let $k \leq r = O(k (\log k)^{c_1})$, and $r \leq n = O(k (\log k)^{c_2})$. Let $V \in \R^{n \times r}$ be a matrix with orthonormal columns. For an arbitrary constant $\alpha \in (0, 0.5)$, if $A \in \R^{k \times r}$ is such that $\calEHat(A, V, \frac{1 + \alpha}{2}, \frac{1}{p} - \frac{1}{2} - \alpha, p)$ holds, then with probability $1 - 2^{-\Theta(tk)}$, there are at least $\lceil\frac{t}{10}\rceil$ such $j \in [t]$ that $\forall x \in \R^r$, either $\|Ax - v_j\|_p \geq \Omega(k^{\frac{1}{p} - \frac{1}{2} - \alpha})$ or $\|Vx\|_p \geq \Omega(k^{\frac{1}{p} - \frac{1}{2} - \alpha})$.
\end{lemma}

\begin{remark}
The meaning of this lemma is that when a subset $A_S$ of columns of the hard instance $A \in \R^{(k + n) \times n}$ is chosen, then there exists a large enough subset of the remaining columns such that for each of those columns $v_j$, the regression coefficient vector $y_j$ used to fit those $v_j$ (using $A_S$ as the left factor) either leads to a large regression error on the top $k$ rows, or has a large norm, in which case it leads to a large regression cost on the bottom $n$ rows, which are simply the identity matrix.
\end{remark}

\begin{proof}
The proof is similar to that of the analogous Lemma G.22 of \cite{l1_lower_bound_and_sqrtk_subset}. Here we use a simplified version of that argument, that was given to us by Peilin Zhong.

We show the desired statement using a net argument. For convenience, let $\gamma = \frac{1}{p} - \frac{1}{2} - \alpha$, and let $\beta = \frac{1 + \alpha}{2}$. Then, note that $p\gamma = 1 - \frac{p}{2} - p\alpha$, which is less than $\beta(2 - p) = 1 - \frac{p}{2} + \alpha \cdot \frac{2 - p}{2}$. Hence, by Lemma \ref{lemma:high_prob_flat_vectors_stay_flat}, $\calEHat(A, V, \frac{1 + \alpha}{2}, \frac{1}{p} - \frac{1}{2} - \alpha, p)$ holds with probability $1 - 2^{-\Theta(k)}$. For a particular $x \in \R^r$, if we let $y := Vx$, meaning $x = V^Ty$, then the statement becomes equivalent to showing that there are at least $\lceil \frac{t}{10} \rceil$ indices $j \in [t]$ such that $\forall y \in \R^n$, either $\|y\|_p \geq \Omega(k^\gamma)$, or $\|AV^Ty - v_j\|_p \geq \Omega(k^\gamma)$. Throughout this proof, we consider a \textit{fixed} matrix $A \in \R^{k \times r}$ such that $\calEHat(A, V, \beta, \gamma, p)$ holds --- all randomness will be over the coordinates of the $v_j$.

Let $B_{\gamma, p}$ denote the $\ell_p$-ball of radius $O(k^\gamma)$ with center $0$. Then, we must show that for all $y \in B_{\gamma, p}$, $\|AV^Ty - v\|_p \geq \Omega(k^\gamma)$ with high probability. For each such $y$, we can write it as $y_0 + y_1$, where all the coordinates of $y_0$ are either $0$ or greater than $\frac{1}{k^\beta}$ in absolute value, all the coordinates of $y_1$ are at most $\frac{1}{k^\beta}$, and the supports of $y_0$ and $y_1$ are disjoint. Recall that this means $y_1 \in Y_{\beta, \gamma, p}$ as defined above. Since $\calEHat(A, V, \beta, \gamma, p)$ holds, for each $y \in B_{\gamma, p}$, its corresponding $y_1$ is such that $AV^Ty_0$ has at most $O(k/\log k)$ coordinates that are at least $O(1/\log k)$ in absolute value.

We first build a net for all possible $y_0$. For convenience, let $T_{\beta, \gamma, p}$ be the set of all possible $y_0$, i.e. $T_{\beta, \gamma, p} := \{x \in \R^n \mid \forall i \in [n], \text{ either }  |x_i| \geq \frac{1}{k^\beta} \text{ or } x_i = 0\}$. Then, for any $\eps > 0$, the following is an $\eps$-net in the $\ell_\infty$ norm for $T_{\beta, \gamma, p}$ (and hence an $\eps \sqrt{n}$-net in the $\ell_2$ norm):
$$\calN := \{x \in \R^n \mid \exists x' \in \Z^n ,\, x = \eps x' ,\, \|x\|_p \leq O(k^\gamma), \, \forall i \in [n], \text{ either }  |x_i| \geq \frac{1}{k^\beta} \text{ or } x_i = 0\}$$
Let us calculate the size of $|\calN|$. Note that if $s$ is the support size of $y_0$, then $\frac{1}{k^{\beta p}}s \leq \|y_0\|_p^p \leq \|y\|_p^p \leq k^{\gamma p}$ meaning $s \leq k^{\beta p + \gamma p}$. Moreover, for each coordinate of $x \in \calN$, the number of choices is at most $O(\frac{k^\gamma}{\eps})$ (since each coordinate is at most $k^\gamma$ and has a granularity of $\eps$). Hence, if we let $\eps = O(1/rnk^3)$ as in \cite{l1_lower_bound_and_sqrtk_subset}, then
$$|\calN| \leq \binom{n}{s} \cdot (k^{\gamma}/\eps)^{O(s)} = 2^{O(s)\log k} = 2^{k^{\beta p + \gamma p} \log k}$$
since $n = O(k(\log k)^{c_2})$.

We will perform a union bound over all $y_0 \in \calN$, to show that the failure event $\calE_1(y_0)$ does not occur, where $\calE_1(y_0)$ is the event that for some $y_1 \in Y_{\beta, \gamma, p}$ with disjoint support from $y_0$, if $y := y_0 + y_1$, then $\|AV^Ty - v\|_p \leq O(k^\gamma)$. First, let $\calE_2(y_0)$ be the event that $AV^Ty_0 - v$ has at most $O(k/\log k)$ coordinates that are at least $O(1/\log k)$ in absolute value.

\begin{claim}[Similar to Claim G.23 of \cite{l1_lower_bound_and_sqrtk_subset}]
Assume $\calEHat(A, V, \beta, \gamma, p)$ holds. Then, for all $y_0 \in T_{\beta, \gamma, p}$, $\calE_1(y_0)$ implies $\calE_2(y_0)$.
\end{claim}
\begin{proof}
The proof is similar to that of Claim G.23 of \cite{l1_lower_bound_and_sqrtk_subset}. Suppose $\calE_1(y_0)$ occurs, meaning there is $y_1 \in Y_{\beta, \gamma, p}$ such that if $y := y_0 + y_1$, then $\|AV^Ty - v\|_p \leq O(k^\gamma)$. Let $s$ be the number of coordinates of $AV^Ty - v$ that are greater than $\frac{1}{\log k}$ in absolute value --- then, $\frac{s}{\log^p k} \leq \|AV^Ty - v\|_p^p \leq k^{p\gamma}$, meaning $s \leq k^{p\gamma} \log^p k$, and this is $o(k/\log k)$, since $p\gamma < p \cdot (\frac{1}{p} - \frac{1}{2}) < 1$.

In addition, observe that
$$AV^Ty_0 - v = (AV^Ty - v) - AV^Ty_1$$
since $y = y_0 + y_1$. Since $\calEHat(A, V, \beta, \gamma, p)$ occurs, $AV^Ty_1$ has at most $O(k/\log k)$ coordinates that are greater than $\frac{1}{\log k}$ in absolute value, by the previous lemma. Hence, since $AV^Ty - v$ has at most $o(k/\log k)$ coordinates that are greater than $\Omega(1/\log k)$ in absolute value, $AV^Ty_0 - v$ also has at most $O(k/\log k)$ coordinates that are greater than $\Omega(1/\log k)$ in absolute value, meaning $\calE_2(y_0)$ holds.
\end{proof}

Now, we show that the probability of $\calE_2(y_0)$ is small, where the randomness is over the entries of a vector $v \in \R^k$ with i.i.d. $N(0, 1)$ entries.

\begin{claim}
Let $z \in \R^k$, and let $v$ be a $k$-dimensional vector with i.i.d. $N(0, 1)$ entries. Then, with probability $1 - 2^{-\Theta(k)}$, there exist \textit{at most} $O(k/\log k)$ coordinates $i \in [k]$ such that $|z_i - v_i| \leq O(1/\log k)$.
\end{claim}
\begin{proof}
Let $i \in [k]$. Then,
\begin{equation}
\begin{split}
\Prob{|z_i - v_i| \leq O(1/\log k)}
& = \int_{v_i - O(1/\log k)}^{v_i + O(1/\log k)} \frac{1}{\sqrt{2\pi}} e^{-x^2/2} dx \\
& \leq O(1/\log k)
\end{split}
\end{equation}
Hence, if $Z_i$ is $1$ if $|z_i - v_i| \leq O(1/\log k)$ and $0$ otherwise, and $Z = \sum Z_i$ is the number of coordinates on which the difference between $z$ and $v$ is at most $O(1/ \log k)$, then $E[Z] \leq O(k/\log k)$. Since the $Z_i$ are independent, by a Chernoff bound, with probability $1 - 2^{-\Theta(k)}$, $Z \leq O(k/\log k)$. This proves the claim.
\end{proof}

By applying the claim above with $z = AV^Ty_0$ for any particular $y_0 \in \R^n$, the probability that $AV^Ty_0 - v$ has at least $\frac{k}{10}$ coordinates on which the difference is at most $O(1/\log k)$ (meaning $AV^Ty_0 - v$ has less than $\frac{9k}{10}$ coordinates on which the difference is greater than $O(1/\log k)$) is at most $2^{-\Theta(k)}$ --- since this implies $\calE_2(y_0)$, the probability of $\calE_2(y_0)$ is also at most $2^{-\Theta(k)}$. Now, we union bound over the net $\calN$ to show that this occurs for \textit{all} $y_0 \in \calN$ with high probability:
\begin{equation}
\begin{split}
\ProbBig{\exists y_0 \in \calN, \calE_2(y_0)}
& \leq |\calN| \cdot 2^{-\Theta(k)} \\
& \leq 2^{k^{(\beta + \gamma)p}\log k} 2^{-\Theta(k)} \\
& \leq 2^{-\Theta(k)}
\end{split}
\end{equation}
Here we use the hypothesis that $\beta + \gamma < \frac{1}{p}$, meaning $k^{(\beta + \gamma)p} = o(k)$.

The above argument shows that for all $y_0 \in \calN$, with probability at least $1 - 2^{-\Theta(k)}$, the event $\mathcal{F}$ holds that there does not exist $y_1 \in Y_{\beta, \gamma, p}$ such that $\|AV^Ty - v\|_p \leq O(k^\gamma)$ where $y = y_0 + y_1$. Assume that $\mathcal{F}$ holds. Now, suppose $y_0 \in T_{\beta, \gamma, p}$, and there exists $y_1 \in Y_{\beta, \gamma, p}$ such that, if $y = y_0 + y_1$, then $\|AV^Ty - v\|_p \leq O(k^\gamma)$. Let $\widehat{y}_0 \in \calN$ such that $\|\widehat{y}_0 - y_0\|_2 \leq \eps \sqrt{n} = O(1/r\sqrt{n}k^3)$. Then, letting $\widehat{y} := \widehat{y}_0 + y_1$, we obtain
$$\|AV^T\widehat{y} - v\|_p \leq \|AV^Ty - v\|_p + \|AV^T(\widehat{y} - y)\|_p \leq O(k^\gamma) + k^{\frac{1}{p} - \frac{1}{2}} \|AV^T(\widehat{y}_0 - y_0)\|_2$$
Here, the first inequality is by the triangle inequality, and the second is because $\widehat{y}_0 \in \calN$, meaning $\|AV^T\widehat{y} - v\|_p \leq O(k^\gamma)$. Finally, because of our choice of $\widehat{y}_0$ as the net vector closest to $y_0$, and because $k^{\frac{1}{p} - \frac{1}{2}} \leq \sqrt{k}$, we obtain $\|AV^T\widehat{y} - v\|_p \leq O(k^\gamma) + \sqrt{k}\sqrt{r} \cdot O\Big(\frac{1}{r\sqrt{n}k^3}\Big) = O(k^\gamma)$ since $\|A\|_2 \leq O(\sqrt{r})$, and $\|V\|_2 \leq 1$.

Therefore, if for all $y_0 \in \calN$, there does not exist a corresponding $y_1$ for which $y := y_0 + y_1$ satisfies $\|AV^Ty - v\|_p$, then the same holds for all $y_0 \in T_{\beta, \gamma, p}$. This means that, with probability $1 - 2^{-\Theta(k)}$ over the entries of $v$, for all $y_0 \in T_{\beta, \gamma, p}$, there does not exist $y_1 \in Y_{\beta, \gamma, p}$ with disjoint support from $y_0$ such that $\|AV^Ty - v\|_p \leq O(k^\gamma)$, meaning for all $y \in B_{\gamma, p}$, $\|AV^Ty - v\|_p \geq \Omega(k^\gamma)$. We have shown that for any $j \in [t]$, the event $\calE_3(v_j)$ holds with probability $1 - 2^{-\Theta(k)}$, where $\calE_3(v_j)$ is the event that with probability $1 - 2^{-\Theta(k)}$, for all $y \in \R^n$, either $\|y\|_p \geq \Omega(k^\gamma)$ or $\|AV^Ty - v_j\|_p \geq \Omega(k^\gamma)$. 

Note that the $v_j$ are independent for different $j$. Hence, we can show that $\calE_3(v_j)$ holds simultaneously for $\Omega(t)$ indices $j \in [t]$ with the desired probability, as follows. Let $Z_j = 1$ if $\calE_3(v_j)$ does not hold. Then, the probability that at least $\frac{9t}{10}$ of the $Z_j$ are equal to $1$ is at most
$$\sum_{d = \frac{9t}{10}}^{t} \binom{t}{d} (2^{-\Theta(k)})^d \leq O(t) \cdot 2^t \cdot 2^{-\Theta(tk)} = 2^{-\Theta(tk)}$$
meaning that with probability $1 - 2^{-\Theta(tk)}$, $\calE_3(v_j)$ holds for at least $\frac{t}{10}$ indices $j \in [t]$.
\end{proof}

Now we prove the main result of this appendix:

\begin{theorem}[Lower Bound for $\ell_p$ Low Rank Approximation through Column Subset Selection - Similar to Theorem G.28 of \cite{l1_lower_bound_and_sqrtk_subset}]
For a sufficiently large $k \in \N$, and any constant $c \geq 1$, let $n = O(k(\log k)^c)$ and let $M \in \R^{(k + n) \times n}$ be a random matrix such that the top $k \times n$ submatrix has i.i.d. $N(0, 1)$ entries, and the bottom $n \times n$ submatrix is the identity matrix. Then, with probability $1 - 2^{-\Theta(k)}$, for any subset $S \subset [n]$ with $|S| \leq \frac{n}{2} =: r$,
$$\min_{X \in \R^{r \times n}} \|M_SX - M\|_p \geq \Omega(k^{\frac{1}{p} - \frac{1}{2} - \alpha}) \min_{M_k \text{ rank }k} \|M_k - M\|_p$$
where $\alpha \in (0, \frac{1}{p} - \frac{1}{2})$ can be arbitrary.
\end{theorem}
\begin{proof}

Throughout the proof, fix $\alpha \in (0, \frac{1}{p} - \frac{1}{2})$, $\gamma = \frac{1}{p} - \frac{1}{2} - \alpha$, and $\beta = \frac{1 + \alpha}{2}$. In addition, let $A \in \R^{k \times n}$ be the submatrix of $M$ consisting of its top $k$ rows. The proof of this theorem follows that of Theorem G.28 of \cite{l1_lower_bound_and_sqrtk_subset}: we apply Lemma \ref{lemma:paying_with_either_regression_cost_or_norm} with $A_S$ in the place of $A$, to show that each subset $S$ of size at most $r$ incurs a large cost, then perform a union bound over all such $S$. For convenience, we consider the $p^{th}$ power of the $\ell_p$ norm throughout the proof.

First, notice that $\min_{M_k \text{ rank }k} \|M_k - M\|_p^p \leq n$, since we could take $M_k$ to be the $(k + n) \times n$ matrix whose first $k$ rows are the same as those of $M$, and whose last $n$ rows are $0$. Hence, it suffices to show that for any subset $S \subset [n]$ of size at most $r$, $\min_{X \in \R^{r \times n}} \|M_SX - M\|_p^p \geq k^{p\gamma} n$. First, fix a subset $S \subset [n]$ with $|S| \leq r$. Observe that the cost of fitting a single column $M_l$ of $M$ using $M_S$ is
$$\text{cost}(S, l) := \min_{x_l \in \R^r} \Big( \|A_Sx_l - A_l\|_p^p + \|x_l\|_p^p - 1\Big)$$
(where the notation $\text{cost}(S, l)$ is from the proof of Theorem G.28 in \cite{l1_lower_bound_and_sqrtk_subset}). This is because the first $k$ entries of $M_l$ are given by $A_l$, and one of the last $n$ entries of $M_l$ is $1$ and the others are $0$. We can apply Lemma \ref{lemma:paying_with_either_regression_cost_or_norm} now, as follows. Suppose $\calEHat(A_S, I_r, \beta, \gamma, p)$ occurs. Then, since the columns of $A$ are independent, and $A_l$ is independent from $A_S$ for $l \not\in S$, by Lemma \ref{lemma:paying_with_either_regression_cost_or_norm}, with probability $1 - 2^{-\Theta(nk)}$, for at least $\Omega(n)$ of the indices $l \in [n] \setminus S$, either $\|A_Sx_l - A_l\|_p^p \geq \Omega(k^{\gamma p})$ or $\|x_l\|_p^p \geq \Omega(k^{\gamma p})$ (since while applying that lemma, we can take $n = r$ and the matrix $V \in \R^{n \times r}$ to be $I_r$, which has orthonormal columns --- note that we are taking $t = n$ in that lemma). Hence,
$$\min_{X \in \R^{r \times n}} \|M_SX - M\|_p^p = \sum_{l \in [n] \setminus S} \text{cost}(S, l) \geq \Omega(n \cdot k^{\gamma p}) \geq \Omega(k^{\gamma p} \cdot OPT)$$
Now, instead of conditioning on $\calEHat(A_S, I_r, \beta, \gamma, p)$ simultaneously for all $S$, we can simply condition on $\calEHat(A, I_n, \beta, \gamma, p)$:
\begin{claim}[Similar to Claim G.29 of \cite{l1_lower_bound_and_sqrtk_subset})]
\label{claim:condition_on_single_event}
$\calEHat(A, I_n, \beta, \gamma, p)$ implies $\calEHat(A_S, I_r, \beta, \gamma, p)$.
\end{claim}
\begin{proof}
In the proof of this claim, we use $Y_{\beta, \gamma, p, t}$ to denote the instance of $Y_{\beta, \gamma, p}$ in $\R^t$. Suppose $\calEHat(A, I_n, \beta, \gamma, p)$ occurs. Then, for all $y \in Y_{\beta, \gamma, p, n} \subset \R^n$, $Ay$ has at most $O(k/\log k)$ coordinates with absolute value at least $\Omega(1/\log k)$. In particular, for any subset $S \subset [n]$, let $y_S \in \R^n$ be a point whose support is contained in $S$, and let $\widetilde{y_S}$ be $y_S$, with all coordinates not indexed by $S$ removed. Observe that $A_SI_r\widetilde{y_S} = AI_ny_S$, and since $\calEHat(A, I_n, \beta, \gamma, p)$ holds, $AI_ny_S$ has at most $O(k/\log k)$ coordinates with absolute value at least $\Omega(1/\log k)$. This implies that for any $y \in Y_{\beta, \gamma, p, r}$, $A_SI_ry$ has at most $O(k/\log k)$ coordinates with absolute value at least $\Omega(1/\log k)$. Moreover, since $A_S$ is a subset of columns of $A$, the operator norm of $A_S$ is at most $\|A\|_2 \leq \sqrt{n} = O(\sqrt{r})$, since $r = \frac{n}{2}$.
\end{proof}

Hence, we can perform a union bound over all subsets $S \subset [n]$ with $|S| \leq r$, while conditioning on $\calEHat(A, I_n, \beta, \gamma, p)$. For convenience, let $\calETilde = \calEHat(A, I_n, \beta, \gamma, p)$. Then,
\begin{equation}
\begin{split}
& \ProbBig{\exists S \subset [n] ,\, |S| \leq r \text{ s.t. } \min_{X \in \R^{r \times n}} \|M_SX - M\|_p^p < O(k^{\gamma p} \cdot OPT)} \\
& = \ProbBig{\exists S \subset [n] ,\, |S| \leq r \text{ s.t. } \min_{X \in \R^{r \times n}} \|M_SX - M\|_p^p < O(k^{\gamma p} \cdot OPT) \Big| \calETilde} \Prob{\calETilde} \\
& + \ProbBig{\exists S \subset [n] ,\, |S| \leq r \text{ s.t. } \min_{X \in \R^{r \times n}} \|M_SX - M\|_p^p < O(k^{\gamma p} \cdot OPT) \Big| \neg\calETilde} \Prob{\neg\calETilde} \\
& \leq \ProbBig{\exists S \subset [n] ,\, |S| \leq r \text{ s.t. } \min_{X \in \R^{r \times n}} \|M_SX - M\|_p^p < O(k^{\gamma p} \cdot OPT) \Big| \calETilde} + \Prob{\neg\calETilde} \\
& \leq \sum_{S \subset [n], |S| \leq r} \ProbBig{\min_{X \in \R^{r \times n}} \|M_SX - M\|_p^p < O(k^{\gamma p}OPT) \mid \calETilde} + 2^{-\Theta(k)} \\
& \leq \sum_{S \subset [n], |S| \leq r} 2^{-\Theta(nk)} + 2^{-\Theta(k)} \\
& = \binom{n + 1}{r + 1} 2^{-\Theta(nk)} + 2^{-\Theta(k)} \\
& \leq 2^{O(r \log n)}2^{-\Theta(nk)} + 2^{-\Theta(k)} \\
& = 2^{O(n \log k)}2^{-\Theta(nk)} + 2^{-\Theta(k)} \\
& = 2^{-\Theta(k)}
\end{split}
\end{equation}
Here, the first inequality is because probabilities are at most $1$. The second is because $\calETilde$ occurs with probability at least $1 - 2^{-\Theta(k)}$, and by a union bound over $S \subset [n]$ with $|S| \leq r$. The third is because of our observation above, that $\widetilde{\calE}$ implies $\calEHat(A_S, I_r, \beta, \gamma, p)$, and if $A_S$ is such that $\calEHat(A_S, I_r, \beta, \gamma, p)$ holds, then the probability that $\min_{X \in \R^{r \times n}} \|M_SX - M\|_p^p < O(k^{\gamma p} \cdot OPT)$ is at most $2^{-\Theta(nk)}$. After that, the second equality is because the number of subsets of $[n]$ of size at most $r$ is $\binom{n + 1}{r + 1}$.

In summary, the probability that there is a subset achieving error less than $O(k^{\gamma p} \cdot OPT)$ is at most $2^{-\Theta(k)}$.
\end{proof}

\newpage
\section{$\poly(k)$-Approximation Algorithms with Bicriteria Rank Independent of $n$ and $d$} \label{appendix:poly_k_bicriteria_rank}

\subsection{A $\poly(k)$-Approximation with $\poly(k)$ Bicriteria Rank}

We can remove the $O(\log d)$-factor in the bicriteria rank of Algorithm \ref{algorithm:column_subset_sampling}, at the cost of an increase in the approximation factor to $O(k^{\frac{2}{p} - 1} \poly(\log k))$. This works as follows: note that Algorithm \ref{algorithm:column_subset_sampling} selects columns in $O(\log d)$ blocks, each having size $r$ (which is $O(k \log k)$ if $p = 1$ and $O(k \log k \log \log k)$ if $p \in (1, 2)$) --- out of these $O(\log d)$ blocks, we show that there exists a subset of blocks of size $r$ (hence giving a column subset of size $r^2$) which spans an $O(k^{\frac{2}{p} - 1} \poly(\log k))$-approximation to $A$. An advantage of this algorithm is that both the rank \textit{and} the approximation factor are polynomial in $k$, and do not depend on $n$ and $d$.

\begin{algorithm}
\caption{Obtaining a $\poly(k)$-approximation of rank $O(r^2)$, where $r = O(k \log k)$ if $p = 1$ and $O(k \log k \log \log k)$ if $p \in (1, 2)$. This algorithm simply finds the left factor, and the right factor can be obtained using linear programming.}
\label{algorithm:poly_k_rank_k_check_subsets_algorithm}
\begin{algorithmic}
\Require $A \in \R^{n \times d}$, $k \in \N$, $p \in [1, 2)$
\Ensure $U \in \R^{n \times O(r^2)}$
\Procedure{PolyKErrorAndRankApproximation}{$A, k, p$}
\State {$S_1, S_2, \ldots, S_{b} \gets \textsc{RandomColumnSubsetSelection}(A, k, p)$ (Note that the output $S$ is being discarded, since all we need are the blocks. Also note that the number of blocks $b$ is $O(\log d)$.)}
\State {$U \gets $ The $n \times O(r^2)$ zero matrix}
\State {$\textsc{MinError} \gets 0$}
\If {$b < \frac{C}{\eps} r$}
    \State {// Here $C$ is a sufficiently large absolute constant. In this case \textsc{RandomColumnSubsetSelection}}
    \State {// already gives $O(k^2 \log^2 k/\eps)$ columns.}
    \State {$S \gets \cup_{i = 1}^b S_i$}
    \State {$U \gets A_S$}
\Else
    \For {$I \subset [b], |I| = r$}
        \State {$T \gets \cup_{i \in I} S_i$}
        \State {$U_{temp} \gets A_T$}
        \State {$V_{temp} \gets \argmin_{V \in \R^{r^2 \times d}} \|U_{temp} V - A\|_p$}
        \If {$\|U_{temp}V_{temp} - A\|_p \leq \textsc{MinError}$}
            \State {$\textsc{MinError} \gets \|U_{temp}V_{temp} - A\|_p$}
            \State {$U \gets U_{temp}$}
        \EndIf
    \EndFor
\EndIf
\EndProcedure
\end{algorithmic}
\end{algorithm}

\begin{theorem} [Column Subset Selection - Removing Dependence on $\log d$] \label{thm:poly_k_bicriteria_rank}
Let $A \in \R^{n \times d}$, $p \in [1, 2)$, and $k \in \N$. In addition, let $r = O(k \log k)$ if $p = 1$ and $r = O(k \log k \log \log k)$ otherwise. Then, with constant probability, Algorithm \ref{algorithm:poly_k_rank_k_check_subsets_algorithm} returns $U \in \R^{n \times O(r^2/\eps)}$ and $V \in \R^{O(r^2/\eps) \times d}$ such that
$$\|UV - A\|_p \leq O(r^{\frac{2}{p} - 1} (\log k)^{\frac{1}{p}}) \min_{A_k \text{ rank }k} \|A_k - A\|_p$$
with running time $\nnz(A) + d^{1 + \eps}\poly(k \log(d))$, aside from the time required for Algorithm \ref{algorithm:column_subset_sampling}.

This running time can be achieved if the multiple-response $\ell_p$-regression is performed using a single sampling/rescaling matrix based on the $\ell_p$-Lewis weights of $A_T$, where $T$ is the subset of columns returned by Algorithm \ref{algorithm:column_subset_sampling}.
\end{theorem}
\begin{proof}
Let $T \subset [d]$ be the subset of columns returned by Algorithm \ref{algorithm:column_subset_sampling}. We will show that with constant probability, $T$ contains a subset of size $r^2$ which provides a left factor for $A$ giving a $\poly(k)$-approximation.

Note that $A_T$ is a left factor which already gives an $O(r^{\frac{1}{p} - \frac{1}{2}} (\log k)^{\frac{1}{p}})$-approximation for $A$. By the proof of Theorem \ref{thm:column_sampling_better_approximation_factor}, to obtain this approximation factor, it suffices to fit each column $A_j$ of $A$ using at most $r$ columns of $A_T$ --- i.e. there is $B_j \in \R^{|T|}$ having at most $r$ nonzero entries, such that we can use $A_TB_j$ as our approximation to $A_j$. This is because in the proof of Theorem \ref{thm:column_sampling_better_approximation_factor}, to bound the cost on each discarded column $A_j$, we only consider the $\ell_p$-regression error obtained by fitting $A_j$ using the block of $r$ columns that were sampled on the round where $A_j$ is discarded.

Let $B \in \R^{|T| \times d}$ such that the $j^{th}$ column of $B$ is $B_j$. Then, as a restatement of Theorem \ref{thm:column_sampling_better_approximation_factor}, $\|A_T B - A\|_p \leq O(r^{\frac{1}{p} - \frac{1}{2}} (\log k)^{\frac{1}{p}}) \min_{A_k \text{ rank }k} \|A - A_k\|_p$. 

Let $M = A_T B$ --- we will find a $\poly(k)$-approximation for $M$ and show that it is also a $\poly(k)$-approximation for $A$ by the triangle inequality. By Theorem \ref{theorem:subset_existence}, there exists a subset of columns of $M$, of size $r$, giving rise to a left factor $U_0 \in \R^{n \times r}$ such that
$$\min_{V \in \R^{r \times d}} \|U_0V - M\|_p \leq O(r^{\frac{1}{p} - \frac{1}{2}}) \min_{M_k \text{ rank }k} \|M_k - M\|_p$$
Let $V_0$ be the optimal right factor for $U_0$. Then, letting $M_k$ and $A_k$ be the optimal rank-$k$ approximations in the $\ell_p$-norm for $M$ and $A$ respectively,
\begin{equation}
\begin{split}
\|U_0V_0 - A\|_p
& \leq \|U_0V_0 - M\|_p + \|M - A\|_p \\
& \leq O(r^{\frac{1}{p} - \frac{1}{2}}) \|M_k - M\|_p + \|M - A\|_p \\
& \leq O(r^{\frac{1}{p} - \frac{1}{2}}) \|A_k - M\|_p + \|M - A\|_p \\
& \leq O(r^{\frac{1}{p} - \frac{1}{2}}) \Big(\|A_k - A\|_p + \|A - M\|_p \Big) + \|M - A\|_p \\
& \leq O(r^{\frac{1}{p} - \frac{1}{2}}) \|A - M\|_p + O(r^{\frac{1}{p} - \frac{1}{2}}) \|A - A_k\|_p \\
& \leq O(r^{\frac{1}{p} - \frac{1}{2}}) \cdot O(r^{\frac{1}{p} - \frac{1}{2}} (\log k)^{\frac{1}{p}}) \|A - A_k\|_p \\
& = O(r^{\frac{2}{p} - 1} (\log k)^{\frac{1}{p}}) \|A - A_k\|_p
\end{split}
\end{equation}
meaning $U_0$ is a left factor giving an $O(r^{\frac{2}{p} - 1}(\log k)^{\frac{1}{p}})$-approximation for $A$.

Finally, for each column $v$ of $U_0$, $v$ is a column of $M$, meaning it can be written as a linear combination of at most $r$ columns of $A_T$ --- these are the $r$ columns in the block which was sampled in the round when $v$ was approximately covered and discarded. Hence, $A_T$ has a subset of columns of size $r^2$, which is the union of $r$ blocks out of the $O(\log d)$ blocks which were sampled over the course of Algorithm \ref{algorithm:column_subset_sampling}, such that this subset spans an $O(r^{\frac{2}{p} - 1} (\log k)^{\frac{1}{p}})$ approximation for $A$. Since Algorithm \ref{algorithm:poly_k_rank_k_check_subsets_algorithm} checks all subsets of this size, it gives an $O(r^{\frac{2}{p} - 1} (\log k)^{\frac{1}{p}})$-approximation for $A$.

Let us now analyze the runtime of this algorithm, aside from the time required to perform \textsc{RandomColumnSubsetSelection}. If the number of blocks $b < \frac{C}{\eps}r$ ($C$ is a sufficiently large absolute constant mentioned in Algorithm \ref{algorithm:poly_k_rank_k_check_subsets_algorithm} and will be chosen appropriately later), then \textsc{RandomColumnSubsetSelection} already returns a subset of $O(r^2/\eps)$ columns.

Otherwise, $r \leq \frac{\eps b}{C}$, and Algorithm \ref{algorithm:poly_k_rank_k_check_subsets_algorithm} checks all subsets of $[b]$ of size $r$. In this case, the number of subsets checked is
$$\binom{b}{r} \leq \binom{b}{\frac{\eps b}{C}} \leq 2^{H_2(\frac{\eps}{C})b} = d^{O(1) \cdot H_2(\frac{\eps}{C})} \leq d^\eps$$
where the first inequality holds as long as $r \leq \frac{\eps b}{C} \leq \frac{b}{2}$, and in the second inequality $H_2$ is the binary entropy function. The first equality is because $b = O(\log d)$, meaning $2^b = d^{O(1)}$. Finally, the fourth inequality is because we can choose $C$ to be sufficiently large.

Hence, the number of subsets checked is at most $d^\eps$ if $C$ is chosen to be sufficiently large. For each subset, we perform multiple-response $\ell_p$-regression. We can speed up this $\ell_p$-regression as follows. Let $S \in \R^{\widetilde{O}(k \log d) \times n}$ be a sampling and rescaling matrix generated according to the Lewis weights of $A_T$ (recall that $T$ is the subset of columns returned by Algorithm \ref{algorithm:column_subset_sampling}). Recall that by Lemma \ref{lemma:lewis_weights_affine_contraction}, with probability $1 - O(1)$, for all $V \in \R^{|T| \times d}$,
\begin{equation} \label{eq:no_affine_contraction}
\begin{split}    
\|SA_TV - SA\|_p^p \geq \Omega(1) \|A_TV - A\|_p^p - O(1) \|A_TV^* - A\|_p^p
\end{split}
\end{equation}
where $V^* = \argmin_{V'} \|A_T V' - A\|_p$.

Now, out of all the subsets tried by Algorithm \ref{algorithm:poly_k_rank_k_check_subsets_algorithm}, let $U_{temp}^*$ be the optimal subset of blocks, and let $U_{temp}$ be an arbitrary subset of blocks that is tried. We make the following definitions: $V_{temp}^* = \argmin_V \|U_{temp}^* V - A\|_p$, $V_{temp, S}^* = \argmin_V \|SU_{temp}^*V - SA\|_p$, and $V_{temp, S} = \argmin_V \|SU_{temp} V - SA\|_p$. Then, since $U_{temp}$ and $U_{temp}^*$ are both given by subsets of columns of $A$, we can apply Equation \ref{eq:no_affine_contraction}. Applying it to $U_{temp}$ gives us
\begin{equation} \label{eq:block_subset_lower_bound}
\begin{split}
\|SU_{temp}V_{temp, S} - SA\|_p^p
& \geq \Omega(1) \|U_{temp}V_{temp, S} - A\|_p^p - O(1) \|A_TV^* - A\|_p^p
\end{split}
\end{equation}

Now, let $U_{temp}^{**}$ be the subset of blocks tried by Algorithm \ref{algorithm:poly_k_rank_k_check_subsets_algorithm} which minimizes the sketched error, i.e. it is the subset $U_{temp}$ minimizing $\|SU_{temp}V_{temp, S} - SA\|_p$. Define $V_{temp, S}^{**} = \argmin_V \|SU_{temp}^{**}V - SA\|_p$ Then,
\begin{equation} \label{eq:best_subset_lower_bound}
\begin{split}
\|SU_{temp}^*V_{temp}^* - SA\|_p^p
& \geq \|SU_{temp}^*V_{temp, S}^* - SA\|_p^p \\
& \geq \|SU_{temp}^{**}V_{temp, S}^{**} - SA\|_p^p \\
& \geq \Omega(1) \|U_{temp}^{**}V_{temp, S}^{**} - A\|_p^p - O(1) \|A_TV^* - A\|_p^p
\end{split}
\end{equation}
where the first inequality is because $V_{temp, S}^*$ is the minimizer of the sketched error for the left factor $U_{temp}^*$, the second is because $U_{temp}^{**}$ is the subset of blocks minimizing the sketched error, and the third is by Equation \ref{eq:block_subset_lower_bound}.

Finally, by Lemma \ref{lemma:lewis_weights_no_dilation_contraction}, with probability $1 - O(1)$,
$$\|SU_{temp}^*V_{temp}^* - SA\|_p^p \leq O(1)\|U_{temp}^*V_{temp}^* - A\|_p^p$$
with probability $1 - O(1)$.

Putting this together, if $U_{temp}^{**}$ is the subset of blocks minimizing the sketched $\ell_p$ regression error, and $V_{temp, S}^{**}$ is its corresponding right factor (that minimizes the sketched $\ell_p$ regression error), then with probability $1 - O(1)$,
\begin{equation}
\begin{split}
\|U_{temp}^{**}V_{temp, S}^{**} - A\|_p^p
& \leq O(1)\|U_{temp}^*V_{temp}^* - A\|_p^p + O(1) \|A_TV^* - A\|_p^p \\
& \leq \Big(O(r^{\frac{2}{p} - 1}(\log k)^{\frac{1}{p}})\Big)^p \min_{A_k \text{ rank } k} \|A - A_k\|_p^p
\end{split}
\end{equation}
where the second inequality is because $U_{temp}^*$ is an $O(r^{\frac{2}{p} - 1}(\log k)^{\frac{1}{p}})$-approximate column subset as shown above, while $A_T$ provides an $O(r^{\frac{1}{p} - \frac{1}{2}}(\log k)^{\frac{1}{p}})$-approximation by our analysis of Algorithm \ref{algorithm:column_subset_sampling}.

Hence, by taking $p^{th}$ powers, we find that performing multiple-response $\ell_p$ regression for each of the subsets of blocks, while reusing a single sampling and rescaling matrix $S$ generated using the $\ell_p$ Lewis weights of $A_T$, only worsens the approximation guarantees by an $O(1)$ factor.

Now we finish analyzing the running time when Algorithm \ref{algorithm:poly_k_rank_k_check_subsets_algorithm} is implemented using $\ell_p$ Lewis weights. Computing the $\ell_p$ Lewis weights of $A_T$, and generating the sampling matrix $S$, takes at most $\nnz(A) + \poly(k)$ time, by Lemma \ref{lemma:sampling_by_lewis_weights}, and performing the multiplication $SA$ takes time at most $\nnz(A)$. Now, for each of the $d^\eps$ subsets we try, we perform multiple-response $\ell_p$ regression on the $\widetilde{O}(k \log d) \times d$ matrix $SA$, fitting the columns of $SA$ using a $\widetilde{O}(k \log d) \times r^2$ matrix $SU_{temp}$. Each of the $d$ $\ell_p$-regression steps takes time $\poly(k \log d)$. Therefore, each multiple-response $\ell_p$ regression step takes $d \poly(k \log d)$ time, and overall, trying all subsets takes $d^{1 + \eps}\poly(k \log d)$ time. This completes the proof.
\end{proof}

\begin{remark}
Note that in \cite{algorithm3_original}, it was shown that every $A \in \R^{n \times d}$ has a subset of $k$ columns spanning an $O(k)$-approximation to the optimum. Using this in our analysis, instead of Theorem \ref{theorem:subset_existence} allows us to get an $O(r^{\frac{1}{p} - \frac{1}{2}} \cdot k \cdot (\log k)^{\frac{1}{p}})$-approximation instead (in particular, for $p = 1$, this is an $\widetilde{O}(k^{\frac{3}{2}})$-approximation) while giving us a rank $O(rk)$-solution, in particular removing one $\log(k)$ factor from the rank in the $p = 1$ case.
\end{remark}

\begin{remark}
In the analysis of our algorithm with rank at most $k$, we will take $\eps$ from Theorem \ref{thm:poly_k_bicriteria_rank} to just be a sufficiently small constant to optimize the running time, but to also keep the rank at most $O(k^2 \log^2 k)$.
\end{remark}

\begin{remark}
Note that this kind of analysis is not applicable to any bi-criteria column subset selection algorithm, and it uses special properties of Algorithm \ref{algorithm:column_subset_sampling} --- given an arbitrary bi-criteria column subset selection algorithm, we cannot use this analysis to reduce the rank. The key property is that for each column that is discarded by Algorithm \ref{algorithm:column_subset_sampling}, it can be fit using at most $r$ columns belonging to the left factor that Algorithm \ref{algorithm:column_subset_sampling} returns. Note that Algorithm 1 of \cite{DBLP:conf/nips/SongWZ19} also has this property, and hence it can also be used to obtain a $\poly(k)$-approximation algorithm with $\poly(k)$ bicriteria rank, albeit with a somewhat larger approximation factor.
\end{remark}

\subsection{Reducing the Rank to At Most $k$}

We can combine our algorithm from the previous subsection with Theorem C.19 of \cite{l1_lower_bound_and_sqrtk_subset}:

\begin{theorem}[Theorem C.19 of \cite{l1_lower_bound_and_sqrtk_subset} --- Reducing the Rank]
\label{theorem:decreasing_bicriteria_rank_earlier}
Given a rank $r$ matrix $B = U_BV_B \in \R^{n \times d}$, where $U_B \in \R^{n \times r}$ and $V_B \in \R^{r \times d}$, $p \in [1, 2)$, and $k \in [r]$ there exists an algorithm which takes $(n + d) \cdot \poly(k)$ time to output two matrices $U \in \R^{n \times k}$ and $V \in \R^{k \times d}$ such that
$$\|UV - B\|_p \leq \poly(r) \min_{B_k \text{ rank }k}\|B_k - B\|_p$$
\end{theorem}
\begin{proof}
In the case $p = 1$, this is Theorem C.19 of \cite{l1_lower_bound_and_sqrtk_subset}. Theorem C.19 of \cite{l1_lower_bound_and_sqrtk_subset} also generalizes to $p \in (1, 2)$. This is because all of the lemmas used in the proof of this theorem can be extended to the $\ell_p$ norm case. Claim C.18 of \cite{l1_lower_bound_and_sqrtk_subset}, which is that multiplying any vector in a low-dimensional subspace by a dense Cauchy matrix leads to at most a constant-factor decrease in its $\ell_1$-norm, can be extended to the $\ell_p$-norm by considering matrices with $p$-stable entries instead. Similarly, Lemma D.23 of \cite{l1_lower_bound_and_sqrtk_subset} (multiplying a fixed matrix by a dense Cauchy matrix does not increase its $\ell_1$-norm significantly) can be generalized to the $\ell_p$-norm case by using matrices with $p$-stable entries. Claims B.10 and B.11 of \cite{l1_lower_bound_and_sqrtk_subset} already apply to the $\ell_p$-norm case, and Lemmas D.8 and D.10 of \cite{l1_lower_bound_and_sqrtk_subset} are generalized to the $\ell_p$-norm case in Lemma E.8 and E.10 respectively.
\end{proof}

Denote the algorithm described in Theorem C.19 of \cite{l1_lower_bound_and_sqrtk_subset} by \textsc{RemoveBicriteriaRank}$(U_B, V_B, k, p)$, where $B = U_BV_B$ is a matrix of rank at most $r$ and $k$ is the desired rank. Then, we obtain a $\poly(k)$-approximation algorithm, shown in Algorithm \ref{algorithm:poly_k_rank_exactly_k}:

\begin{algorithm}
\caption{Reducing the bicriteria rank of Algorithm \ref{algorithm:poly_k_rank_k_check_subsets_algorithm}.}
\label{algorithm:poly_k_rank_exactly_k}
\begin{algorithmic}
\Require $A \in \R^{n \times d}$, $k \in \N$, $p \in [1, 2)$
\Ensure $W \in \R^{n \times k}$, $Z \in \R^{k \times d}$
\Procedure{PolyKErrorNotBiCriteriaApproximation}{$A, k, p$}
\State {$U \gets \textsc{PolyKErrorAndRankApproximation}(A, k, p)$}
\State {$V \gets \argmin_{V'} \|UV' - A\|_p$}
\State {$W, Z \gets \textsc{RemoveBicriteriaRank}(U, V, k, p)$}
\EndProcedure
\end{algorithmic}
\end{algorithm}

\begin{theorem}[Analysis of Algorithm \ref{algorithm:poly_k_rank_exactly_k}]
\label{theorem:analysis_of_rank_exactly_k_algorithm}
Let $A \in \R^{n \times d}$ and $k \in \N$. Then, with constant probability, Algorithm \ref{algorithm:poly_k_rank_exactly_k} returns $W \in \R^{n \times k}$ and $Z \in \R^{k \times d}$ such that
$$\|WZ - A\|_p \leq \poly(k) \min_{A_k \text{ rank }k} \|A_k - A\|_p$$
The running time of Algorithm \ref{algorithm:poly_k_rank_exactly_k} is $\nnz(A) + (d^{1.1} + n)\poly(k\log d)$, aside from the running time of Algorithm \ref{algorithm:column_subset_sampling}.
\end{theorem}
\begin{proof}
For a matrix $M \in \R^{n \times d}$ and $k \in \N$, let $OPT_{M, k} := \min_{M_k \text{ rank }k} \|M_k - M\|_p$. Now, by the guarantees on Algorithm \ref{algorithm:poly_k_rank_k_check_subsets_algorithm}, $\|UV - A\|_p \leq \poly(k) OPT_{A, k}$. Let $r = O(k^2 \log^2 k)$ be the bicriteria rank of Algorithm \ref{algorithm:poly_k_rank_k_check_subsets_algorithm}. Then, by Theorem \ref{theorem:decreasing_bicriteria_rank_earlier}, $\|WZ - UV\|_p \leq \poly(r) OPT_{UV, k} = \poly(k) OPT_{UV, k}$. Therefore,
$$\|WZ - A\|_p \leq \|WZ - UV\|_p + \|UV - A\|_p \leq \poly(k) OPT_{UV, k} + \poly(k) OPT_{A, k}$$
and it suffices to show that $OPT_{UV, k} \leq \poly(k) OPT_{A, k}$. However, if $A_k = \argmin_{A' \text{ rank }k} \|A' - A\|_p$, then
$$OPT_{UV, k} \leq \|A_k - UV\|_p \leq \|A_k - A\|_p + \|UV - A\|_p \leq OPT_{A, k} + \poly(k) OPT_{A, k} \leq \poly(k) OPT_{A, k}$$
Here, the first inequality is because $A_k$ has rank at most $k$. The second inequality is by the triangle inequality, and the third is because of the approximation guarantees on Algorithm \ref{algorithm:poly_k_rank_k_check_subsets_algorithm}. In summary, $OPT_{UV, k} \leq \poly(k) OPT_{A, k}$, meaning $\|WZ - A\|_p \leq \poly(k) OPT_{A, k}$.

Now, let us analyze the running time of Algorithm \ref{algorithm:poly_k_rank_exactly_k}. Taking $\eps = 0.1$ in Algorithm \ref{algorithm:poly_k_rank_k_check_subsets_algorithm}, we find that the running time needed to obtain $U$ is $\nnz(A) + d^{1.1} \poly(k\log d)$, aside from the time taken to run Algorithm \ref{algorithm:column_subset_sampling}. We can also find $V$ in $\nnz(A) + d \cdot \poly(k)$ time. To see this, note that we can multiply $U$ and $A$ by a sampling matrix $S$ corresponding to the Lewis weights of $U$ --- by Lemma \ref{lemma:sampling_by_lewis_weights}, $S$ can be generated in $\nnz(U) + \poly(k)$ time, and computing $SU$ and $SA$ takes at most $\nnz(A)$ time. By Lemma \ref{lemma:lewis_weights_affine_contraction}, performing $\ell_p$ regression using $SU$ and $SA$ instead of $U$ and $A$ leads to an $O(1)$-approximate right factor $V$. Each $\ell_p$ regression takes $\poly(k)$ time, meaning finding $V$ takes $d \cdot \poly(k)$ time. Finally, by Theorem C.19 of \cite{l1_lower_bound_and_sqrtk_subset}, finding $W$ and $Z$ takes $(n + d) \cdot \poly(k)$ time. The overall running time is therefore $\nnz(A) + (d^{1.1} + n)\poly(k)$, aside from the time taken to run Algorithm \ref{algorithm:column_subset_sampling}.
\end{proof}

\begin{remark}
Note that the bicriteria rank $r$ has to be $\poly(k)$, rather than $\poly(k, \log d)$, before applying the algorithm described in Theorem C.19 of \cite{l1_lower_bound_and_sqrtk_subset}, since the approximation factor is $\poly(r)$.
\end{remark}

\begin{remark}
It is also possible to combine our Algorithm \ref{algorithm:poly_k_rank_k_check_subsets_algorithm} with Algorithm 4 of \cite{algorithm3_original} to reduce the bicriteria rank to at most $k$, and obtain a $\poly(k)$-approximate matrix, although this requires $2^{O(k \log k)} + \poly(nd)$ running time. Algorithm 4 of \cite{algorithm3_original} as originally presented instead obtains a $\poly(k, \log d)$-approximate matrix, since the initialization for Algorithm 4 in \cite{algorithm3_original} has bicriteria rank $O(k \log d)$, and the approximation factor of Algorithm 4 of \cite{algorithm3_original} is proportional to the bicriteria rank of the initialization (due to the use of well-conditioned bases \cite{ddhkm_09_well_conditioned_basis}). This is further discussed in Appendix E of \cite{dwzzr19}. We could remove the $\log d$, in the approximation factor of Algorithm 4 of \cite{algorithm3_original}, by using our Algorithm \ref{algorithm:poly_k_rank_k_check_subsets_algorithm} as the initialization instead.
\end{remark}

\newpage
\section*{Acknowledgements}

We are grateful to Peilin Zhong for suggesting the proof of Lemma \ref{lemma:paying_with_either_regression_cost_or_norm}, which is part of the $\ell_p$ column subset selection lower bound. A. Mahankali was supported by the SURF award from CMU's Undergraduate Research Office. D. Woodruff was supported by the Office of Naval Research (ONR) grant N00014-18-1-256.



\end{document}